\documentclass[10pt,onecolumn,amsmath,amssymb,nofootinbib,superscriptaddress]{revtex4-2}

\usepackage[utf8]{inputenc}
\usepackage[margin=1in]{geometry}
\usepackage{amsmath, amsthm, amssymb,amscd, mathrsfs, amsfonts, mathtools,pgfplots}
\usepackage{appendix}
\usepackage{graphicx}
\usepackage{relsize}
\usepackage{mathtools}
\usepackage{dsfont}
\usepackage{float}
\usepackage{soul}
\usepackage{quantikz}
\usepackage{afterpage}

\usetikzlibrary{shapes, backgrounds}
\tikzset{
  qaoaBlock/.style={
    draw,
    rounded corners=2pt,
    minimum width=1.4cm,
    minimum height=0.8cm,
    align=center,
    font=\small,
  },
  problem/.style={qaoaBlock, fill=green!30},
  mixer/.style={qaoaBlock, fill=blue!30},
}
\usetikzlibrary{decorations.pathmorphing}
\usetikzlibrary{arrows.meta,calc,positioning}
\usetikzlibrary{patterns}
\definecolor{ink}{HTML}{152238}
\definecolor{muted}{HTML}{5D6A7D}
\definecolor{paperbg}{HTML}{F4F7FA}
\definecolor{linegray}{HTML}{D8E0E8}
\definecolor{blue}{HTML}{3478B8}
\definecolor{teal}{HTML}{169C94}
\definecolor{coral}{HTML}{DE735B}
\definecolor{violet}{HTML}{7656A5}
\definecolor{gold}{HTML}{C78A20}

\usepackage{chngcntr}
\usepackage[bookmarks=true,
bookmarksnumbered=true,
breaklinks=true,
pdfstartview=FitH,
hyperfigures=false,
plainpages=false,
naturalnames=true,
colorlinks=true,
linkcolor=blue,       
citecolor=blue,       
urlcolor=blue,        
pagebackref=true,
pdfpagelabels]{hyperref}
\usepackage{ORCIDinREVTeX}
\theoremstyle{definition}
\newtheorem{thm}{Theorem}[section]
\newtheorem{prop}[thm]{Proposition}

\newtheorem{defn}[thm]{Definition}
\newtheorem{cor}[thm]{Corollary}
\newtheorem{rmk}[thm]{Remark}

\newtheorem{ex}[thm]{Example}

\allowdisplaybreaks

\usepackage[draft, commentmarkup=todo,
  todonotes={textsize=tiny, textwidth=0.83in}]{changes}

\definechangesauthor[name={Boris}, color=red]{new}
\definechangesauthor[name={Ilya Safro}, color=yellow]{IS}

\definechangesauthor[name={Boris}, color=blue]{BT}

\definechangesauthor[name={Bao Bach}, color=teal]{BB}

\usepackage{bbm}

\def\a{\alpha}

\def\s{\sigma}

\hypersetup{
	colorlinks = true,
	urlcolor = blue,
	linkcolor = blue,
	citecolor = red,
	pdfpagemode = UseNone
}

\usepackage{youngtab}
\usepackage{ytableau}

\begin{document}
\title{The Expressive Power of Constrained QAOA: What You Might Have MISsed}
\author{Boris Tsvelikhovskiy} 
\orcid{0000-0003-0798-7218}
\email{tsvelibor@gmail.com}
\affiliation{Department of Mathematics, University of California, Riverside, CA 92521, USA} 

\author{Bao Bach} 
\orcid{0000-0001-6210-7725}
\affiliation{Department of Computer and Information Sciences, University of Delaware, Newark, DE 19716, USA} 

\author{Ilya Safro} 
\orcid{0000-0001-6284-7408}
\affiliation{Department of Computer and Information Sciences, University of Delaware, Newark, DE 19716, USA} 
\affiliation{Department of Physics and Astronomy, University of Delaware, Newark, DE 19716, USA}

\begin{abstract}
Feasibility-preserving mixer Hamiltonians offer an attractive alternative to penalty encodings for constrained quantum optimization, yet how enforcing constraints alters the state space reachable by the Quantum Approximate Optimization Algorithm (QAOA) remains largely unexplored. We address this question within the context of the Maximum Independent Set (MIS) problem. Alongside the standard controlled bit-flip mixer, we analyze the alternative \emph{Flip-or-Stay mixer}, whose restriction to the feasible subspace $W_{\mathcal{F}}$ corresponds precisely to the shifted negative Laplacian of the independent-set reconfiguration graph. Although both mixers induce identical transitions between independent sets, the diagonal modification changes the spectral structure and can substantially enlarge the set of reachable quantum states.

For every connected input graph with at least two vertices, we prove that independently parameterized local controls associated with either mixer generate the full unitary Lie algebra on the feasible subspace. Standard QAOA exhibits a more nuanced algebraic structure. We prove that the dynamical Lie algebra associated with MIS-QAOA employing the conventional mixer embeds into its Flip-or-Stay counterpart. Furthermore, we construct an infinite family of graphs for which the corresponding vacuum-state dynamical group orbits exhibit strictly distinct dimensions, establishing a qualitative expressivity separation between the two architectures. For either mixer, we show that standard QAOA initialized in the empty set can prepare a state supported entirely on maximum independent sets at finite depth. For the free architectures, we also derive an exact loss-variance formula in the unitary-design limit, expressed through the number of feasible independent sets and the variance of their cardinalities.
\end{abstract}

\maketitle

\section{Introduction}
\label{sec:introduction}

Many important optimization problems are difficult not only because the number of candidate solutions is large, but because a substantial fraction of candidates may be infeasible. Scheduling, allocation, packing, routing, and graph optimization typically come with hard constraints that must be satisfied exactly. For quantum optimization algorithms, this creates a fundamental design choice. One can allow the quantum state to explore invalid configurations and penalize them through the objective Hamiltonian \cite{perlin2024q,farhi2020quantum}, or one can engineer the dynamics so that the computation never leaves the feasible set \cite{hen2016quantum,hen2016driver,HWORVB}. The latter approach is attractive because the quantum evolution is devoted entirely to valid solutions, but it raises an equally fundamental concern: restricting the dynamics may also restrict what the algorithm is capable of reaching.

Can hard feasibility constraints be enforced throughout a variational quantum algorithm without sacrificing its ability to explore the feasible solution space and reach an optimum? This question extends beyond the performance of a specific operator in variational quantum algorithms (e.g., the Quantum Approximate Optimization Algorithm) or a single optimization problem. It concerns the broader relationship  between constraints, expressivity, controllability, and trainability in variational quantum algorithms. A feasibility-preserving ansatz can be useful only if the restricted dynamics remains sufficiently rich; yet increasing expressivity can in turn produce highly concentrated optimization landscapes especially at large depth. Understanding these competing effects is therefore important for the design of quantum algorithms for constrained optimization.

Here we study this question using the Maximum Independent Set (MIS) problem, one of the canonical constrained graph optimization problems. \emph{Our results show that imposing feasibility need not compromise the expressive power of quantum approximate optimization algorithm (QAOA)}.  When the local controls of constrained QAOA are allowed to vary independently, the resulting dynamics is fully state controllable within the feasible subspace: any feasible quantum state can, in principle, be transformed exactly into any other. Thus, the hard constraint removes invalid configurations from the dynamics without introducing an additional fundamental obstruction to controllability.

\emph{Our contribution further shows that this strong controllability result is not merely an asymptotic statement about a highly flexible ansatz.} For standard QAOA, in which each layer is controlled by only one mixer parameter and one problem parameter, we prove that an optimal MIS state can be reached exactly at a finite circuit depth starting from the empty independent set, a state that is trivial to prepare. The result is existential: it does not imply a polynomial depth, nor does it provide an efficient method for finding the required parameters. Nevertheless, it establishes that the constrained standard QAOA dynamics is not fundamentally prevented from reaching an exact optimum even when its full dynamical Lie algebra is smaller than the full unitary algebra of the feasible space.

A second conclusion is that preserving the same set of feasible moves does not imply equivalent quantum dynamics. We study a modified feasibility-preserving mixer, which we call the Flip-or-Stay mixer, and compare it with the conventional Flip-or-Void construction. Both mixers connect exactly the same pairs of feasible configurations. Nevertheless, they can generate markedly different reachable dynamics. For an explicit infinite family of graphs, the conventional mixer confines the evolution from the empty independent set to a fixed four-dimensional subspace, whereas the proposed mixer accesses a subspace whose dimension grows with the graph family and achieves complete controllability within that subspace. We identify the origin of this separation: the dynamical Lie algebra underlying QAOA with the Flip-or-Stay mixer contains an additional diagonal operator that records the degree of each configuration in the feasible reconfiguration graph, thereby resolving dynamical degeneracies that persist under the standard adjacency mixer.

\emph{This observation suggests a broader principle for constrained quantum optimization. The set of allowed transitions between feasible configurations does not by itself determine the expressive power of a quantum ansatz.} The diagonal structure associated with those configurations can be equally important because it changes how the allowed transitions interfere and how dynamical degeneracies are resolved. In this sense, the geometry of the feasible solution space and the Hamiltonian used to explore that geometry should be viewed as separate ingredients of algorithm design.

The Maximum Independent Set problem provides a particularly transparent setting in which to make this connection. Given a graph, an independent set is a collection of vertices containing no adjacent pair, and MIS asks for an independent set of maximum cardinality. The problem is NP-hard and arises naturally in scheduling, resource allocation, network design, and related applications. Classical exact algorithms have achieved progressively improved exponential running times \cite{XN}, while approximation algorithms and rigorous bounds are known for structured sparse graph families, including large-girth $3$- and $4$-regular graphs \cite{Csoka}. Quantum approaches include the original QAOA formulation \cite{QAOA}, problem-specific constrained-mixer constructions \cite{HHR, HWORVB, Saleem}, Rydberg-atom implementations of MIS optimization \cite{EKC}, and more general quantum Hamiltonian approaches \cite{ZGYYWW}. Rigorous finite-depth analyses of QAOA on structured graph families have also continued to develop, for example for high-girth $3$-regular graphs \cite{FGRV}. These works address complementary
questions concerning approximation quality, physical implementation, and
finite-depth performance. Our focus is instead on the operator-algebraic
structure of feasibility-preserving dynamics.

\subsection{Technical overview}

For constrained optimization, a central alternative to penalty-based formulations is to restrict the quantum dynamics to valid configurations from the outset \cite{hen2016quantum, hen2016driver, HWORVB}. After encoding vertex subsets as computational basis states, let $\mathcal F$ denote the set of independent-set bit strings and let
\[
W_{\mathcal F}=\operatorname{span}\{\ket{x}:x\in\mathcal F\}
\]
denote the corresponding feasible Hilbert space. The Quantum Alternating Operator Ansatz of Hadfield et al.~\cite{HWORVB} generalizes QAOA precisely in this direction: the initial state belongs to $W_{\mathcal F}$ and the mixer is constructed so that the dynamics remains inside this subspace. This avoids spending amplitude on infeasible configurations and eliminates the need for large penalty terms, which can themselves modify the spectral and optimization landscape \cite{perlin2024q, farhi2020quantum}.

In the original QAOA treatment of Maximum Independent Set, the mixer was defined on the feasible subspace by coupling two computational basis states whenever their corresponding independent sets differ by the addition or removal of a single vertex \cite{QAOA}. Subsequent work gave explicit realizations using controlled bit flips \cite{HHR, HWORVB}. A standard choice of initial state in these formulations is the computational basis state $\ket{0}^n$ representing the empty independent set. Although this state is feasible, it is generally not an eigenstate of the constrained mixer. Consequently, convergence arguments of the type used in \cite{BKZS}, which require initialization in an extremal eigenstate of the mixer, do not apply directly.

Building on this feasible-space adjacency viewpoint, we formulate the constrained dynamics in terms of a classical graph. Let
$\mathcal R(\Gamma)$ be the independent-set reconfiguration graph: its
vertices are the independent sets of $\Gamma$, and two independent sets are adjacent when one can be obtained from the other by adding or removing a single vertex. 
When restricted to the feasible subspace $W_{\mathcal{F}}$, the standard controlled bit-flip mixer equals the adjacency operator $A_{\mathcal{R}(\Gamma)}$ of the independent set reconfiguration graph $\mathcal{R}(\Gamma)$, thereby identifying the mixer dynamics as a continuous-time quantum walk on $\mathcal{R}(\Gamma)$ (see \cite[Section~VII]{QAOA} for details). At the same time, it explains the initialization problem: because $\mathcal R(\Gamma)$ is generally irregular, its uniform state need not be an eigenvector of $A_{\mathcal R(\Gamma)}$.

Motivated by this interpretation and by single-site update
dynamics of Glauber type \cite{glauber1963time}, we consider
a modified local rule: when a vertex can be flipped without
violating feasibility, the local operator flips it; otherwise,
it leaves the configuration unchanged. We refer to the resulting operator as the \emph{Flip-or-Stay mixer}, denoted $\widehat{H}_{\mathrm{CX}}$. As established in Proposition~\ref{prop:MixersAsReconfigurationOperators},  the restriction of this mixer to the feasible subspace $W_{\mathcal{F}}$ admits an explicit operator representation in terms of the shifted negative Laplacian of the feasible reconfiguration graph:
\[
    \widehat H_{CX}\big|_{W_{\mathcal F}}
    =nI-L_{\mathcal R(\Gamma)}.
\]
Since every independent set can be reduced to the empty set by removing its occupied vertices, $\mathcal R(\Gamma)$ is connected. It follows immediately that the uniform feasible state
\begin{equation}
\label{eq:uniform_feasible_state}
    \ket{\xi_{\mathcal F}}=\frac{1}{\sqrt{|\mathcal F|}}\sum_{x\in\mathcal F}\ket{x}
\end{equation}
    
is the unique largest-eigenvalue state of $\widehat H_{CX}$, or equivalently the unique ground state of $-\widehat H_{CX}$. Moreover, the mixer gap is equal to the algebraic connectivity of the reconfiguration graph; see Remark~\ref{rmk:ReconfigurationSpectralGap}. The restriction of \(\widehat H_{CX}\) to the feasible space \(W_{\mathcal F}\) is nonnegative and irreducible in the computational basis. Its Perron eigenvector is \(\ket{\xi_{\mathcal F}}\), so the convergence theorem of \cite{BKZS}  applies with this initialization. 

We next investigate which states in the feasible Hilbert space can be
reached by circuits generated by the mixer and problem Hamiltonians. We
study this question through the corresponding dynamical Lie algebra, whose associated connected Lie group determines the continuous family of unitary transformations generated by these controls. Two parameter-sharing architectures are considered. In the standard QAOA ans\"atze, the local terms comprising each Hamiltonian share a common parameter within every layer. In the free, or multi-angle, ans\"atze \cite{herrman2022multi}, the corresponding local terms can instead be controlled independently. 

Our principal result establishes that, for every connected graph $\Gamma$ with at least two vertices, the free dynamical Lie algebras generated by the local $Z$ controls together with, respectively, the Flip-or-Void and Flip-or-Stay local mixer terms both restrict to the full unitary Lie algebra on the feasible subspace; see Theorem~\ref{thm:ExtendedFreeLieAlgFullUnitary}. Hence, the free ans\"atze achieve pure-state controllability on the feasible subspace: any normalized feasible state can be transformed into any other by a corresponding multi-angle QAOA circuit of finite depth. Thus, although the hard constraint restricts the accessible Hilbert space, it does not impose an additional loss of state controllability when the local parameters can be varied independently. Full controllability also holds with only one mixer parameter
and one cost parameter per layer when the vertex weights are positive and pairwise distinct (see Proposition~\ref{prop:WeightedMISStandardDLA}).

For the standard QAOA ans\"atze, the underlying algebraic structure is more subtle. We prove that the dynamical Lie algebra of the Flip-or-Void mixer is contained in that of the Flip-or-Stay mixer. DLA dimension calculations for all connected asymmetric graphs on six and seven vertices show that this containment is frequently strict, indicating that the Laplacian mixer reaches full controllability on $W_{\mathcal{F}}$ more broadly (see Section \ref{subsec:numeric_DLA_dim}).

To complement this numerical evidence, in
Section~\ref{subsec:explicit_mis_separation} we construct an explicit
infinite family of graphs \(\{\Gamma_r\}\) for which the two standard
ans\"atze have sharply different reachable geometries. For the
Flip-or-Void mixer, the orbit of \(\ket{0}^n\) is confined to a
four-dimensional complex cyclic subspace and therefore has real
dimension at most \(7\), independently of \(r\). In contrast, the Flip-or-Stay DLA restricts to the full unitary Lie algebra on the vacuum-generated cyclic subspace of the ambient Hilbert space, with the corresponding reachable state orbit achieving a real dimension of $3r-2$. Consequently, incorporating the diagonal ``stay'' term induces an unbounded separation between the reachable state sets of two collective mixers sharing identical feasible transitions, cost Hamiltonian, and initial computational state.

Importantly, we show that algorithmically relevant state preparation
does not require full Lie-algebraic controllability. In
Theorem~\ref{thm:qaoa_perron_frobenius}, we prove that standard QAOA
initialized in the computational vacuum state $\ket{0}^n$, with either
mixer, prepares a ground state of the problem Hamiltonian \emph{exactly} at a
\emph{finite circuit depth}. As noted in
Remark~\ref{rmk:vacuum_init_practicality}, this initialization avoids
the potentially difficult task of determining and preparing the
corresponding mixer ground state. Section~VII of the original QAOA
paper already presented the underlying asymptotic intuition, using an
adiabatic interpolation and Trotterization to argue convergence to the
optimal objective value as \(p\to\infty\) \cite{QAOA}. Our result
strengthens this argument by characterizing the relevant dynamical
orbits and promoting asymptotic convergence to exact reachability at a
finite, though not explicitly bounded, circuit depth. To the best of
our knowledge, such a finite-depth exact preparation guarantee has not
previously been established for constrained MIS-QAOA
\cite{QAOA,HHR,HWORVB,Saleem}.


It is important to note that full controllability is a statement of
expressivity, not by itself an efficiency or trainability guarantee. When sufficiently deep circuit ensembles approach unitary $2$-designs on the relevant subspace, concentration of measure can arise in their loss landscapes. To quantify this landscape geometry, we leverage the
Lie-algebraic theory of deep parameterized quantum circuits
\cite{RBSKMLC} to analyze the loss variance over the dynamical Lie group $G$; see Theorem~\ref{thm:BP_mitigation}.

Rather than being governed directly by the full ambient Hilbert-space
dimension, our universal lower bound is of order \(|\mathcal F|^{-2}\), which is inverse polynomial when \(|\mathcal F|\) is polynomial in $n$ (see Theorem \ref{thm:BP_mitigation}). Thus, the relevant dimensional scale in this bound is the cardinality of the feasible set rather than the ambient dimension $2^n$. In regimes where the constraints substantially reduce the accessible state space, this gives a correspondingly stronger guarantee against loss concentration. For the Flip-or-Stay multi-angle ansatz, this analysis also yields a barren-plateau criterion. Under the unitary $2$-design assumption and independent uniform sampling of the variational parameters, an exponentially large feasible set implies
exponentially small variances of all loss-function partial derivatives (Corollary~\ref{cor:BP_consequences}). Thus, feasibility-preserving dynamics can remain susceptible to barren plateaus even though the evolution is restricted entirely to valid solutions.

We emphasize this subspace-restricted behavior in the deep-circuit limit to decouple group-level expressivity from finite-depth trainability, which ultimately depends on the convergence rate of the circuit ensemble toward the Haar measure on the dynamical Lie group $G=e^{\mathfrak{g}}$. This distinction is vital for shallow QAOA architectures, where the dynamical Lie algebra can significantly overestimate the set of unitaries attainable at the specified depth.

A concrete manifestation of this overestimation phenomenon was
recently demonstrated by Copp et al.~\cite{Copp2026}. They studied
QAOA for Maximum Independent Set using the unconstrained
edge-penalty Hamiltonian 
\begin{equation}
\label{eq:mis_qubo_hamiltonian}
    H_P:=-\sum_{v\in V}(I-Z_v)
    +\lambda\sum_{\{a,b\}\in E}(I-Z_a)(I-Z_b)-cI=
    \sum_{v\in V}\bigl(1-\lambda\deg(v)\bigr)Z_v +
    \lambda\sum_{\{a,b\}\in E}Z_aZ_b,
\end{equation}
where \(c=-|V|+\lambda|E|\), so that the scalar component of the
original penalty Hamiltonian has been removed. We pair this
Hamiltonian with the transverse-field mixer
\(H_M=-\sum\limits_{v\in V}X_v\). The resulting dynamical Lie algebra $\mathfrak{g}$ contains the MaxCut DLA as a subalgebra (which follows directly from the proof technique of \cite[Lemma~A.5]{TBFS}) and is thus widely conjectured to scale exponentially in dimension with $n$ for generic connected graphs (see~\cite{MYAZ}). Across a large ensemble of shallow QAOA instances, they observed highly localized \textit{cragged-terrain} landscapes rather than the asymptotic barren-plateau behavior suggested by deep-circuit representations. Their findings illustrate that group-theoretic asymptotic predictions need not capture landscape geometry at shallow circuit depths. In contrast, the present work focuses on feasibility-preserving mixers, for which the relevant dynamical Lie algebras and loss-variance bounds are naturally restricted to the feasible subspace $W_{\mathcal{F}}$.

Finally, we compare the proposed local mixers with Grover-mixer QAOA (GM-QAOA) \cite{BE,TNB}. Initialized in the uniform superposition of feasible states $\ket{\xi_{\mathcal F}}$, GM-QAOA employs a global rank-one projector mixer whose evolution is confined to an $m$-dimensional subspace spanned by uniform superpositions over the distinct objective-value levels within $\mathcal F$, where $m$ denotes the number of such levels. Consequently, the dynamical Lie algebra of GM-QAOA on this subspace is isomorphic to $\mathfrak{u}(m)$ and has dimension $m^2$; see Theorem~III.1 of \cite{TNB}. In contrast, the free Flip-or-Stay ans\"atze achieve full controllability on the $|\mathcal F|$-dimensional feasible subspace $W_{\mathcal F}$, with dynamical Lie algebra $\mathfrak{u}(W_{\mathcal F})$, enabling it to distinguish and resolve individual states within the same objective level. This comparison highlights a fundamental hierarchy between nonlocal symmetry reduction and local controllability: while the Grover mixer achieves a drastic reduction of the effective DLA to $\mathfrak{u}(m)$, the local Flip-or-Stay mixer preserves the complete internal combinatorial structure of the feasible configuration space.


The main theoretical results and structural insights of this manuscript are summarized in Figure~\ref{fig:qaoa_mis_overview}.



\begin{figure}
    \centering
      \resizebox{\linewidth}{!}{
    \begin{tikzpicture}[
    x=1cm,
    y=1cm,
    every node/.style={font=\sffamily, text=ink},
    card/.style={
        draw=linegray,
        fill=white,
        rounded corners=5pt,
        line width=0.8pt
    },
    flow/.style={
        -{Latex[length=3.2mm,width=2.2mm]},
        line width=1.4pt,
        draw=muted!75
    },
    graphnode/.style={
        circle,
        draw=ink!85,
        fill=blue!10,
        line width=0.65pt,
        minimum size=0.39cm,
        inner sep=0pt
    },
    rnode/.style={
        circle,
        draw=ink!75,
        fill=teal!13,
        line width=0.55pt,
        minimum size=0.27cm,
        inner sep=0pt
    },
    badge/.style={
        circle,
        fill=#1,
        text=white,
        minimum size=0.58cm,
        inner sep=0pt,
        font=\sffamily\bfseries\large
    }
]

\path[fill=paperbg,rounded corners=8pt]
    (-0.25,0.45) rectangle (18.25,-12.15);

\node[
    anchor=west,
    font=\sffamily\bfseries\LARGE,
    text=ink
] at (0.3,0)
{QAOA for MIS: structural properties \& theoretical guarantees};


\draw[line width=2.2pt,blue]
    (0,-0.95)--(18,-0.95);


\path[card] (0,-1.28) rectangle (5.45,-6.45);
\draw[line width=3pt,blue]
    (0.18,-1.28)--(5.27,-1.28);

\node[badge=blue] at (0.52,-1.75) {1};

\node[
    anchor=west,
    font=\sffamily\bfseries\normalsize
] at (0.92,-1.75)
{Feasible geometry};

\node[
    align=center,
    text width=4.85cm,
    font=\sffamily\scriptsize,
    text=muted
] at (2.725,-2.24)
{Independent sets form a connected configuration graph};

\node[graphnode] (g1) at (0.78,-3.22) {};
\node[graphnode] (g2) at (1.45,-2.72) {};
\node[graphnode] (g3) at (2.12,-3.22) {};
\node[graphnode] (g4) at (1.45,-3.72) {};

\draw[ink!72,line width=0.65pt]
    (g1)--(g2)--(g3)--(g4)--(g1);

\node[font=\sffamily\small,text=muted] at (1.45,-4.08)
{$\Gamma$};

\draw[-{Latex[length=2.5mm]},line width=1pt,draw=blue!75]
    (2.38,-3.22)--(2.85,-3.22);

\node[rnode] (r0) at (3.08,-3.22) {};

\node[rnode] (r1a) at (3.78,-2.72) {}; 
\node[rnode] (r1b) at (3.78,-3.05) {}; 
\node[rnode] (r1c) at (3.78,-3.39) {}; 
\node[rnode] (r1d) at (3.78,-3.72) {}; 

\node[rnode] (r2a) at (4.55,-2.89) {}; 
\node[rnode] (r2b) at (4.55,-3.55) {}; 

\foreach \v in {r1a,r1b,r1c,r1d}{
    \draw[ink!55,line width=0.45pt] (r0)--(\v);
}

\foreach \a in {r1a,r1b}{
    \draw[ink!55,line width=0.45pt] (\a)--(r2a);
}

\foreach \a in {r1c,r1d}{
    \draw[ink!55,line width=0.45pt] (\a)--(r2b);
}

\node[font=\sffamily\small,text=muted] at (3.85,-4.08)
{$\mathcal R(\Gamma)$};

\node[
    align=center,
    text width=4.70cm,
    font=\sffamily\small
] at (2.725,-4.6)
{$\text{Vertices of } \mathcal{R}(\Gamma) \overset{1:1}{\longleftrightarrow} \mathcal{F}$  and 
$\{x,y\}\in E(\mathcal{R}(\Gamma))\Longleftrightarrow y=x\oplus e_v$};

\path[
    draw=blue!35,
    fill=blue!6,
    rounded corners=4pt,
    line width=0.7pt
] (0.35,-5.1) rectangle (5.30,-6.33);

\node[
    font=\sffamily\bfseries\large,
    text=blue!80!black
] at (2.725,-5.45)
{$\displaystyle H_{CX}\big|_{W_{\mathcal F}}=A_{\mathcal R(\Gamma)}$};

\node[
    align=center,
    text width=4.45cm,
    font=\sffamily\scriptsize,
    text=muted
] at (2.725,-6.03)
{The mixer generates a continuous-time quantum walk on
$\mathcal R(\Gamma)$};

\path[card] (5.75,-1.28) rectangle (11.55,-6.45);
\draw[line width=3pt,teal]
    (5.93,-1.28)--(11.37,-1.28);

\node[badge=teal] at (6.27,-1.78) {2};
\node[
    anchor=west,
    font=\sffamily\bfseries\normalsize
] at (6.67,-1.78)
{Laplacian mixer};


\node[
    draw=teal!35,
    fill=teal!5,
    rounded corners=8pt,
    minimum width=2.10cm,
    minimum height=0.68cm,
    font=\sffamily\scriptsize
] at (7.30,-2.50)
{Allowed: $\lvert x\rangle\mapsto\lvert x^{(v)}\rangle$};

\node[
    draw=teal!35,
    fill=teal!5,
    rounded corners=8pt,
    minimum width=2.10cm,
    minimum height=0.68cm,
    font=\sffamily\scriptsize
] at (10.05,-2.50)
{Blocked: $\ket{x}\mapsto\ket{x}$};

\node[
    font=\sffamily\bfseries\large,
    text=teal!72!black
] at (8.65,-3.35)
{$\displaystyle
\widehat H_{CX}\big|_{W_{\mathcal F}}
=nI-L_{\mathcal R(\Gamma)}$};


\node[
    align=center,
    font=\sffamily\small
] at (8.65,-4.20)
{$\displaystyle
\ket{\xi_{\mathcal F}}
=\frac{1}{\sqrt{|\mathcal F|}}\!\sum_{x\in\mathcal F}\ket{x}$};

\node[
    align=center,
    font=\sffamily\scriptsize,
    text=muted
] at (8.65,-4.74)
{Unique ground state of
$-\widehat H_{CX}$};

\path[
    draw=teal!38,
    fill=teal!7,
    rounded corners=4pt,
    line width=0.7pt
] (6.15,-5.07) rectangle (11.15,-6.35);

\node[font=\sffamily\bfseries\normalsize,text=teal!75!black]
    at (8.65,-5.24)
{Spectral gap};
\node[font=\sffamily\small,text=muted]
    at (8.65,-5.7)
{$\displaystyle
\Delta\!\left(-\widehat H_{CX}\big|_{W_{\mathcal F}}\right)
=\lambda_2(L_{\mathcal R(\Gamma)})$};
\node[
    align=center,
    font=\sffamily\scriptsize,
    text=muted
] at (8.65,-6.17)
{Algebraic connectivity of $\mathcal{R}_\Gamma$};


\path[card] (11.85,-1.28) rectangle (18,-6.45);
\draw[line width=3pt,coral]
    (12.03,-1.28)--(17.82,-1.28);

\node[badge=coral] at (12.37,-1.78) {3};
\node[
    anchor=west,
    font=\sffamily\bfseries\normalsize
] at (12.77,-1.78)
{Universal control};

\node[
    anchor=east,
    rounded corners=7pt,
    fill=coral!13,
    text=coral!70!black,
    inner xsep=5pt,
    inner ysep=3pt,
    font=\sffamily\bfseries\scriptsize
] at (17.65,-1.78)
{THEOREM};

\node[
    align=center,
    font=\sffamily\small,
    text=muted
] at (14.925,-2.40)
{For every connected $\Gamma$ with $|V|\ge 2$};

\path[
    draw=coral!42,
    fill=coral!6,
    rounded corners=4pt,
    line width=0.9pt
] (12.25,-2.72) rectangle (17.60,-4.10);

\node[
    align=center,
    text width=5.90cm,
    font=\sffamily\bfseries\large,
    text=coral!70!black
] at (14.925,-3.40)
{$\displaystyle
\mathfrak g_{\Gamma,\mathrm{free},\mathcal F}
\cong\widehat{\mathfrak g}_{\Gamma,\mathrm{free},\mathcal F}
\cong
\mathfrak u(W_{\mathcal F})$};

\node[
    circle,
    draw=coral!65,
    fill=coral!9,
    minimum size=0.80cm,
    font=\sffamily\bfseries
] (h) at (13.15,-5.82)
{$\ket{h}$};

\node[
    rounded corners=4pt,
    fill=coral,
    text=white,
    minimum width=1.15cm,
    minimum height=0.62cm,
    font=\sffamily\bfseries
] (u) at (14.925,-5.82)
{$U_p(\boldsymbol{\theta})$};

\node[
    circle,
    draw=coral!65,
    fill=coral!9,
    minimum size=0.80cm,
    font=\sffamily\bfseries
] (t) at (16.70,-5.82)
{$\ket{t}$};

\draw[-{Latex[length=2.5mm]},coral,line width=1.2pt]
    (h)--(u);
\draw[-{Latex[length=2.5mm]},coral,line width=1.2pt]
    (u)--(t);

\node[
    align=center,
    text width=5.25cm,
    font=\sffamily\bfseries\small
] at (14.925,-4.73)
{Full controllability and finite-depth exact synthesis on $W_{\mathcal F}$};



\path[
    draw=violet!35,
    fill=violet!5,
    rounded corners=5pt,
    line width=0.8pt
] (0,-6.83) rectangle (9.05,-9.22);

\node[
    anchor=west,
    font=\sffamily\bfseries\small,
    text=violet!78!black
] at (-0.05,-7.22)
{Standard QAOA: exact solution preparation at finite depth};

\node[
    anchor=west,
    font=\sffamily\bfseries\normalsize
] at (0.23,-7.82)
{$\displaystyle
\exists\,p,\widehat{p} \in\mathbb{Z}_{>0}:
\ U_{p}\!\left( \boldsymbol{\beta}, \boldsymbol{\gamma} \right) \ket{0}^n,\widehat{U}_{\widehat{p}}\!\left( \widehat{\boldsymbol{\beta}}, \widehat{\boldsymbol{\gamma}} \right) \ket{0}^n\in W_{P,\min}$};

\node[
    anchor=west,
    align=left,
    text width=8.3cm,
    font=\sffamily\scriptsize
] at (0.9,-8.57)
{For either mixer, vacuum initialization permits \textbf{exact} state preparation within the optimal-solution subspace $W_{P,\min}$ at \textbf{finite depth} $p$, rather than asymptotically as $p \to \infty$.};

\path[
    draw=gold!42,
    fill=gold!6,
    rounded corners=5pt,
    line width=0.8pt
] (9.35,-6.83) rectangle (18,-9.22);

\node[
    anchor=west,
    font=\sffamily\bfseries\small,
    text=gold!75!black
] at (10.7,-7.22)
{ma-QAOA deep circuit landscape scale};

\node[
    align=center,
    text width=8.0cm,
    font=\sffamily\small
] at (13.5,-7.85)
{In the deep-circuit regime, loss function variance\\ of ma-QAOA scales with $|\mathcal{F}|$ rather than $2^n$:};

\node[
    align=center,
    text width=7.9cm,
    font=\sffamily\scriptsize
] at (13.675,-8.63)
{$\displaystyle
\operatorname{Var}_{\boldsymbol{\theta}}
    \left[
        \ell_{\boldsymbol{\theta}}
        \left(
            \rho,
            H_P\big|_{W_{\mathcal{F}}}
        \right)
    \right]
    \in
    \Omega\left(\frac{1}{|\mathcal{F}|^2}\right)
    \cap
    \mathcal{O}\left(
        \frac{\alpha^2(\Gamma)}{|\mathcal{F}|}
    \right).$};

\path[
    draw=blue!38,
    fill=blue!5,
    rounded corners=5pt,
    line width=0.8pt
] (0,-9.52) rectangle (18,-12.03);

\node[
    anchor=west,
    font=\sffamily\bfseries\small,
    text=blue!78!black
] at (2.85,-9.85)
{Same problem Hamiltonian, sharp expressivity gap between QAOA ansätze};

\node[
    align=center,
    text width=16.8cm,
    font=\sffamily\small
] at (9,-10.63)
{The infinite graph family $\{\Gamma_r\}_{r \ge 5,\, r \text{ odd}}$ demonstrates a sharp separation in the expressivity of \textbf{adjacency-based and Laplacian-based QAOA ansätze} when initialized in the computational vacuum state $\ket{0}^n$:};

\node[
    align=center,
    text width=12.5cm,
    font=\sffamily\bfseries\normalsize,
    text=blue!80!black
] at (9,-11.25)
{$\displaystyle
\dim_{\mathbb R}\!\left(G_{\Gamma_r}\cdot\ket{0}^n\right)\leq 7
\qquad\text{whereas}\qquad
\dim_{\mathbb R}\!\left(\widehat G_{\Gamma_r}\cdot\ket{0}^n\right)=3r-2
$};

\node[
    align=center,
    text width=16.5cm,
    font=\sffamily\scriptsize,
    text=muted
] at (9,-11.73)
{The adjacency-based ansatz remains confined to a real orbit of dimension at most $7$, whereas the orbit of the Laplacian-based ansatz expands linearly with $r$.};
\end{tikzpicture}
}

  \caption{\textbf{Overview of main results for constrained MIS-QAOA.}
    \textbf{(1)} Restricted to the feasible subspace $W_{\mathcal{F}}$, the standard Flip-or-Void mixer $H_{\mathrm{CX}}\big|_{W_{\mathcal{F}}}$ equals the adjacency matrix $A_{\mathcal{R}}$ of the independent set reconfiguration graph $\mathcal{R}(\Gamma)$, thereby inducing a continuous-time quantum walk on $\mathcal{R}(\Gamma)$ initialized at the vertex corresponding to the initial computational state \cite[Section~VII]{QAOA}.
    \textbf{(2)} The proposed Flip-or-Stay mixer $\widehat{H}_{CX}$ restricts to the shifted negative graph Laplacian $nI - L_{\mathcal{R}}$, whose spectral gap $\Delta\big(\!-\widehat{H}_{CX}\big|_{W_{\mathcal{F}}}\!\big) = \lambda_2(L_{\mathcal{R}})$ equals the algebraic connectivity of $\mathcal{R}(\Gamma)$, with unique ground state given by the uniform feasible superposition $\ket{\xi_{\mathcal{F}}}$.
    \textbf{(3)} For any connected graph $\Gamma$ with $|V| \ge 2$, free multi-angle QAOA ansatz (for either mixer) yields the maximal dynamical Lie algebra $\mathfrak{g}_{\Gamma,\mathrm{free},\mathcal{F}} \cong \widehat{\mathfrak{g}}_{\Gamma,\mathrm{free},\mathcal{F}} \cong \mathfrak{u}(W_{\mathcal{F}})$ (Theorem~\ref{thm:ExtendedFreeLieAlgFullUnitary}), ensuring universal state controllability and exact state synthesis on the feasible subspace $W_{\mathcal{F}}$.
    \textbf{Bottom Left:} initialized in the computational vacuum state $\ket{0}^n$, standard QAOA driven by either mixer achieves exact ground-state preparation at a finite circuit depth (Theorem~\ref{thm:qaoa_perron_frobenius}).
    \textbf{Bottom Right:} in the deep-circuit regime, the loss function variance for ma-QAOA is governed by the feasible sector dimension $|\mathcal{F}|$, satisfying $\operatorname{Var}_{\boldsymbol{\theta}}
    \left[
        \ell_{\boldsymbol{\theta}}
        \left(
            \rho,
            H_P\big|_{W_{\mathcal{F}}}
        \right)
    \right]   \in
    \Omega\left(\frac{1}{|\mathcal{F}|^2}\right)
    \cap
    \mathcal{O}\left(
        \frac{\alpha^2(\Gamma)}{|\mathcal{F}|}
    \right)$ rather than scaling with the full $2^n$-dimensional Hilbert space.
    \textbf{Bottom:} initialized in the computational vacuum state $\ket{0}^n$, parameterized QAOA circuits for MIS on the infinite graph family $\Gamma_r$ exhibit a sharp expressivity gap between mixer architectures (Theorem~\ref{thm:explicit_family_orbit_separation}): states reachable by the adjacency-based QAOA are \textbf{confined to} a manifold of dimension at most $7$ (independent of $r$), whereas the Laplacian-based QAOA generates a linearly expanding $(3r - 2)$-dimensional state orbit with complete controllability on the corresponding cyclic subspace generated by $\ket{0}^n$.}
    \label{fig:qaoa_mis_overview}
\end{figure}
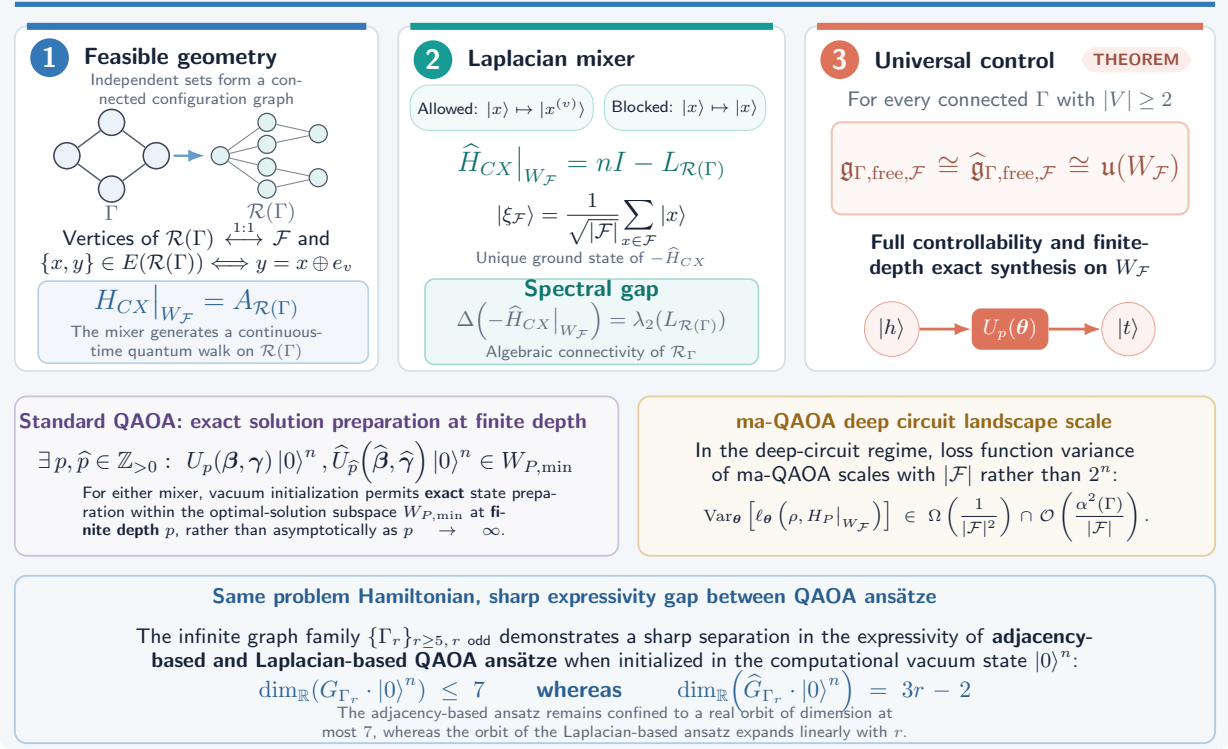

\section{Background}
\textbf{Setup.} Let $\Gamma=(V,E)$ be a graph on $n$ vertices. Recall that an independent set in $\Gamma$ is a subset of vertices containing no two adjacent vertices. We identify subsets of $V(\Gamma)$ with binary strings in $\mathbb{B}^n=\{0,1\}^n$, where $x_v=1$ corresponds to including the vertex $v$ in the subset and $x_v=0$ to omitting it. The objective function for the Maximum Independent Set problem is then
\begin{equation}
    \label{eq:objective_function}
    F(x)=\sum\limits_{v\in V(\Gamma)} x_v.
\end{equation}

A convenient choice of the corresponding problem Hamiltonian, representing this objective function up to an additive constant and an overall scaling factor, is 

\begin{equation}
    \label{eq:MIS_problem_hamiltonian}
    H_P=\sum\limits_{v\in V(\Gamma)} Z_v .
\end{equation}

Note that 
\[
H_P\ket{x}=(n-2F(x))\ket{x}.
\]
Thus minimizing \(H_P\) on \(W_{\mathcal F}\) is equivalent to maximizing independent-set cardinality with the optimal energy being \(n-2\alpha(\Gamma)\).

The independent set constraint requires that no edge in $\Gamma = (V,E)$ has both endpoints selected. Accordingly, we denote the family of independent vertex subsets by $\mathcal{I}(\Gamma)$ and define the corresponding feasible set of binary strings $\mathcal{F} \subseteq \{0,1\}^n$ as:
\begin{equation}
    \label{eq:MIS_feasible_set}
    \begin{aligned}
        \mathcal{I}(\Gamma) &:= \bigl\{ S \subseteq V \;\big|\; \{u,v\} \notin E \text{ for all } u,v \in S \bigr\}, \\
        \mathcal{F} &:= \bigl\{ x \in \{0,1\}^n \;\big|\; x_u x_v = 0 \text{ for all } \{u,v\} \in E \bigr\}.
    \end{aligned}
\end{equation}
For each integer $j \ge 0$, we partition $\mathcal{F}$ into level sets of fixed size:
\begin{equation}
    \label{eq:MIS_feasible_set_j}
    \mathcal{F}_j := \bigl\{ x \in \mathcal{F} \;\big|\; \operatorname{Ham}(x) = j \bigr\},
\end{equation}
where $\operatorname{Ham}(x) := \sum\limits_{v \in V} x_v$ denotes the Hamming weight of $x$ (i.e., the number of entries in $x$ equal to $1$, which equals the cardinality of the independent set encoded by $x$). The independence number $\alpha(\Gamma)$ of $\Gamma$ is then the maximum weight among all feasible configurations:
\begin{equation}
    \label{eq:independence_number}
    \alpha(\Gamma) := \max \bigl\{ j \;\big|\; \mathcal{F}_j \neq \varnothing \bigr\}.
\end{equation}

Since $\mathcal{F}_j = \varnothing$ for all $j > \alpha(\Gamma)$ by definition of $\alpha(\Gamma)$, and $\mathcal{F}_0 = \{0^n\}$ corresponds to the empty subset of vertices, the feasible set $\mathcal{F}$ can be partitioned into its non-empty, disjoint weight layers. This yields the following decomposition:
\begin{equation}
    \label{eq:feasible_set_disjoint_union}
    \mathcal{F}=\bigsqcup_{j=0}^{\alpha(\Gamma)}\mathcal{F}_j.
\end{equation}

\textbf{QAOA and feasibility-preserving mixers.}
The Quantum Approximate Optimization Algorithm (QAOA) is a variational quantum algorithm designed to solve classical optimization problems encoded into a diagonal problem Hamiltonian $H_P$ \cite{QAOA}. The ansatz alternates unitary evolutions generated by $H_P$ and a mixer Hamiltonian $H_M$, where the latter drives transitions between computational basis states to explore the search space. At depth $p$, the algorithm prepares the parameterized state
\begin{equation}
\label{eq:qaoa_ansatz_compact}
\ket{\psi_p(\boldsymbol{\beta},\boldsymbol{\gamma})} = U_p(\boldsymbol{\beta},\boldsymbol{\gamma})\ket{\psi_0},
\end{equation}
where the depth-$p$ unitary evolution operator is defined by
\begin{equation}
\label{eq:qaoa_chain}
U_p(\boldsymbol{\beta},\boldsymbol{\gamma}) := U_M(\beta_p) U_P(\gamma_p) \cdots U_M(\beta_1) U_P(\gamma_1) = e^{-i \beta_p H_M} e^{-i \gamma_p H_P} \cdots e^{-i \beta_1 H_M} e^{-i \gamma_1 H_P},
\end{equation}
with variational parameters $\boldsymbol{\beta} = (\beta_1, \dots, \beta_p) \in \mathbb{R}^p$ and $\boldsymbol{\gamma} = (\gamma_1, \dots, \gamma_p) \in \mathbb{R}^p$. 

\begin{rmk}
\label{rmk:simultaneous_mixer_evolution}
When the mixer Hamiltonian $H_M = \sum\limits_{j} H_{M,j}$ is comprised of a sum of local generator terms, our standard ans\"atze employ simultaneous Hamiltonian evolution $U_M(\beta) \coloneqq e^{-i\beta H_M}$ parameterized by a single variational angle $\beta$ per layer. If the constituent terms commute (e.g., the transverse-field mixer $H_M = \sum\limits_j X_j$), this factors exactly into a product of local unitaries $\prod_j e^{-i\beta H_{M,j}}$. 
\end{rmk}

These parameters are optimized classically using measurement-based estimates of the expected objective value
\begin{equation}
\label{eq:qaoa_objective_compact}
\mathbb{E}(\boldsymbol{\beta},\boldsymbol{\gamma}) := \bra{\psi(\boldsymbol{\beta},\boldsymbol{\gamma})} H_P \ket{\psi(\boldsymbol{\beta},\boldsymbol{\gamma})}.
\end{equation}

An appealing structural feature of the mixer arises from the adiabatic foundation of QAOA. If there exists a constant $c \in \mathbb{R}$ such that $cI - H_M$ is entrywise nonnegative and irreducible in the computational basis of the relevant Hilbert space $W$, the Perron--Frobenius theorem guarantees that $H_M$ possesses a unique, nondegenerate ground state $\ket{\xi}$. Initializing QAOA in $\ket{\xi}$ allows the alternating circuit structure in \eqref{eq:qaoa_chain} to be interpreted as a digitized, variational analog of an adiabatic path from $H_M$ to $H_P$ \cite{farhi2000quantum,QAOA}.

For constrained optimization problems, Hadfield \emph{et al.}~\cite{HWORVB} generalized QAOA to the Quantum Alternating Operator Ansatz, replacing the standard unconstrained transverse-field mixer with a feasibility-preserving operator. Let $W := (\mathbb{C}^2)^{\otimes n}$ denote the Hilbert space of $n$ qubits.

Given a feasible configuration set $\mathcal{F} \subseteq \{0,1\}^n$, the corresponding feasible subspace $W_{\mathcal{F}}$ is defined as
\begin{equation}
    \label{eq:MIS_feasible_subspace}
    W_{\mathcal{F}} := \operatorname{span} \left\{ \ket{x} \;\middle|\; x \in \mathcal{F} \right\} \subseteq (\mathbb{C}^2)^{\otimes n}.
\end{equation}
To ensure system dynamics remain within the valid search space, one selects an initial state $\ket{\psi_0} \in W_{\mathcal{F}}$ and a mixer Hamiltonian $H_M$ whose generated unitaries $U_M(\beta) = e^{-i \beta H_M}$ preserve $W_{\mathcal{F}}$:
\begin{equation}
    \label{eq:feasible_preserving_general}
    U_M(\beta) W_{\mathcal{F}} \subseteq W_{\mathcal{F}} \qquad \text{for all } \beta \in \mathbb{R}.
\end{equation}
This condition restricts the state evolution strictly to $W_{\mathcal{F}}$, enforcing hard constraints natively at the circuit level.


\begin{rmk}
    \label{rmk:QAOA_Adiabatic_Difference}
    While it is conventional to start with a lowest energy state of the mixer Hamiltonian, the operational mechanism of QAOA is fundamentally different from that of the quantum adiabatic algorithm.  In particular, successful optimization can be achieved even if the system is not initiated in a ground state of the mixer Hamiltonian. In some cases, theoretical convergence guarantees can be established even when starting from an arbitrary state within a certain invariant subspace of the Hilbert space. Specifically, this will be shown to be the case for the QAOA applied to the MIS problem discussed later in this work.

\end{rmk}

\section{Constraint-preserving mixer for MIS}
\textbf{The Two Mixers for MIS.} Having reviewed the general structure of the QAOA algorithm, we now specialize this framework to the Maximum Independent Set problem. 

Since the objective is to maximize the size of the independent set while strictly avoiding invalid configurations, it is natural to seek a mixer Hamiltonian that preserves the subspace of feasible states throughout the QAOA alternating layers. Notice that the standard choice of a mixer, given by the sum of single-qubit Pauli $X$-gates, does not preserve feasibility, as individual bit flips can drive the state out of the valid subspace $W_\mathcal{F}$. 

\textbf{The Flip-or-Void Mixer.} A canonical approach to resolve this was introduced by Farhi et al.~\cite{QAOA} (Section VII); for a more precise formulation, which we reproduce below, see also Section~4.1.2 of~\cite{HWORVB}. For each vertex $v \in V(\Gamma)$, let $\mathcal{N}_v$ denote its neighborhood. The partial mixer Hamiltonian $H_{CX,v}$ and the total \textit{Flip-or-Void mixer Hamiltonian} $H_{CX}$ are defined by
\begin{equation}
    \label{eq:CX_mixer}
    \begin{aligned}
    \mathcal{P}_v:=\prod\limits_{w \in \mathcal{N}_v} \frac{I+Z_w}{2},\qquad
     H_{CX,v}:=X_v \cdot \mathcal{P}_v, \qquad
    H_{CX} :=\sum\limits_{v \in V(\Gamma)}H_{CX,v}. 
    \end{aligned}
\end{equation}
The feasible-subspace preservation property follows directly from the relation $\frac{I + Z_w}{2}\ket{1}=0$. Indeed, let $\ket{x}\in W_{\mathcal{F}}$ and let
\begin{equation}
    \label{eq:v-bit-flip}
    x^{(v)}:=x\oplus e_v
\end{equation}
denote the binary string obtained from \(x\) by flipping the \(v\)-th bit. Then the action of the local term \(H_{CX,v}\) on a computational basis state is given by
\begin{equation}
    \label{eq:HCXv_action}
    H_{CX,v}\ket{x}
    =
    \begin{cases}
        0,
        &
        x_w=1
        \text{ for some }
        w\in \mathcal{N}_v,
        \\[1ex]
        \ket{x^{(v)}},
        &
        x_w=0
        \text{ for all }
        w\in \mathcal{N}_v.
    \end{cases}
\end{equation}
Therefore, the operator $X_v$ may act nontrivially only when all neighbors of $v$ are unoccupied, in which case flipping the bit at $v$ preserves feasibility. Consequently,
\begin{equation}
    \label{eq:mixerPreservesW_F}
    \begin{aligned}
        H_{CX,v}(W_{\mathcal{F}})&\subseteq W_{\mathcal{F}},\\
        H_{CX}(W_{\mathcal{F}})&\subseteq W_{\mathcal{F}}.
    \end{aligned}
\end{equation}

\begin{rmk}
\label{rmk:ground_state_drawback} 
A notable limitation of the $H_{CX}$ mixer is that while its transitions preserve the feasible subspace $W_{\mathcal{F}}$, the Perron eigenstate of $H_{CX}\big|_{W_{\mathcal{F}}}$, equivalently, the ground state of $-H_{CX}\big|_{W_{\mathcal{F}}}$, has no general, efficient preparation procedure. Consequently, the standard choice for the initial state is typically the vacuum state
$\ket{0}^{\otimes n}$,
which physically corresponds to the empty independent set~\cite{QAOA, HHR, Saleem}. Although $\ket{0}^{\otimes n}$ is guaranteed to be feasible for any graph $\Gamma$, it is not an eigenstate of $H_{CX}$ restricted to $W_{\mathcal{F}}$. Therefore, the adiabatic convergence theorem of~\cite{BKZS} does not apply directly with this initialization.
\end{rmk}

\begin{rmk}
    \label{rmk:MatricesOfMixers}

    For each $v\in V$, let
\[
    N_v:=\frac{I-Z_v}{2}
\]
denote the occupation-number projector at vertex $v$. The orthogonal
projector onto the feasible subspace can be realized as
\begin{equation}
\label{eq:FeasibleProjector}
    \mathcal{P}_{\mathcal F}:=\prod_{(v,w)\in E(\Gamma)}
    \bigl(I-N_vN_w\bigr)=\sum_{x\in\mathcal F}\ket{x}\!\bra{x}.
\end{equation}
    The restrictions of these mixers to the feasible subspace can be related back to the unconstrained transverse-field mixer $B = \sum\limits_v X_v$ via the following proposition.
\end{rmk}

\begin{prop}
\label{prop:mixer_restriction_properties}
The restriction of the multi-controlled mixer components to the feasible subspace matches the projection of the standard unconstrained bit-flip terms: 
\begin{equation}
\label{eq:constrained_hamiltonians}
\begin{aligned}
H_{CX,v}\big|_{\mathcal{W}_\mathcal{F}} &= \mathcal{P}_{\mathcal{F}} X_v \mathcal{P}_{\mathcal{F}}\big|_{\mathcal{W}_\mathcal{F}}, \\
H_{CX}\big|_{\mathcal{W}_\mathcal{F}}   &= \mathcal{P}_{\mathcal{F}} B \mathcal{P}_{\mathcal{F}}\big|_{\mathcal{W}_\mathcal{F}}.
\end{aligned}
\end{equation}
\end{prop}

In addition, note that the matrix representation of a partial mixer $H_{CX,v}$ in the computational basis of $W_\mathcal{F}$ has entries of $1$ indexed by pairs of feasible strings $x,y\in\mathcal{F}$ that differ exclusively at vertex $v$:
\begin{equation}
    x_w = y_w \quad \forall w \neq v, \text{ and } x_v = 1 - y_v,
\end{equation}
and $0$ entries elsewhere.

\textbf{The Flip-or-Stay Mixer.}
To address the initialization issue discussed in
Remark~\ref{rmk:ground_state_drawback}, we consider a closely related
mixer whose restriction to the feasible subspace has an explicitly
known extremal eigenstate. We call it the \textit{Flip-or-Stay} mixer
and define its local and global terms by
\begin{equation}
    \label{eq:Bao_mixer}
    \begin{aligned}
        \widehat{H}_{CX,v}
        &:= X_v\mathcal{P}_v+(I-\mathcal{P}_v)
        =H_{CX,v}+(I-\mathcal{P}_v),\\
        \widehat{H}_{CX}
        &:=\sum_{v\in V(\Gamma)}\widehat{H}_{CX,v}.
    \end{aligned}
\end{equation}
The local term flips the bit at $v$ when all its neighbors are
unoccupied and otherwise leaves the basis state unchanged:
\begin{equation}
    \label{eq:HCXvHat_action}
    \widehat{H}_{CX,v}\ket{x}
    =
    \begin{cases}
        \ket{x},
        &
        x_w=1
        \text{ for some }w\in\mathcal{N}_v,
        \\[1ex]
        \ket{x^{(v)}},
        &
        x_w=0
        \text{ for all }w\in\mathcal{N}_v.
    \end{cases}
\end{equation}

The Flip-or-Stay mixer is related directly to the independent-set
sampling Hamiltonian of Wild et al.\
\cite[Eq.~(4)]{wild2021quantum}, constructed from single-spin
Metropolis--Hastings updates. At inverse temperature $\beta=0$, the coefficients in their Hamiltonian satisfy $V_{e,v}(0)=V_{g,v}(0)=\Omega_v(0)=1$. Consequently, on the feasible subspace,
\[
H_q(0)
=
\sum_{v\in V(\Gamma)}
\mathcal P_v(I-X_v)\big|_{W_{\mathcal F}}
=
(nI-\widehat H_{CX})\big|_{W_{\mathcal F}}.
\]
Thus, up to an additive scalar and an overall sign, the
Flip-or-Stay mixer coincides with the zero-inverse-temperature
specialization of their independent-set parent Hamiltonian.

We study its use in constrained QAOA and compare its dynamical
Lie algebra and reachable states with those of the Flip-or-Void
mixer. As shown below, the uniform feasible state is the unique
ground state of $-\widehat{H}_{CX}|_{W_{\mathcal F}}$; this explicit
characterization does not by itself imply efficient state preparation.

\begin{rmk}
    \label{rmk:favorable_properties_of_new_mixer}
    To highlight the operational distinction between the two local mixers, consider a computational basis state $\ket{x}$ that violates the neighborhood condition at vertex $v$. Under this condition, the Flip-or-Stay mixer $\widehat{H}_{CX,v}$ acts as the identity, leaving the state invariant ($\widehat{H}_{CX,v}\ket{x} = \ket{x}$), whereas the Flip-or-Void mixer $H_{CX,v}$ annihilates it, i.e. $H_{CX,v}\ket{x} = 0$ (see Figure \ref{fig:flip_or_void_stay}). 
\end{rmk}

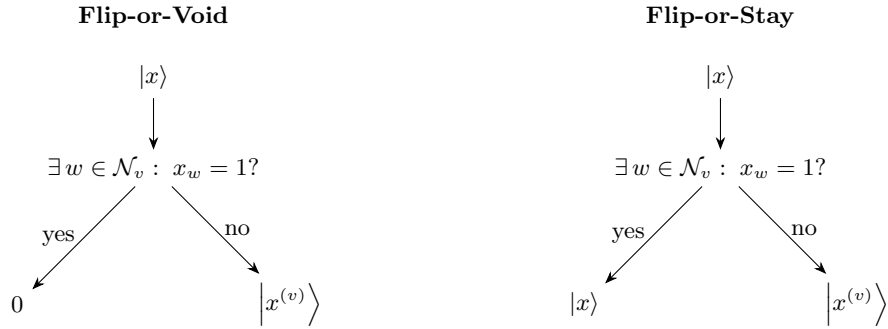
\begin{figure}[h]
\centering
\begin{tikzpicture}[
    >=Stealth,
    every node/.style={font=\small},
    node distance=9mm
]

\node at (-3.5,2.0) {\textbf{Flip-or-Void}};
\node (x1) at (-3.5,1.2) {$\ket{x}$};
\node (c1) at (-3.5,0)
{$\exists\,w\in\mathcal N_v:\;x_w=1?$};

\node (z) at (-5.3,-1.8) {$0$};
\node (f1) at (-1.7,-1.8) {$\ket{x^{(v)}}$};

\draw[->] (x1)--(c1);
\draw[->] (c1)--node[left]{yes}(z);
\draw[->] (c1)--node[right]{no}(f1);

\node at (4.0,2.0) {\textbf{Flip-or-Stay}};
\node (x2) at (4.0,1.2) {$\ket{x}$};
\node (c2) at (4.0,0)
{$\exists\,w\in\mathcal N_v:\;x_w=1?$};

\node (s) at (2.2,-1.8) {$\ket{x}$};
\node (f2) at (5.8,-1.8) {$\ket{x^{(v)}}$};

\draw[->] (x2)--(c2);
\draw[->] (c2)--node[left]{yes}(s);
\draw[->] (c2)--node[right]{no}(f2);

\end{tikzpicture}
\caption{Comparison of the local mixer terms acting on a computational basis state \(\ket{x}\).  The local Flip-or-Void mixer \(H_{CX,v}\) annihilates configurations in which at least one neighbor of \(v\) is occupied, whereas the Flip-or-Stay mixer \(\widehat{H}_{CX,v}\) leaves such configurations unchanged. When all neighbors of \(v\) are unoccupied, both mixers flip the \(v\)-th bit.}
\label{fig:flip_or_void_stay}
\end{figure}

\subsection{The Independent-Set Reconfiguration Graph}
\label{subsec:IndependentSetReconfigurationGraph}

The two constrained mixers admit a natural interpretation in terms of the
reconfiguration graph of the feasible set.

\begin{defn}
\label{def:IndependentSetReconfigurationGraph}
Let \(\Gamma=(V,E)\) be a graph and let \(\mathcal F\) be its set of
independent-set bit strings. The \emph{independent-set reconfiguration graph}
\(\mathcal R(\Gamma)\) is the graph with vertex set
\[
V\bigl(\mathcal R(\Gamma)\bigr)=\mathcal F,
\]
in which two feasible strings \(x,y\in\mathcal F\) are adjacent if and only if
they differ at exactly one vertex. Equivalently,
\[
\{x,y\}\in E\bigl(\mathcal R(\Gamma)\bigr)
\quad\Longleftrightarrow\quad
y=x^{(v)}
\]
for some \(v\in V(\Gamma)\), where \(x^{(v)}=x\oplus e_v\).
\end{defn}

Recall that for a simple graph $H = (V, E)$, its \textit{simplex graph} $\kappa(H)$ is defined as the graph whose vertex set $V(\kappa(H))$ consists of all cliques of $H$ (including the empty set $\varnothing$), where two cliques $C_1, C_2$ are adjacent if and only if their symmetric difference is of cardinality one: $|C_1 \Delta C_2| = 1$.

As independent sets of $\Gamma$ correspond bijectively to cliques in the complement graph $\bar{\Gamma}$, single-vertex bit flips in $\mathcal{R}(\Gamma)$ mirror single-vertex additions or deletions in $\kappa(\bar{\Gamma})$. This combinatorial duality establishes a canonical isomorphism between the two state-space topologies.

\begin{prop}
\label{prop:reconfig_simplex_isomorphism}
The \emph{independent-set reconfiguration graph} $\mathcal{R}(\Gamma)$ of a graph $\Gamma = (V, E)$ under single-vertex addition and deletion is canonically isomorphic to the simplex graph of its complement graph $\bar{\Gamma}$:
\begin{equation}
\label{eq:reconfig_simplex_iso}
\mathcal{R}(\Gamma) \cong \kappa(\bar{\Gamma}).
\end{equation}
\end{prop}

Let \(A_{\mathcal R(\Gamma)}\) and \(D_{\mathcal R(\Gamma)}\) denote the adjacency and  degree matrices of \(\mathcal R(\Gamma)\), regarded as operators on \(W_{\mathcal F}\). The combinatorial Laplacian of \(\mathcal R(\Gamma)\) is defined as
\begin{equation}
\label{eq:Laplacian_of_R_Gamma}
    L_{\mathcal R(\Gamma)} := D_{\mathcal R(\Gamma)} - A_{\mathcal R(\Gamma)}.
\end{equation}

The following proposition provides a direct graph-theoretic characterization of the restricted mixer operators in terms of the state-space reconfiguration topology.

\begin{prop}
\label{prop:MixersAsReconfigurationOperators}
The restrictions of the Flip-or-Void and Flip-or-Stay mixers to the feasible subspace $W_{\mathcal{F}}$ satisfy
\begin{equation}
\label{eq:FlipOrStayAsLaplacian}
\begin{aligned}
    H_{CX}\big|_{W_{\mathcal{F}}} &= A_{\mathcal{R}}, \\
    \widehat{H}_{CX}\big|_{W_{\mathcal{F}}} &= nI - L_{\mathcal{R}}.
\end{aligned}
\end{equation}
\end{prop}

\begin{cor}
\label{cor:UniformStateFromReconfigurationLaplacian}
As every independent set can be reduced to the empty set $\varnothing$ by sequentially deleting occupied vertices, the reconfiguration graph $\mathcal{R}(\Gamma)$ is connected. Consequently, the kernel of its combinatorial Laplacian $L_{\mathcal{R}}$ is one-dimensional and spanned by the uniform feasible state
\begin{equation}
\label{eq:UniformFeasibleStateReconfiguration}
    \ket{\xi_{\mathcal{F}}} := \frac{1}{\sqrt{|\mathcal{F}|}} \sum_{x \in \mathcal{F}} \ket{x}.
\end{equation}
Equivalently, by identity~\eqref{eq:FlipOrStayAsLaplacian}, $n$ is the simple largest eigenvalue of $\widehat{H}_{CX}\big|_{W_{\mathcal{F}}}$ with
\begin{equation}
    \widehat{H}_{CX} \ket{\xi_{\mathcal{F}}} = n \ket{\xi_{\mathcal{F}}},
\end{equation}
rendering $\ket{\xi_{\mathcal{F}}}$ the unique ground state of the restricted mixer Hamiltonian $-\widehat{H}_{CX}\big|_{W_{\mathcal{F}}}$.
\end{cor}

\begin{rmk}
\label{rmk:ReconfigurationSpectralGap}
Let $0 = \lambda_1(L_{\mathcal{R}}) < \lambda_2(L_{\mathcal{R}}) \le \dots \le \lambda_{|\mathcal{F}|}(L_{\mathcal{R}})$ denote the eigenvalues of the combinatorial Laplacian $L_{\mathcal{R}}$. By identity~\eqref{eq:FlipOrStayAsLaplacian}, the restricted mixer Hamiltonian satisfies
\[
    -\widehat{H}_{CX}\big|_{W_{\mathcal{F}}} = L_{\mathcal{R}} - nI.
\]
Consequently, the eigenvalues of $-\widehat{H}_{CX}\big|_{W_{\mathcal{F}}}$, arranged in increasing order, are
\[
    -n < \lambda_2(L_{\mathcal{R}}) - n \le \dots \le \lambda_{|\mathcal{F}|}(L_{\mathcal{R}}) - n.
\]
The ground-state energy is thus $-n$, and the spectral gap of the Flip-or-Stay mixer simplifies to
\begin{equation}
\label{eq:spectral_gap_Bao_mixer}
    \Delta\!\left(-\widehat{H}_{CX}\big|_{W_{\mathcal{F}}}\right) = \bigl(\lambda_2(L_{\mathcal{R}}) - n\bigr) - (-n) = \lambda_2(L_{\mathcal{R}}),
\end{equation}
which is precisely the algebraic connectivity (Fiedler value) \cite{fiedler1973algebraic} of the state-space reconfiguration graph $\mathcal{R}(\Gamma)$.
\end{rmk}

We next establish the key structural properties of \(\widehat{H}_{CX}\).

\begin{rmk}
    \label{rmk:symmetry_properties}
    Both operators $H_{CX}$ and $\widehat{H}_{CX}$ inherit the structural symmetries of the underlying graph. Specifically, they are invariant under the action of the graph automorphism group $\text{Aut}(\Gamma)$:
    \begin{equation}
        \label{eq:automorphism_invariance}
        \begin{aligned}
        \varphi(g) H_{CX} \varphi(g)^{-1} &= H_{CX} \quad \forall g \in \text{Aut}(\Gamma),\\ 
        \varphi(g) \widehat{H}_{CX} \varphi(g)^{-1} &= \widehat{H}_{CX} \quad \forall g \in \text{Aut}(\Gamma),
        \end{aligned}
    \end{equation}
    where 
    \begin{equation}
    \label{eq:Aut_group_repn}
    \begin{aligned}
        \varphi:& \text{Aut}(\Gamma) \rightarrow \text{U}(W)\\
        \varphi(g)\ket{x}=\ket{g\cdot x}
    \end{aligned}
    \end{equation} is the representation capturing the action of the automorphism group on the Hilbert space $W$.
\end{rmk}

\textbf{The Action of the Flip-or-Stay Mixer Group.} We consider the subgroup $\mathcal{H} \subseteq \text{Sym}(\mathbb{B}^n)$ of permutations over the set of $n$-bit binary strings $\mathbb{B}^n$, generated by the family of local mixers $\{\widehat{H}_{CX,v}\}_{v\in V}$. It follows immediately from \eqref{eq:HCXvHat_action} that each local mixer acts as an involution:
\begin{equation}
    \label{eq:mixer_involution}
    \widehat{H}^2_{CX,v} = I.
\end{equation}

The action of the group $\mathcal{H}$ naturally partitions the configuration space into a disjoint union of orbits. Because the set of feasible states $\mathcal{F}$ itself constitutes a single $\mathcal{H}$-orbit, we can explicitly isolate it from the remaining non-feasible orbits $\{\mathcal{O}_j\}_{j=1}^k$:
\begin{equation}
    \label{eq:orbit_decomposition}
    \mathbb{B}^n = \mathcal{F} \sqcup \bigsqcup_{j=1}^k \mathcal{O}_j.
\end{equation}

\begin{rmk}
    \label{rmk:invariant_strings}
    A configuration $x \in \mathbb{B}^n$ forms a singleton orbit (i.e., remains invariant under the action of the entire group $\mathcal{H}$) if and only if for each vertex $v \in V$, there exists at least one neighbor $w \in \mathcal{N}_v$ such that $x_w = 1$. In graph theory, the support of such strings corresponds to the total dominating sets of the graph.
\end{rmk}

\begin{ex}
    Let us illustrate this decomposition using a specific family of graphs.

    For the complete graph $K_n$ on $n$ vertices, the feasible set $\mathcal{F}$ consists of all configurations containing at most one active vertex (i.e., at most one $1$). Explicitly, this is given by the all-zeros string and the standard basis vectors:
    \[
    \mathcal{F} = \{0^n\} \cup \{e_v \mid v \in V\}.
    \]
    Conversely, all remaining configurations are frozen; that is, they form singleton orbits invariant under $\mathcal{H}$.
\end{ex}

\begin{rmk}
\label{rmk:toggle_group_connection}
The permutation group induced by $\mathcal{H}$ on the feasible set $\mathcal{F}$ is the independent set toggle group studied in~\cite[Section~3.6]{striker2018rowmotion}, where local generators act as elementary toggles on independent set configurations.
\end{rmk}

\section{Free QAOA Ansätze}

The free, or multi-angle, QAOA ans\"atze generalize the standard QAOA framework by removing the constraint that all local terms within a given layer share the same variational parameter. In the standard formulation, a single parameter is assigned to all mixer terms and another to all problem Hamiltonian terms in each layer, thereby restricting the variational search space. In contrast, the free ans\"atze assign an independent variational parameter to each local operator within every layer, substantially increasing the expressive power of the circuit.

\subsection{Free Dynamical Lie Algebras}

To analyze the expressive capabilities of these ansätze, we study their associated dynamical Lie algebras. These Lie algebras characterize the set of unitary transformations generated by the available local controls and therefore play a central role in our controllability analysis.

The free dynamical Lie algebras are obtained by treating the local constituents of the mixer and problem Hamiltonians as independent generators. Specifically, for the constrained mixer \(H_{CX,v}\) defined in \eqref{eq:CX_mixer} and the Flip-or-Stay \(\widehat{H}_{CX,v}\) defined in \eqref{eq:Bao_mixer}, we define

\begin{equation}
    \label{eq:freeDLAdefn}
    \mathfrak{g}_{\Gamma,\mathrm{free}}
    :=
    \left\langle
    \{\,iH_{CX,v},\,iZ_v \mid v\in V(\Gamma)\,\}
    \right\rangle_{\mathrm{Lie}},
    \qquad
    \widehat{\mathfrak{g}}_{\Gamma,\mathrm{free}}
    :=
    \left\langle
    \{\,i\widehat{H}_{CX,v},\,iZ_v \mid v\in V(\Gamma)\,\}
    \right\rangle_{\mathrm{Lie}}.
\end{equation}

Evaluating the double commutator of the local longitudinal operator $iZ_v$ with each local mixer term isolates its off-diagonal component, yielding the structural identities:
\begin{equation}
\label{eq:double_commutators}
\begin{aligned}
    -\frac{1}{4}\left[iZ_v, \left[iZ_v, iH_{CX,v}\right]\right] &= iH_{CX,v}, \\
    -\frac{1}{4}\left[iZ_v, \left[iZ_v, i\widehat{H}_{CX,v}\right]\right] &= iH_{CX,v}.
\end{aligned}
\end{equation}
Since $iZ_v \in \widehat{\mathfrak{g}}_{\Gamma,\mathrm{free}}$ and $i\widehat{H}_{CX,v} \in \widehat{\mathfrak{g}}_{\Gamma,\mathrm{free}}$ for every $v \in V$, the identity in \eqref{eq:double_commutators} directly yields the local operator containment:
\begin{equation}
\label{eq:local_containment}
    iH_{CX,v} \in \widehat{\mathfrak{g}}_{\Gamma,\mathrm{free}} \quad \text{for all } v \in V.
\end{equation}
Consequently, the extended free Lie algebra $\widehat{\mathfrak{g}}_{\Gamma,\mathrm{free}}$ contains the complete generating set of the Flip-or-Void free Lie algebra, establishing the algebra inclusion:
\begin{equation}
\label{eq:free_dla_containment}
    \mathfrak{g}_{\Gamma,\mathrm{free}} \subseteq \widehat{\mathfrak{g}}_{\Gamma,\mathrm{free}}.
\end{equation}

We now turn to characterizing the global structural properties of the free dynamical Lie algebras $\mathfrak{g}_{\Gamma,\mathrm{free}}$ and $\widehat{\mathfrak{g}}_{\Gamma,\mathrm{free}}$. By analyzing the common invariant subspaces under the action of the generators in Eq.~\eqref{eq:freeDLAdefn}, we obtain a fundamental block-diagonal embedding.

\begin{prop}
\label{prop:free_DLA_block_embedding}
    Let $\Gamma = (V,E)$ be a graph. Then the free dynamical Lie algebras $\mathfrak{g}_{\Gamma,\mathrm{free}},\widehat{\mathfrak{g}}_{\Gamma,\mathrm{free}}$ admit a block-diagonal embedding:
    \begin{equation}
        \label{eq:lie_algebra_decomposition}
      \mathfrak{g}_{\Gamma,\mathrm{free}} \subseteq   \widehat{\mathfrak{g}}_{\Gamma,\mathrm{free}} \subseteq \mathfrak{u}(W_\mathcal{F}) \oplus \bigoplus_{j=1}^k \mathfrak{u}(W_{\mathcal{O}_j}),
    \end{equation}
    where each $W_{\mathcal{O}_j}$ is the orbit Hilbert subspace defined by
    \begin{equation}
        \label{eq:orbit_subspace}
        W_{\mathcal{O}_j} := \operatorname{span} \{ \ket{x} \mid x \in \mathcal{O}_j \}
    \end{equation}
    (see \eqref{eq:orbit_decomposition}).
\end{prop}

The Lie algebra embedding in \eqref{eq:lie_algebra_decomposition} establishes an upper bound on the reachable subspace, confirming that the free DLA dynamics preserves the invariant subspace decomposition and leaves the feasible subspace $W_{\mathcal{F}}$ strictly decoupled from the individual orbit sectors $W_{\mathcal{O}_j}$.

From a quantum control perspective, the primary interest lies in the system's dynamics restricted to the invariant feasible subspace $W_{\mathcal{F}}$. 

Formally, let $\mathfrak{g}_{\Gamma,\mathrm{free}, \mathcal{F}}$ and $\widehat{\mathfrak{g}}_{\Gamma,\mathrm{free}, \mathcal{F}}$ denote the restrictions of the free dynamical Lie algebra $\mathfrak{g}_{\Gamma,\mathrm{free}}$ and the extended free dynamical Lie algebra $\widehat{\mathfrak{g}}_{\Gamma,\mathrm{free}}$ to $W_{\mathcal{F}}$, respectively:
\begin{equation}
\label{eq:dla_subspace_restrictions}
    \mathfrak{g}_{\Gamma,\mathrm{free}, \mathcal{F}} := \mathfrak{g}_{\Gamma,\mathrm{free}}\big|_{W_{\mathcal{F}}}, 
    \qquad 
    \widehat{\mathfrak{g}}_{\Gamma,\mathrm{free}, \mathcal{F}} := \widehat{\mathfrak{g}}_{\Gamma,\mathrm{free}}\big|_{W_{\mathcal{F}}}.
\end{equation}
The global containment in \eqref{eq:free_dla_containment} naturally induces a corresponding inclusion between the restricted algebras:
\begin{equation}
\label{eq:dla_restriction_containment}
    \mathfrak{g}_{\Gamma,\mathrm{free}, \mathcal{F}} \subseteq \widehat{\mathfrak{g}}_{\Gamma,\mathrm{free}, \mathcal{F}}.
\end{equation}

Our main result demonstrates that each of the restricted ansätze generates the complete unitary Lie algebra on $W_{\mathcal{F}}$, guaranteeing full quantum controllability on the feasible subspace.

\begin{thm}
\label{thm:ExtendedFreeLieAlgFullUnitary}
Let $\Gamma = (V, E)$ be a connected simple graph with $|V| \ge 2$. Then the restricted free dynamical Lie algebra $\mathfrak{g}_{\Gamma,\mathrm{free}, \mathcal{F}}$ and the extended free dynamical Lie algebra $\widehat{\mathfrak{g}}_{\Gamma,\mathrm{free}, \mathcal{F}}$ are both isomorphic to the full unitary Lie algebra on the feasible subspace $W_{\mathcal{F}}$:
\begin{equation}
\label{eq:algebra_isomorphism}
    \mathfrak{g}_{\Gamma,\mathrm{free}, \mathcal{F}}=\widehat{\mathfrak{g}}_{\Gamma,\mathrm{free}, \mathcal{F}} \cong \mathfrak{u}(W_{\mathcal{F}}).
\end{equation}
\end{thm}

\begin{rmk}
\label{rmk:generator_reduction}
The required generating sets for both free dynamical Lie algebras can be significantly reduced. Because single-qubit Pauli operators $iZ_w$ commute with local mixer terms whose flipped vertex is different ($[iZ_w, H_{CX,v}] = 0$ and $[iZ_w, \widehat{H}_{CX,v}] = 0$ for all $w \neq v$), taking the commutator with $iZ_v$ isolates the local contribution at vertex $v$:
\begin{equation}
\label{eq:LocalMixerIsolation}
    [iZ_v, H_{CX}] = [iZ_v, H_{CX,v}], \qquad [iZ_v, \widehat{H}_{CX}] = [iZ_v, \widehat{H}_{CX,v}].
\end{equation}
Furthermore, because $iZ_v$ commutes with any diagonal projector terms in the extended mixer $\widehat{H}_{CX,v}$, taking a double commutator with $iZ_v$ eliminates these diagonal components and directly extracts the standard bit-flip mixer $iH_{CX,v}$:
\begin{equation}
\label{eq:LocalMixerDoubleCommutator}
    -\frac{1}{4}\bigl[iZ_v, [iZ_v, iH_{CX}]\bigr] = iH_{CX,v}, \qquad 
    -\frac{1}{4}\bigl[iZ_v, [iZ_v, i\widehat{H}_{CX}]\bigr] = iH_{CX,v}.
\end{equation}
This double-commutator identity establishes that the Lie algebra generated by 
\[
\{ i\widehat{H}_{CX}\big|_{W_{\mathcal{F}}} \} \cup \{ iZ_v\big|_{W_{\mathcal{F}}} \mid v \in V(\Gamma) \}
\]
contains all unextended local generators $iH_{CX,v}$, and therefore contains $\mathfrak{g}_{\Gamma,\mathrm{free},\mathcal{F}}$. Applying Theorem~\ref{thm:ExtendedFreeLieAlgFullUnitary}, we obtain the chain of inclusions:
\begin{equation}
    \mathfrak{u}(W_{\mathcal{F}}) = \mathfrak{g}_{\Gamma,\mathrm{free},\mathcal{F}} \subseteq \bigl\langle \{ i\widehat{H}_{CX}\big|_{W_{\mathcal{F}}} \} \cup \{ iZ_v\big|_{W_{\mathcal{F}}} \mid v \in V(\Gamma) \} \bigr\rangle_{\mathrm{Lie}} \subseteq \widehat{\mathfrak{g}}_{\Gamma,\mathrm{free},\mathcal{F}} \subseteq \mathfrak{u}(W_{\mathcal{F}}),
\end{equation}
which forces all inclusions to be equalities. Consequently, both free DLAs on $W_{\mathcal{F}}$ are generated solely by a single global mixer together with local $Z$-rotations:
\begin{equation}
\label{eq:ReducedGeneratorsFreeDLAs}
\begin{aligned}
    \mathfrak{g}_{\Gamma,\mathrm{free},\mathcal{F}} &= \bigl\langle \{ iH_{CX}\big|_{W_{\mathcal{F}}} \} \cup \{ iZ_v\big|_{W_{\mathcal{F}}} \mid v \in V(\Gamma) \} \bigr\rangle_{\mathrm{Lie}}, \\
    \widehat{\mathfrak{g}}_{\Gamma,\mathrm{free},\mathcal{F}} &= \bigl\langle \{ i\widehat{H}_{CX}\big|_{W_{\mathcal{F}}} \} \cup \{ iZ_v\big|_{W_{\mathcal{F}}} \mid v \in V(\Gamma) \} \bigr\rangle_{\mathrm{Lie}}.
\end{aligned}
\end{equation}
Thus, a global mixer and local \(Z\)-controls generate the same restricted DLA. Independent local mixer controls are unnecessary for the controllability conclusion.
\end{rmk}

Remark~\ref{rmk:generator_reduction} shows that full state
controllability on $W_{\mathcal F}$ does not require
independently parameterized local mixers: a single global
mixer together with individual single-qubit controls
$\left\{iZ_v\big|_{W_{\mathcal{F}}}\right\}_{v\in V}$ suffices. In standard unweighted MIS-QAOA, however, the problem Hamiltonian $H_P=\sum\limits_{v\in V}Z_v$ acts uniformly across all qubits. On symmetric graphs, this uniform weighting prevents the isolation of individual $iZ_v$ terms, trapping the system within low-dimensional invariant subspaces.

Assigning positive, pairwise distinct vertex weights breaks
this structural degeneracy. The inhomogeneous cost Hamiltonian
\begin{equation}
    \label{eq:weighted_problem_Ham}
    H_{P,w}:=\sum_{v\in V}w_vZ_v
\end{equation}
acts as a symmetry-breaking operator, allowing the individual
local mixer generators to be isolated from the global mixer.
The controllability argument underlying
Theorem~\ref{thm:ExtendedFreeLieAlgFullUnitary}
then extends to this setting, establishing that, for every
connected graph with at least two vertices, scalar-parameter
weighted QAOA with either mixer is fully controllable on
$W_{\mathcal F}$ while using just one mixer parameter $\beta$
and one cost parameter $\gamma$ per layer
(Proposition~\ref{prop:WeightedMISStandardDLA}).

\begin{prop}
\label{prop:WeightedMISStandardDLA}
Let $\Gamma = (V,E)$ be a connected simple graph with $|V| \ge 2$, and let 
\[
H_{P,w} := \sum_{v\in V} w_v Z_v
\]
(with positive, pairwise distinct weights $w_v > 0$) be the weighted problem Hamiltonian. Then the standard scalar-parameter dynamical Lie algebras driven by either the Flip-or-Void mixer $H_{\mathrm{CX}}$ or the Flip-or-Stay mixer $\widehat{H}_{\mathrm{CX}}$ together with $H_{P,w}$, restricted to the feasible subspace $W_{\mathcal{F}}$, coincide with the full unitary Lie algebra:
\begin{equation}
\label{eq:weighted_dla_equality}
    \left\langle i H_{\mathrm{CX}}\big|_{W_{\mathcal{F}}}, i H_{P,w}\big|_{W_{\mathcal{F}}} \right\rangle_{\mathrm{Lie}} 
    =
    \left\langle i \widehat{H}_{\mathrm{CX}}\big|_{W_{\mathcal{F}}}, i H_{P,w}\big|_{W_{\mathcal{F}}} \right\rangle_{\mathrm{Lie}} 
    =
    \mathfrak{u}(W_{\mathcal{F}}).
\end{equation}
\end{prop}

Thus, standard QAOAs with pairwise distinct vertex weights achieve full state controllability on the feasible subspace $W_{\mathcal{F}}$.

The proof of this result is given in Appendix~B.
The weights are fixed coefficients of the problem Hamiltonian,
so the ansatz retains only two variational parameters per layer.
The unweighted case does not satisfy this distinctness condition
and can exhibit substantially smaller reachable state spaces (see Subsection~\ref{subsec:explicit_mis_separation}).

\begin{rmk}
\label{rmk:perturbative_weighting_mis}
Proposition~\ref{prop:WeightedMISStandardDLA} yields a straightforward mechanism for transforming an unweighted Maximum Independent Set  instance into a weighted problem with full state controllability, without altering the optimal solution set. Label the vertices $V = \{1, \dots, n\}$ and assign vertex weights
\begin{equation}
\label{eq:binary_weight_perturbation}
w_v \coloneqq 1 + \varepsilon 2^{-v}, \qquad 0 < \varepsilon < 1.
\end{equation}
As $\sum\limits_{v=1}^n \varepsilon 2^{-v} = \varepsilon (1 - 2^{-n}) < 1$, the cumulative weight change across all vertices is strictly less than $1$. Thus, for any two independent sets $x, y \in \mathcal{F}$ with $|x| > |y|$ (so that $|x| \ge |y| + 1$), the difference in their total weights satisfies
\begin{equation}
\label{eq:weight_difference_bound}
\sum_{v \in x} w_v - \sum_{v \in y} w_v = |x| - |y| + \varepsilon \left( \sum_{v \in x} 2^{-v} - \sum_{v \in y} 2^{-v} \right) \ge 1 - \varepsilon \left(1 - 2^{-n}\right) > 0.
\end{equation}
In particular, maximum-weight independent set with respect to $w$ is guaranteed to be a maximum-cardinality independent set of the original unweighted graph $\Gamma$. Furthermore, because distinct subsets of $\{1, \dots, n\}$ have unique finite binary sum representations $\sum\limits_{v \in x} 2^{-v}$, this perturbation ensures pairwise distinct vertex weights $w_v$ and breaks all energy degeneracies within each cardinality sector. Proposition~\ref{prop:WeightedMISStandardDLA} therefore yields a two-parameter-per-layer, fully controllable architecture whose cost function resolves degeneracies while preserving the MIS optimal cardinality.
\end{rmk}

\subsection{Subspace State Controllability}
\label{subsec:controllability}

A direct consequence of Theorem~\ref{thm:ExtendedFreeLieAlgFullUnitary} is the existence of finite-depth parameterized quantum circuits (PQCs) capable of universal state preparation within the feasible subspace $W_{\mathcal{F}}$.

\begin{cor}
\label{cor:controllability}
Let $\Gamma = (V, E)$ be a connected simple graph with $|V| \ge 2$. For any pair of normalized state vectors $\ket{h}, \ket{t} \in W_{\mathcal{F}}$, there exists a finite circuit depth $L \in \mathbb{Z}_{>0}$, a selection of Hamiltonians $H_\ell$ from the generating sets of $\mathfrak{g}_{\Gamma,\mathrm{free}}$ or $\widehat{\mathfrak{g}}_{\Gamma,\mathrm{free}}$ in \eqref{eq:ReducedGeneratorsFreeDLAs}, and a parameter vector $\boldsymbol{\theta} = (\theta_1, \theta_2, \dots, \theta_L) \in \mathbb{R}^L$ such that the parameterized unitary
\begin{equation}
\label{eq:PQC_ansatz}
    U(\boldsymbol{\theta}) = \prod_{\ell=1}^{L} e^{-i \theta_\ell H_\ell}
\end{equation}
achieves exact state transfer from source state $\ket{h}$ to target state $\ket{t}$:
\begin{equation}
\label{eq:exact_state_prep}
    U(\boldsymbol{\theta})\ket{h} = \ket{t}.
\end{equation}
\end{cor}

In the language of quantum control theory, Corollary~\ref{cor:controllability} asserts that both the Flip-or-Void and Flip-or-Stay free ansätze exhibit \emph{subspace state controllability} on $W_{\mathcal{F}}$, ensuring that no feasible quantum state is dynamically inaccessible.

\subsection{From Maximal Expressivity to Barren Plateaus}
\label{subsec:BPforFree}

Another direct application of our Lie-algebraic framework lies in characterizing trainability landscapes and quantifying the occurrence of \emph{barren plateaus} in constraint-preserving variational quantum optimization.

Barren plateaus manifest as an exponential suppression of cost function gradients with respect to ansatz parameters—a concentration of measure that renders classical gradient-based optimization routines ineffective \cite{mcclean2018barren,larocca2025barren, FHCKYHSP}. A standard proxy for quantifying this landscape flattening is the variance of the cost function partial derivatives over parameter space. Following \cite{RBSKMLC}, we formalize this phenomenon as follows.

\begin{defn}
\label{defn:BP}
Let $\rho = \ket{\xi}\bra{\xi}$ denote a pure initial state. A parameterized quantum circuit architecture exhibits a \emph{barren plateau} if the variance of the partial derivatives of the loss function $\ell_{\boldsymbol{\theta}}(\rho, H_P)$ decays exponentially with system size $n$:
\begin{equation}
\label{eq:BPdefn}
    \operatorname{Var}_{\boldsymbol{\theta}}\!\bigl[\partial_k \ell_{\boldsymbol{\theta}}(\rho, H_P)\bigr] \in \mathcal{O}\left(\frac{1}{b^n}\right)
\end{equation}
for some dimension-independent constant $b > 1$.
\end{defn}

In our restricted ma-QAOA setting, the loss function is defined as the expectation value of the problem Hamiltonian $H_P$ restricted to the invariant feasible subspace $W_{\mathcal{F}}$:
\begin{equation}
\label{eq:MIS_loss_function}
    \ell_{\boldsymbol{\theta}}\!\left(\rho, H_P\big|_{W_{\mathcal{F}}}\right) := \bra{\psi(\boldsymbol{\theta})} H_P\big|_{W_{\mathcal{F}}} \ket{\psi(\boldsymbol{\theta})} = \operatorname{Tr}\!\left[ U(\boldsymbol{\theta}) \, \rho \, U^\dagger(\boldsymbol{\theta}) \, H_P\big|_{W_{\mathcal{F}}} \right],
\end{equation}
where $U(\boldsymbol{\theta})$ is the alternating unitary ansatz generated by the free mixer architectures (see \eqref{eq:PQC_ansatz}).

Using the full-controllability result
$\mathfrak{g}_{\Gamma,\mathrm{free},\mathcal{F}}
\cong \mathfrak{u}(W_{\mathcal{F}})$
from Theorem~\ref{thm:ExtendedFreeLieAlgFullUnitary},
we derive an exact expression for the loss variance under
the unitary $2$-design assumption and establish asymptotic
bounds in terms of the feasible-subspace dimension
$|\mathcal{F}|$ and the independence number $\alpha(\Gamma)$.
These bounds quantify loss concentration in deep
ma-QAOA circuits.

\begin{thm}
\label{thm:BP_mitigation}
Consider either of the two multi-angle QAOA ans\"atze associated with the MIS problem on a connected simple graph $\Gamma = (V,E)$ with $|V| \ge 2$ vertices, generated by the independently parameterized local control operators
$\{H_{\mathrm{CX},v}, Z_v \mid v \in V\}$
or
$\{\widehat{H}_{\mathrm{CX},v}, Z_v \mid v \in V\}$.
Let $d := |\mathcal{F}|$ denote the dimension of the feasible subspace $W_{\mathcal{F}}$, and let $\rho$ be any fixed pure state supported on $W_{\mathcal{F}}$. Suppose that the circuit ensemble forms an exact unitary $2$-design on its dynamical group $G \cong \mathrm{U}(W_{\mathcal{F}})$.

If $S$ is an $\mathcal{F}$-valued random variable uniformly distributed on $\mathcal{F}$, then the variance of the loss function is given exactly by
\begin{equation}
\label{eq:HaarVarianceExactIdentity}
    \operatorname{Var}_{\boldsymbol{\theta}}
    \left[
        \ell_{\boldsymbol{\theta}}
        \left(
            \rho,
            H_P\big|_{W_{\mathcal{F}}}
        \right)
    \right] = \frac{4\,\operatorname{Var}_{S \sim \operatorname{Unif}(\mathcal{F})}(|S|)}{d+1}.
\end{equation}
In particular, this variance exhibits the asymptotic scaling
\begin{equation}
\label{eq:HaarVarianceAsymptoticScaling}
    \operatorname{Var}_{\boldsymbol{\theta}}
    \left[
        \ell_{\boldsymbol{\theta}}
        \left(
            \rho,
            H_P\big|_{W_{\mathcal{F}}}
        \right)
    \right]
    \in
    \Omega\left(\frac{1}{d^2}\right)
    \cap
    \mathcal{O}\left(
        \frac{\alpha(\Gamma)^2}{d}
    \right).
\end{equation}
\end{thm}

Together, Theorems~\ref{thm:ExtendedFreeLieAlgFullUnitary} and \ref{thm:BP_mitigation} highlight a fundamental expressivity--trainability trade-off in feasibility-preserving QAOA architectures. Because the dynamical Lie algebra spans the full unitary algebra $\mathfrak{u}(W_{\mathcal{F}})$, the circuit is maximally expressive and capable of uniformly exploring the feasible subspace. However, whenever the independent set count $|\mathcal{F}|$ grows exponentially with $n$, the $  \mathcal{O}\left(
        \frac{\alpha(\Gamma)^2}{d}\right)$ upper bound guarantees an exponential suppression of landscape variance under the stated ensemble assumption.

These variance bounds yield immediate implications for algorithmic trainability via the concentration of measure framework established in~\cite[Theorem~2, first statement]{arrasmith2022equivalence}, which formalizes the rigorous equivalence between exponential concentration of the cost function and the emergence of barren plateaus.

\begin{cor}
\label{cor:BP_consequences}
Under the assumptions of Theorem~\ref{thm:BP_mitigation}, consider the depth-$p$ Flip-or-Stay ma-QAOA ansatz driven by the local generators $\{\widehat{H}_{\mathrm{CX},v}, Z_v \mid v \in V\}$, where the variational parameters $\boldsymbol{\theta}$ are sampled independently and uniformly from $[0, 2\pi)^{2np}$. 

As the local generators are involutions, i.e.
\begin{equation}
\label{eq:generator_involutions}
    \widehat{H}_{\mathrm{CX},v}^{\,2} = Z_v^{\,2} = I \qquad \forall v \in V,
\end{equation}
they fulfill the generator condition required by~\cite[Theorem~2]{arrasmith2022equivalence}. Therefore, whenever the feasible subspace dimension $d = |\mathcal{F}|$ grows exponentially with the graph order $n = |V|$, the ansatz exhibits a barren plateau in the sense of Definition~\ref{defn:BP}.
\end{cor}

\section{Standard QAOA Ansätze}

We turn our attention to the standard Quantum Approximate Optimization Algorithm formulations, which restrict the optimization trajectory by employing the alternating layer structure defined in \eqref{eq:qaoa_chain}. In this setup, rather than using decoupled terms, the collective mixer and problem Hamiltonians are driven by a single pair of scalar parameters $(\beta, \gamma)$ per layer. 

\subsection{Standard Dynamical Lie Algebras}
\label{subsec:standard_dlas}

In contrast to free dynamical Lie algebras where local mixer terms are independently parameterized, the \emph{standard dynamical Lie algebra}  of a QAOA instance is generated by the global mixer and problem Hamiltonians. For the Flip-or-Void and Flip-or-Stay ansätze, these real Lie algebras are defined respectively as:
\begin{equation}
\label{eq:stdDLAdefn}
    \mathfrak{g}_{\Gamma,\mathrm{std}} := \langle iH_{CX},\, iH_P \rangle_{\mathrm{Lie}},
    \qquad
    \widehat{\mathfrak{g}}_{\Gamma,\mathrm{std}} := \langle i\widehat{H}_{CX},\, iH_P \rangle_{\mathrm{Lie}}.
\end{equation}

Analogous to the free algebra inclusions in \eqref{eq:free_dla_containment} and \eqref{eq:dla_restriction_containment}, the global nested commutator identity
\begin{equation}
\label{eq:global_containment}
    [iH_P, [iH_P, i\widehat{H}_{CX}]] = -4 \, iH_{CX}
\end{equation}
ensures that $iH_{CX} \in \widehat{\mathfrak{g}}_{\Gamma,\mathrm{std}}$. This double commutator directly establishes algebra containments both globally and on the feasible subspace $W_{\mathcal{F}}$:
\begin{equation}
\label{eq:std_dla_containment}
    \mathfrak{g}_{\Gamma,\mathrm{std}} \subseteq \widehat{\mathfrak{g}}_{\Gamma,\mathrm{std}}
    \qquad \text{and} \qquad
    \mathfrak{g}_{\Gamma,\mathrm{std},\mathcal{F}} \subseteq \widehat{\mathfrak{g}}_{\Gamma,\mathrm{std},\mathcal{F}},
\end{equation}
where $\mathfrak{g}_{\Gamma,\mathrm{std},\mathcal{F}} := \mathfrak{g}_{\Gamma,\mathrm{std}}\big|_{W_{\mathcal{F}}}$ and $\widehat{\mathfrak{g}}_{\Gamma,\mathrm{std},\mathcal{F}} := \widehat{\mathfrak{g}}_{\Gamma,\mathrm{std}}\big|_{W_{\mathcal{F}}}$ denote the respective subspace restrictions.

\begin{rmk}
\label{rmk:stay_dla_projector_containment}
We note that identities~\eqref{eq:Bao_mixer} and~\eqref{eq:global_containment} together imply that the dynamical Lie algebra of the Flip-or-Stay ansatz, $\widehat{\mathfrak{g}}_{\Gamma,\mathrm{std},\mathcal{F}}$, contains the diagonal constraint operator $i\sum\limits_{v \in V}(I - \mathcal{P}_{v})$ (see~\eqref{eq:CX_mixer}).
\end{rmk}

\begin{rmk}
\label{rmk:dla_invariance}
By Remark~\ref{rmk:symmetry_properties}, the generators of $\mathfrak{g}_{\Gamma,\mathrm{std}}$ and $\widehat{\mathfrak{g}}_{\Gamma,\mathrm{std}}$ commute with the linear representation of the graph automorphism group $\operatorname{Aut}(\Gamma)$. Because $\operatorname{Aut}(\Gamma)$ acts on operator space via unitary conjugation, it preserves the Lie bracket structure, rendering the resulting DLAs explicitly invariant under $\operatorname{Aut}(\Gamma)$. Whenever $\operatorname{Aut}(\Gamma)$ is nontrivial, these symmetry-enforced commutation constraints restrict operator generation, forcing strict inclusions relative to their independently controlled free counterparts:
\begin{equation}
\label{eq:strict_containment}
    \mathfrak{g}_{\Gamma,\mathrm{std}} \subsetneq \mathfrak{g}_{\Gamma,\mathrm{free}},
    \qquad
    \widehat{\mathfrak{g}}_{\Gamma,\mathrm{std}} \subsetneq \widehat{\mathfrak{g}}_{\Gamma,\mathrm{free}}.
\end{equation}
\end{rmk}

This observation motivates  investigating the dimension and expressivity of standard DLAs on asymmetric graphs ($\operatorname{Aut}(\Gamma) = \{\operatorname{id}\}$), where spatial automorphisms induce no nontrivial symmetry-based decomposition of the Hilbert space (though non-spatial dynamical symmetries may still persist).

\subsection{Dynamical Group Orbits and Exact Ground-State Reachability}
\label{subsec:group_orbits_reachability}

Next we establish reachability and state-preparation guarantees for standard QAOA implementations by analyzing the orbits generated by the action of the dynamical Lie groups corresponding to the standard Lie algebras $\mathfrak{g}_{\Gamma,\mathrm{std},\mathcal{F}}$ and $\widehat{\mathfrak{g}}_{\Gamma,\mathrm{std},\mathcal{F}}$ on the feasible subspace $W_{\mathcal{F}}$.

Recall that for a subgroup $G \subseteq \mathrm{U}(W_{\mathcal{F}})$ acting on $W_{\mathcal{F}}$, the \emph{group orbit}  of a state vector $\ket{\psi} \in W_{\mathcal{F}}$ is defined as the set of all quantum states reachable under the action of $G$: 
\begin{equation}
\label{eq:dynamical_group_orbit_def}
    G \cdot \ket{\psi} := \left\{ U \ket{\psi} \;\middle|\; U \in G \right\} \subseteq W_{\mathcal{F}}.
\end{equation}

\begin{thm}
\label{thm:equal_standard_orbits}
Let $n \ge 1$, and let
\begin{equation}
\label{eq:standard_dynamical_groups}
\begin{aligned}
G_{\Gamma,\mathrm{std}}
&:=
\left\langle
    \exp(X) : X \in \mathfrak{g}_{\Gamma,\mathrm{std},\mathcal F}
\right\rangle, \\
\widehat{G}_{\Gamma,\mathrm{std}}
&:=
\left\langle
    \exp(X) : X \in \widehat{\mathfrak{g}}_{\Gamma,\mathrm{std},\mathcal F}
\right\rangle
\end{aligned}
\end{equation}
be the connected analytic dynamical Lie groups of the standard Flip-or-Void and Flip-or-Stay ans\"atze, respectively, acting on $W_{\mathcal{F}}$. Let $\ket{\zeta_\Gamma}$ denote the unique ground state of $-H_{CX}\big|_{W_{\mathcal{F}}}$. Then, the mixer ground states lie within the corresponding dynamical group orbits generated from the computational vacuum state $\ket{0}^n \in W_{\mathcal{F}}$:
\begin{equation}
\label{eq:equal_standard_vector_orbits}
\begin{aligned}
    \ket{\zeta_\Gamma}
    &\in
    G_{\Gamma,\mathrm{std}} \cdot \ket{0}^n, \\
    \ket{\xi_{\mathcal{F}}}
    &\in
    \widehat{G}_{\Gamma,\mathrm{std}} \cdot \ket{0}^n.
\end{aligned}
\end{equation}
\end{thm}

\begin{rmk}
\label{rmk:vacuum_init_practicality}
A major practical bottleneck in constrained QAOA is state initialization: constructing the uniform feasible state $\ket{\xi_{\mathcal{F}}}$ for large graphs $\Gamma$ typically incurs significant circuit overhead (see \cite{wild2021quantum, wild2021quantumpra}). In general, determining and preparing the mixer Perron state (e.g. $\ket{\zeta_\Gamma}$) can be computationally demanding. Theorem~\ref{thm:equal_standard_orbits} shows that vacuum initialization has the same reachable set when depth is unrestricted. Therefore, one can instead initialize QAOA in the computational vacuum $\ket{0}^n$ with zero state-preparation depth without restricting the accessible state space.
\end{rmk}

Theorem~\ref{thm:equal_standard_orbits} establishes that
$\ket{0}^n$ lies on the same dynamical orbit as
$\ket{\zeta_\Gamma}$ for the Flip-or-Void ansatz and
$\ket{\xi_{\mathcal F}}$ for the Flip-or-Stay ansatz.
Consequently, vacuum initialization yields the same
reachable states as initialization in the respective
mixer ground state, when circuit depth is unrestricted.
For each ansatz, the convergence theorem of~\cite{BKZS}
implies that the closure of this orbit intersects the
optimal-solution subspace. Closedness of the corresponding
(projective) orbit ensures that this intersection is attained,
yielding exact optimal-state preparation by a finite-depth
QAOA circuit, as stated below.

\begin{thm}
\label{thm:qaoa_perron_frobenius}
Let $W_{P,\min} \subseteq W_{\mathcal F}$ denote the ground-state subspace of the problem Hamiltonian $H_P\big|_{W_{\mathcal F}}$. 

\begin{enumerate}
    \item \textbf{Flip-or-Stay ansatz:} for either initial state $\ket{\phi} \in \left\{ \ket{\xi_{\mathcal F}}, \ket{0}^n \right\}$, there exist a finite depth $\widehat{p}_\phi \in \mathbb{Z}_{>0}$ and variational parameters $\widehat{\boldsymbol{\beta}}^\phi, \widehat{\boldsymbol{\gamma}}^\phi \in \mathbb{R}^{\widehat{p}_\phi}$ such that
    \begin{equation}
    \label{eq:qaoa_finite_exact_preparation_hat}
        \widehat{U}_{\widehat{p}_\phi} \!\left( \widehat{\boldsymbol{\beta}}^\phi, \widehat{\boldsymbol{\gamma}}^\phi \right) \ket{\phi} \in W_{P,\min}.
    \end{equation}

    \item \textbf{Flip-or-Void ansatz:} for either initial state $\ket{\phi} \in \left\{ \ket{\zeta_\Gamma}, \ket{0}^n \right\}$, there exist a finite depth $p_\phi \in \mathbb{Z}_{>0}$ and variational parameters $\boldsymbol{\beta}^\phi, \boldsymbol{\gamma}^\phi \in \mathbb{R}^{p_\phi}$ such that
    \begin{equation}
    \label{eq:qaoa_finite_exact_preparation_reg}
        U_{p_\phi} \!\left( \boldsymbol{\beta}^\phi, \boldsymbol{\gamma}^\phi \right) \ket{\phi} \in W_{P,\min}.
    \end{equation}
\end{enumerate}
\end{thm}

\begin{rmk}
\label{rmk:finite_depth_nonconstructive}
The finite-depth guarantee in
Theorem~\ref{thm:qaoa_perron_frobenius} is existential. The argument
does not provide an explicit upper bound on \(p_\phi\), nor does it
construct the corresponding optimal parameters.
\end{rmk}

\subsection{An Explicit Family Separating the Two MIS Dynamical Lie Algebras}
\label{subsec:explicit_mis_separation}

We next exhibit an infinite family of graphs demonstrating a sharp separation between the two collective MIS mixers paired with the same cost Hamiltonian. Specifically, when initialized in the computational vacuum state $\ket{0}^n$, the state orbit under the adjacency-mixer DLA retains a constant dimension, whereas that under the Laplacian-mixer DLA grows linearly with the size parameter of the family. Furthermore, in the latter case, we establish complete controllability on the corresponding invariant subspace.

Let $K_{a,r-a}$ denote the complete bipartite graph with bipartition classes of cardinalities $a$ and $r-a$. For an odd integer $r \ge 5$, set $m := \frac{r-1}{2}$ and consider the disjoint union graph $H_r$ and its complement $\Gamma_r$:
\begin{equation}
\label{eq:Hr_definition}
    H_r := \bigsqcup_{a=1}^{m} K_{a,r-a}, \qquad \Gamma_r := \overline{H}_r.
\end{equation}
By definition, $\Gamma_r$ shares the vertex set of $H_r$, with two vertices adjacent in $\Gamma_r$ if and only if they are non-adjacent in $H_r$. 

As the complement of a disjoint union is the graph join of the individual complements, and $\overline{K_{a,r-a}} \cong K_a \bigsqcup K_{r-a}$, $\Gamma_r$ can be expressed equivalently as
\begin{equation}
\label{eq:Gamma_r_join_description}
    \Gamma_r = \bigvee_{a=1}^{m} \left( K_a \bigsqcup K_{r-a} \right),
\end{equation}
where $\bigvee$ denotes the graph join. Furthermore, since each of the $m$ components of $H_r$ contains $a + (r-a) = r$ vertices, the total vertex count of $\Gamma_r$ is

\begin{equation}
\label{eq:nr_definition}
    n := |V(\Gamma_r)| = rm = \frac{r(r-1)}{2}.
\end{equation}

\begin{rmk}
\label{rmk:Gamma_r_feasible_cardinality}
The independent sets of $\Gamma_r = \overline{H}_r$ correspond precisely to the cliques of $H_r$. Because $H_r$ is a disjoint union of complete bipartite graphs $K_{a, r-a}$, it is triangle-free; hence, the clique number of $H_r$ is $\omega(H_r) = 2$, which implies that the independence number of $\Gamma_r$ is $\alpha(\Gamma_r) = 2$. Consequently, the feasible 
set $\mathcal{F}_{\Gamma_r}$ consists solely of the empty set, the $rm$ singletons, and the two-element sets corresponding to the edges of $H_r$.

The dimension  of the feasible subspace $W_{\mathcal{F}_r}$ is therefore given by
\begin{equation}
\label{eq:Gamma_r_feasible_cardinality}
   \dim(W_{\mathcal{F}_r})= |\mathcal{F}_{\Gamma_r}| 
    = 1 + |V(H_r)| + |E(H_r)| = 1 + rm + \sum_{a=1}^{m} a(r-a) 
    = \frac{r^3 + 6r^2 - 7r + 12}{12}.
\end{equation}
\end{rmk}

\begin{ex}
\label{ex:Gamma_5_feasible_set}
Consider the test graph $\Gamma_5 = \overline{H}_5$ corresponding to $r=5$ and $m=\frac{5-1}{2}=2$. The graph is defined by the join $\Gamma_5 = (K_1 \bigsqcup K_4) \vee (K_2 \bigsqcup K_3)$ and contains $n = 10$ vertices (see Figure~\ref{fig:graph_Gamma_5}). Evaluating \eqref{eq:Gamma_r_feasible_cardinality} for $r=5$ yields a feasible set cardinality of
\[
    |\mathcal{F}_{\Gamma_5}| = \frac{5^3 + 6\cdot5^2 - 7\cdot5 + 12}{12} = \frac{252}{12} = 21.
\]
These $21$ feasible basis states consist of the computational vacuum $\ket{0}^{\otimes10}$(corresponding to $\emptyset$), $10$ singletons $\ket{\{v\}}$, and $10$ pairs $\ket{\{v,w\}}$ corresponding to the edges $E(H_5) = E(K_{1,4}) \cup E(K_{2,3})$.
\end{ex}

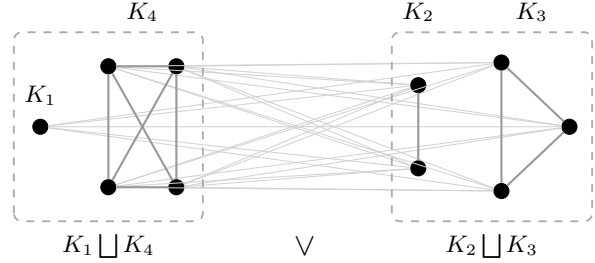
\begin{figure}[h!]
\centering
\begin{tikzpicture}[
    scale=1.0,
    every node/.style={
        circle,
        draw,
        fill=black,
        inner sep=2pt
    }
]


\node (a1) at (-0.25,0) {};

\node (b11) at (0.65, 0.80) {};
\node (b12) at (1.55, 0.80) {};
\node (b13) at (1.55,-0.80) {};
\node (b14) at (0.65,-0.80) {};


\node (a21) at (4.75, 0.55) {};
\node (a22) at (4.75,-0.55) {};

\node (b21) at (5.85, 0.85) {};
\node (b22) at (6.75, 0.00) {};
\node (b23) at (5.85,-0.85) {};


\foreach \u in {a1,b11,b12,b13,b14}{
    \foreach \v in {a21,a22,b21,b22,b23}{
        \draw[line width=0.35pt,gray!35] (\u)--(\v);
    }
}


\draw[thick,gray!80] (b11)--(b12);
\draw[thick,gray!80] (b12)--(b13);
\draw[thick,gray!80] (b13)--(b14);
\draw[thick,gray!80] (b14)--(b11);
\draw[thick,gray!80] (b11)--(b13);
\draw[thick,gray!80] (b12)--(b14);

\draw[thick,gray!80] (a21)--(a22);

\draw[thick,gray!80]
    (b21)--(b22)--(b23)--(b21);


\draw[
    dashed,
    rounded corners=3pt,
    gray!65,
    line width=0.7pt
]
    (-0.60,-1.25) rectangle (1.90,1.25);

\draw[
    dashed,
    rounded corners=3pt,
    gray!65,
    line width=0.7pt
]
    (4.40,-1.25) rectangle (7.05,1.25);


\node[
    draw=none,
    fill=none,
    font=\small
] at (-0.25,0.42) {$K_1$};

\node[
    draw=none,
    fill=none,
    font=\small
] at (1.10,1.52) {$K_4$};

\node[
    draw=none,
    fill=none,
    font=\small
] at (4.75,1.52) {$K_2$};

\node[
    draw=none,
    fill=none,
    font=\small
] at (6.25,1.52) {$K_3$};

\node[
    draw=none,
    fill=none,
    font=\small
] at (0.65,-1.58)
    {$K_1\bigsqcup K_4$};

\node[
    draw=none,
    fill=none,
    font=\large
] at (3.25,-1.58) {$\vee$};

\node[
    draw=none,
    fill=none,
    font=\small
] at (5.72,-1.58)
    {$K_2\bigsqcup K_3$};

\end{tikzpicture}

\caption{
The graph
\(\Gamma_5=\overline{K_{1,4}\bigsqcup K_{2,3}}
=(K_1\bigsqcup K_4)\vee(K_2\bigsqcup K_3)\).
The dashed boxes indicate the two join factors. The subgraphs
\(K_3\) and \(K_4\) are drawn as a triangle and a complete graph on
four vertices, respectively. Every vertex in one dashed box is
adjacent to every vertex in the other.
}
\label{fig:graph_Gamma_5}
\end{figure}

\begin{thm}
\label{thm:explicit_family_orbit_separation}
Let $r \ge 5$ be an odd integer and let
\begin{equation}
\label{eq:family_dynamical_groups}
\begin{aligned}
    G_{\Gamma_r}
    &:=
    \left\langle
        \exp(X) : X \in \mathfrak{g}_{\Gamma_r,\mathrm{std}}
    \right\rangle, \\
    \widehat{G}_{\Gamma_r}
    &:=
    \left\langle
        \exp(X) : X \in \widehat{\mathfrak{g}}_{\Gamma_r,\mathrm{std}}
    \right\rangle
\end{aligned}
\end{equation}
be the connected analytic dynamical Lie groups corresponding to the standard Flip-or-Void and Flip-or-Stay DLAs, respectively.
 The restriction of $\widehat{\mathfrak{g}}_{\Gamma_r,\mathrm{std}}$ to $\mathcal{W}_{\widehat{\mathfrak{g}}_{\Gamma_r,\mathrm{std},\ket{0}^n}}$, the minimal $\widehat{\mathfrak{g}}_{\Gamma_r,\mathrm{std}}$-invariant subspace containing $\ket{0}^n$, is the full unitary Lie algebra:
\begin{equation}
\label{eq:family_extended_DLA_exact}
    \widehat{\mathfrak{g}}_{\Gamma_r,\mathrm{std}}
\big|_{\mathcal{W}_{\widehat{\mathfrak{g}}_{\Gamma_r,\mathrm{std},\ket{0}^n}}}
    =
    \mathfrak{u}\left(\mathcal{W}_{\widehat{\mathfrak{g}}_{\Gamma_r,\mathrm{std},\ket{0}^n}}\right)
    \cong
    \mathfrak{u}\left(\frac{3r-1}{2}\right).
\end{equation}

The vacuum-state orbits satisfy
\begin{equation}
\label{eq:family_vector_orbits}
\begin{aligned}
    G_{\Gamma_r} \cdot \ket{0}^n
    &\subseteq
\mathbb{S}\left(\mathcal{W}_{\mathfrak{g}_{\Gamma_r,\mathrm{std},\ket{0}^n}}\right), \\
    \widehat{G}_{\Gamma_r} \cdot \ket{0}^n
    &=
\mathbb{S}\left(\mathcal{W}_{\widehat{\mathfrak{g}}_{\Gamma_r,\mathrm{std},\ket{0}^n}}\right),
\end{aligned}
\end{equation}
where $\mathbb{S}(V) := \left\{ \ket{\psi} \in V \;\middle|\; \|\psi\| = 1 \right\}$ denotes the unit sphere of a subspace $V$. In particular, the real orbit dimensions satisfy
\begin{equation}
\label{eq:family_orbit_dimension_gap}
\begin{aligned}
    \dim_{\mathbb{R}} \bigl( G_{\Gamma_r} \cdot \ket{0}^n \bigr)
    &\le 7, \\
    \dim_{\mathbb{R}} \bigl( \widehat{G}_{\Gamma_r} \cdot \ket{0}^n \bigr)
    &= 3r - 2,
\end{aligned}
\end{equation}
yielding an orbit dimension difference of
\begin{equation}
\label{eq:family_orbit_dimension_difference}
    \dim_{\mathbb{R}} \bigl( \widehat{G}_{\Gamma_r} \cdot \ket{0}^n \bigr) - \dim_{\mathbb{R}} \bigl( G_{\Gamma_r} \cdot \ket{0}^n \bigr)
    \ge
    3r - 9.
\end{equation}
\end{thm}

\begin{rmk}
\label{rmk:laplacian_subspace_scaling_bottleneck}
Every state of unit norm within the effective subspace $\mathcal{W}_{\widehat{\mathfrak{g}}_{\Gamma_r,\mathrm{std}},\ket{0}^{\otimes n}}$ is reachable by standard QAOA employing the Flip-or-Stay mixer and initialized in the vacuum state $\ket{0}^{\otimes n}$. The real dimension of this subspace is given by
\begin{equation}
\label{eq:effective_subspace_real_dim}
\dim_{\mathbb{R}}\!\left(\mathcal{W}_{\widehat{\mathfrak{g}}_{\Gamma_r,\mathrm{std}},\ket{0}^{\otimes n}}\right) = 3r - 1 = \frac{1 + 3\sqrt{1+8n}}{2} = \Theta(\sqrt{n}).
\end{equation}
(cf.~\eqref{eq:nr_definition} and \eqref{eq:family_extended_DLA_exact}). By contrast, the entire feasible subspace $W_{\mathcal{F}_{\Gamma_r}}$ has real dimension
\begin{equation}
\label{eq:feasible_subspace_real_dim}
    \dim_{\mathbb{R}}\!\left(W_{\mathcal{F}_{\Gamma_r}}\right)
    =2|\mathcal{F}_{\Gamma_r}|=\frac{r^3 + 6r^2 - 7r + 12}{6}
    =\Theta(n^{3/2})
\end{equation}
(cf.~\eqref{eq:Gamma_r_feasible_cardinality}). Thus, although the Flip-or-Stay mixer gives an unbounded increase in reachable orbit dimension, the ratio of its  effective subspace dimension to the full feasible space dimension is only $\Theta(n^{-1})$. Full controllability on the cyclic subspace $\mathcal{W}_{\widehat{\mathfrak{g}}_{\Gamma_r,\mathrm{std}},\ket{0}^n}$ therefore coexists with a substantial restriction on exploration of the feasible space.

For this family, either independent local controls
(Theorem~\ref{thm:ExtendedFreeLieAlgFullUnitary})
or positive, pairwise distinct vertex weights
(Proposition~\ref{prop:WeightedMISStandardDLA})
enlarge the dynamical Lie algebra to
$\mathfrak u(W_{\mathcal F_{\Gamma_r}})$.
The latter achieves this while retaining only two
variational parameters per layer, although it changes
the objective to weighted MIS.
\end{rmk}

The separation concerns reachable states and controllability; it does not establish an approximation-ratio advantage for MIS.
Indeed, $\alpha(\Gamma_r)=2$, so measuring any normalized
feasible state supported on the two-particle sector
returns a maximum independent set with probability one.
By Theorem~\ref{thm:qaoa_perron_frobenius},
both QAOAs can attain such a state from the vacuum
at some finite depth.


\section{Maximum Independent Set with the Grover Mixer QAOA}

In this section, we analyze the Maximum Independent Set problem under a distinct variant of the Quantum Approximate Optimization Algorithm: the Grover-Mixer QAOA (GM-QAOA), introduced in \cite{BE}. In contrast to the local, feasibility-preserving mixers $H_{CX}$ and $\widehat{H}_{CX}$ studied in preceding sections, GM-QAOA employs a highly nonlocal rank-one projection mixer that acts globally over the feasible subspace $W_{\mathcal{F}}$.

 The Grover mixer $G_M$ is defined as the negative orthogonal projector onto the uniform feasible state $\ket{\xi_{\mathcal{F}}}$ (see Eq. \eqref{eq:uniform_feasible_state}):
\begin{equation}
\label{eq:GroverMixer}
G_M:=-\ket{\xi_{\mathcal F}}\bra{\xi_{\mathcal F}}.
\end{equation}

Recall that the objective function \(F\) restricted to the feasible set \(\mathcal F\subseteq\mathbb B^n\) assumes \(\alpha(\Gamma)+1\) distinct values (see~\eqref{eq:objective_function}).  For each nonempty level set $\mathcal{F}_j$ of $\mathcal F$ (see Eq. \eqref{eq:MIS_feasible_set_j}), we define the normalized uniform level-set state:
\begin{equation}
\label{eq:LevelSetStates}
\ket{\xi_j}:=\frac{1}{\sqrt{|\mathcal F_j|}}
\sum_{x\in\mathcal F_j}
\ket{x}
\end{equation}

As established in \cite{TNB}, the GM-QAOA state evolution is strictly confined to the $\alpha(\Gamma)+1$-dimensional  subspace:

\begin{equation}
\label{eq:GMQAOASubspace}
W_{\mathrm{GM}}:=\operatorname{span}
\bigl\{
\ket{\xi_0},
\dots,
\ket{\xi_{\alpha(\Gamma)}}
\bigr\}
\subseteq
W_{\mathcal F}.
\end{equation}

By restricting both the mixer and the problem Hamiltonian to the subspace $W_{\rm GM}$, we obtain the associated dynamical Lie algebra
\begin{equation}
\label{eq:GroverMixerDLA}
    \mathfrak{g}_{\xi_{\mathcal{F}}}\big|_{W_{\rm GM}} := \left\langle i G_M\big|_{W_{\rm GM}}, \, i H_P\big|_{W_{\rm GM}} \right\rangle_{\mathrm{Lie}}.
\end{equation}

This algebra attains the maximal possible dimension on this invariant subspace. Specifically, \cite[Theorem~IV.1]{TNB} shows that
\begin{equation}
\label{eq:GroverMixerLieAlgebra}
    \mathfrak{g}_{\xi_{\mathcal{F}}}\big|_{W_{\rm GM}} \cong \mathfrak{u}(W_{\rm GM}) \cong \mathfrak{su}\bigl(\alpha(\Gamma) + 1\bigr) \oplus \mathfrak{u}(1).
\end{equation}

Consequently, GM-QAOA is fully controllable on the reduced subspace \(W_{\mathrm{GM}}\), whose dimension is equal to the number of distinct objective-function values attained on the feasible set.

For MIS, Theorem~V.1 of~\cite{TNB} gives an explicit
deep-circuit loss variance. The restriction $H_P|_{W_{\mathrm{GM}}}$ has the $d=\alpha(\Gamma)+1$ distinct eigenvalues
\[
\Lambda=\{n-2j\mid 0\leq j\leq\alpha(\Gamma)\}.
\]

If $\zeta_\Lambda$ is uniformly distributed on $\Lambda$, then
\[
\operatorname{Var}(\zeta_\Lambda)
=
\frac{1}{d^2}
\sum_{0\leq i<j\leq\alpha(\Gamma)}
\bigl((n-2i)-(n-2j)\bigr)^2
=
\frac{\alpha(\Gamma)(\alpha(\Gamma)+2)}{3}.
\]
Consequently, for a circuit ensemble whose restriction to \(W_{\rm GM}\) forms a unitary $2$-design,
\begin{equation}
    \label{eq:GM_QAOA_Var}
    \operatorname{Var}_{\boldsymbol\beta,\boldsymbol\gamma}
\!\left[\ell_{\boldsymbol\beta,\boldsymbol\gamma}(\rho,H_P)\right]
=
\frac{\operatorname{Var}(\zeta_\Lambda)}{d+1}
=
\frac{\alpha(\Gamma)}{3},
\qquad
\rho=|\xi_{\mathcal F}\rangle\langle\xi_{\mathcal F}|.
\end{equation}
The equality is exact whenever the circuit ensemble forms
a unitary $2$-design on $U(W_{\mathrm{GM}})$.

\begin{rmk}
    For the normalized observable $H_P/n$, the corresponding
variance is $\alpha(\Gamma)/(3n^2)$.
\end{rmk}

Table~\ref{tab:MIS_GM_QAOA_comparison} provides a comprehensive structural comparison between the free constrained MIS-QAOA ans\"atze explored in this work and GM-QAOA.

\begin{table*}[htbp]
\centering
    
\renewcommand{\arraystretch}{1.3}
\setlength{\tabcolsep}{7pt}
\begin{tabular}{p{0.22\textwidth} p{0.36\textwidth} p{0.36\textwidth}}
\hline\hline
\textbf{Feature} & \textbf{Free Constrained MIS-QAOA} & \textbf{Grover-Mixer QAOA (GM-QAOA)} \\
\hline
\textbf{Mixer Operator} & 
Local Flip-or-Void $H_{CX,v}$ or Flip-or-Stay $\widehat{H}_{CX,v}$ terms & 
Nonlocal rank-one projector $G_M = -\ket{\xi_{\mathcal{F}}}\!\bra{\xi_{\mathcal{F}}}$ \\

\textbf{Control Generators} & 
$\{i H_{CX,v}\big|_{W_{\mathcal F}}, i Z_v\big|_{W_{\mathcal F}} \mid v \in V\}$ or $\{i \widehat{H}_{CX,v}, i Z_v \mid v \in V\}$ & 
Global pair $\{i G_M\big|_{W_{\mathrm GM}}, i H_P\big|_{W_{\mathrm GM}}\}$ \\

\textbf{Mixer Locality} & 
$(\deg(v)+1)$-local (acts on $v$ and its open neighborhood $N(v)$) & 
Fully nonlocal (acts globally on $W_{\mathcal{F}}$) \\

\textbf{Dynamical Invariant Subspace} & 
Full feasible subspace $W_{\mathcal{F}} = \operatorname{span}\{\ket{x} : x \in \mathcal{F}\}$ & 
Objective-level subspace $W_{\mathrm{GM}} = \operatorname{span}\{\ket{\xi_0}, \dots, \ket{\xi_{\alpha(\Gamma)}}\} \subseteq W_{\mathcal{F}}$ \\

\textbf{Subspace Dimension} & 
$\dim(W_{\mathcal{F}}) = |\mathcal{F}|$ (typically exponential in $n$) & 
$\dim(W_{\mathrm{GM}}) = \alpha(\Gamma)+1 \le n + 1$ (at most linear in $n$) \\

\textbf{DLA} & 
$\mathfrak{g}_{\Gamma,\mathrm{free},\mathcal{F}} = \widehat{\mathfrak{g}}_{\Gamma,\mathrm{free},\mathcal{F}} = \mathfrak{u}(W_{\mathcal{F}})$ (Theorem~\ref{thm:ExtendedFreeLieAlgFullUnitary}) & 
$\mathfrak{g}_{\xi_{\mathcal{F}}}\big|_{W_{\mathrm{GM}}} = \mathfrak{u}(W_{\mathrm{GM}}) \cong \mathfrak{su}(m) \oplus \mathfrak{u}(1)$ \cite{TNB} \\

\textbf{DLA Dimension} & 
$\dim(\mathfrak{g}) = |\mathcal{F}|^2$ & 
$\dim(\mathfrak{g}) = (\alpha(\Gamma)+1)^2 \le (n+1)^2$ \\

\textbf{Controllability Scope} & 
Full state controllability throughout $W_{\mathcal{F}}$ & 
Exact state controllability restricted to $W_{\mathrm{GM}}$ \\

\textbf{Degeneracy Resolution} & 
Distinguishes and controls individual independent sets & 
Amplitudes remain equal within each objective-value level for uniform feasible initialization states $\ket{\xi_j}$ \\

\textbf{Loss variance under the 2-design assumption.} & 
Loss function variance is governed by $|\mathcal{F}|$ (typically exponentially small, but graph-dependent): $\Omega(1/|\mathcal{F}|^2)$ to $\mathcal{O}(\alpha(\Gamma)^2/|\mathcal{F}|)$ & 
Loss function variance is $\frac{\alpha(\Gamma)}{3}$\\

\textbf{Structural Trade-off} & 
Local, feasibility-preserving controls yield maximal expressivity & 
Global nonlocal mixer compresses dynamics into a polynomial-dimensional subspace \\
\hline\hline
\end{tabular}
\caption{Structural and algebraic comparison between the free constrained MIS-QAOA ans\"atze ($H_{CX}$ and $\widehat{H}_{CX}$) and the Grover-Mixer QAOA ansatz. Here, $m = \alpha(\Gamma) + 1 \le n + 1$ denotes the number of distinct objective energy levels attained on the feasible set $\mathcal{F}$. Statements for free-control ans\"atze assume a connected graph $\Gamma$ with $n \ge 2$, while the GM-QAOA column describes dynamics initialized in the uniform feasible state $\ket{\xi_{\mathcal{F}}}$.}
\label{tab:MIS_GM_QAOA_comparison}
\end{table*}

\section{Numerical Results}
\subsection{Numerical Dimensionality Analysis for standard QAOA}
\label{subsec:numeric_DLA_dim}
In contrast to free dynamical Lie algebras, which universally generate the full unitary Lie algebra $\mathfrak{u}(W_{\mathcal{F}})$ for any connected graph (Theorem~\ref{thm:ExtendedFreeLieAlgFullUnitary}), standard DLAs display non-universal, graph-dependent behavior that resists simple analytical characterization. To systematically explore this dependency, we performed numerical calculations of DLA dimensions across all non-isomorphic, connected asymmetric graphs on $6$ and $7$ vertices. The DLA calculation is done using the PennyLane library \cite{bergholm2018pennylane} with numerical tolerance for the linear independence of $1\mathrm{e}{-10}$.

Table~\ref{tab:SixNodesDimensions}  details the feasible subspace cardinalities $|\mathcal{F}|$, whether the restricted Lie algebras generate the full unitary Lie algebra $\mathfrak{u}(W_{\mathcal{F}})$, and the corresponding Lie algebra dimensions.

\begin{table}[H]
\centering
\begin{tabular}{|c|c|c|c|c|c|}

\hline
\textbf{Graph} $i$ & $|\mathcal{F}|$ & $\mathfrak{g}_{\Gamma,\mathrm{std},\mathcal{F}} = \mathfrak{u}(W_{\mathcal{F}})$ & $\widehat{\mathfrak{g}}_{\Gamma,\mathrm{std},\mathcal{F}} = \mathfrak{u}(W_{\mathcal{F}})$ & $\dim\!\left(\mathfrak{g}_{\Gamma,\mathrm{std},\mathcal{F}}\right)$ & $\dim\!\left(\widehat{\mathfrak{g}}_{\Gamma,\mathrm{std},\mathcal{F}}\right)$ \\
\hline
\hline
$1$ & $18$ & $\checkmark$ & $\checkmark$ & $324$ & $324$ \\
\hline
$2$ & $18$ & $\checkmark$ & $\checkmark$ & $324$ & $324$ \\
\hline
$3$ & $19$ & $\checkmark$ & $\checkmark$  & $361$ & $361$ \\
\hline
$4$ & $16$ & $\times$ & $\times$ & $225$ & $226$ \\
\hline
$5$ & $17$ & $\times$ & $\times$ & $256$ & $257$ \\
\hline
$6$ & $15$ & $\times$ & $\checkmark$ & $172$ & $225$ \\
\hline
$7$ & $15$ & $\checkmark$ & $\checkmark$ & $225$ & $225$ \\
\hline
$8$ & $14$ & $\checkmark$ & $\checkmark$ & $196$ & $196$ \\
\hline

\end{tabular}
\caption{Feasible configuration subset cardinalities $|\mathcal{F}|$, attainment of the full unitary Lie algebra $\mathfrak{u}(W_{\mathcal{F}})$, and dimensions of the restricted dynamical Lie algebras for standard QAOA implementations across all $8$ non-isomorphic, asymmetric graphs on $6$ vertices. 
}
\label{tab:SixNodesDimensions}
\end{table}

To quantify how closely these standard dynamical Lie algebras approach their maximal theoretical expressivity, namely, the full unitary Lie algebra $\mathfrak{u}(W_{\mathcal{F}})$, across a larger graph ensemble, we introduce a normalized dimension ratio, $\delta$. Since the restrictions of standard DLAs to $W_\mathcal{F}$ are fundamentally contained within their corresponding free counterparts (which span $\mathfrak{u}(W_{\mathcal{F}})$ with maximal dimension $|\mathcal{F}|^2$), the ratio is defined as:
\begin{equation}
\label{eq:dim_ratio}
\delta := \frac{\dim\left(\mathfrak{g}\big|_{W_\mathcal{F}}\right)}{|\mathcal{F}|^2}.
\end{equation}

A value of $\delta = 1$ indicates that the standard implementation generates the full unitary Lie algebra $\mathfrak{u}(W_{\mathcal{F}})$ on the feasible subspace, matching the expressive power of the free algebra.

Table~\ref{tab:SevenNodesDimensions} details the distribution of this normalized ratio across all 144 connected, non-isomorphic, asymmetric graphs on $7$ vertices, illustrating the performance gap between the two mixer choices.

\begin{table}[H]
\centering
\begin{tabular}{|c|*{9}{c|}}
\hline
$\delta/\mathfrak{g}$
& \(\mathbf{<0.3}\)
& \(\mathbf{[0.3,0.4)}\)
& \(\mathbf{[0.4,0.5)}\)
& \(\mathbf{[0.5,0.6)}\)
& \(\mathbf{[0.6,0.7)}\)
& \(\mathbf{[0.7,0.8)}\)
& \(\mathbf{[0.8,0.9)}\)
& \(\mathbf{[0.9,1)}\)
& \(\mathbf{1}\)
\\ \hline

$\widehat{\mathfrak{g}}_{\mathrm{std}}$
& 0 & 0 & 0 & 0 & 0 & 1 & 12 & 31 & 100
\\ \hline

$\mathfrak{g}_{\mathrm{std}}$
& 1 & 0 & 1 & 0 & 1 & 7 & 33 & 55 & 46
\\ \hline
\end{tabular}

\caption{Distribution of the normalized dimension ratio $\delta$ for standard QAOA implementations across all 144 connected, non-isomorphic, asymmetric graphs on 7 vertices. The counts indicate the number of graph configurations falling within each interval for the Flip-or-Void mixer $\mathfrak{g}_{\Gamma,\mathrm{std}}$ and the Flip-or-Stay variant $\widehat{\mathfrak{g}}_{\Gamma,\mathrm{std}}$.}
\label{tab:SevenNodesDimensions}
\end{table}

Table~\ref{tab:SevenNodesDimensions_diff} details the distribution of the normalized relative dimension gap $\Delta$ between the restricted standard DLAs $\widehat{\mathfrak{g}}_{\mathrm{std}}\big|_{W_\mathcal{F}}$ (Flip-or-Stay) and $\mathfrak{g}_{\mathrm{std}}\big|_{W_\mathcal{F}}$ (Flip-or-Void) across all $144$ connected, non-isomorphic, asymmetric graphs on $7$ vertices:

\begin{equation}
\label{eq:dla_relative_gap}
    \Delta = \frac{\dim\left(\widehat{\mathfrak{g}}_{\mathrm{std}}\big|_{W_\mathcal{F}}\right) - \dim\left(\mathfrak{g}_{\mathrm{std}}\big|_{W_\mathcal{F}}\right)}{\dim\left(\widehat{\mathfrak{g}}_{\mathrm{std}}\big|_{W_\mathcal{F}}\right)}.
\end{equation}
A value \(\Delta=0\) implies equality of the restricted DLAs and hence of their connected dynamical groups. Their fixed-depth circuit families and parameter landscapes may nevertheless differ. Conversely, $\Delta > 0$ quantifies the gain in DLA dimension and, potentially, expressivity achieved by the Flip-or-Stay mixer.

\begin{table}[H]
\centering
\begin{tabular}{|c|llllllllc|}

\hline
$\Delta$& \multicolumn{1}{c|}{$\textbf{0.0\%}$} & \multicolumn{1}{c|}{$\textbf{(0.0\%, 5.0\%{]}}$} &  \multicolumn{1}{c|}{$\textbf{(5.0\%, 10.0\%{]}}$} &  \multicolumn{1}{c|}{$\textbf{(10.0\%, 20.0\%{]}}$} & \multicolumn{1}{c|}{$\textbf{(20.0\%, 30.0\%{]}}$} & \multicolumn{1}{c|}{$\textbf{(30.0\%, 65.0\%{]}}$} 
\\ \hline

Number of graphs
& \multicolumn{1}{c|}{52}    & \multicolumn{1}{c|}{35}                    & \multicolumn{1}{c|}{37}                    & \multicolumn{1}{c|}{11}                     & \multicolumn{1}{c|}{6}         & \multicolumn{1}{c|}{3}                                    \\ \hline
\end{tabular}
\caption{Distribution of the relative DLA-dimension gap 
\(\Delta\) defined in \eqref{eq:dla_relative_gap} over the $144$ connected asymmetric seven-vertex graphs. }
\label{tab:SevenNodesDimensions_diff}
\end{table} 

\subsection{Loss Function Gradient Variance and Barren Plateau Scaling}

As determining the  dimensions and structure of dynamical Lie algebras for MIS-QAOA on arbitrary asymmetric graphs becomes computationally intractable as system size grows, we numerically evaluate loss-function gradient variances to characterize the parameter landscapes of the two constrained QAOA architectures. Gradient statistics alone do not uniquely determine DLA dimensions, even in the deep-circuit limit, because they additionally depend on the initial state, cost observable, and parameter-sampling distribution. Nevertheless, these numerical results are of independent theoretical and practical interest, providing direct insight into gradient concentration and the emergence of barren plateaus (see Definition~\ref{defn:BP}).

We evaluate the gradient variance for constrained MIS-QAOA circuits employing both the Flip-or-Void mixer ($H_{CX}$) and the proposed Flip-or-Stay mixer ($\widehat{H}_{CX}$) across problem sizes $n \in \{6, 8, 10, 12, 15\}$. The overall gradient variance is computed by averaging the individual component variances across all $2p$ variational parameters. For each system size $n$, numerical experiments were conducted over $20$ non-isomorphic Erd\H{o}s--R\'enyi random graphs generated with edge probability $p_{\text{edge}} = 1.5\log(n)/n$. For each graph instance, $100$ independent variational parameter vectors $(\boldsymbol{\beta}, \boldsymbol{\gamma})$ were sampled uniformly at random from $[0, 2\pi]^{2p}$, where $p$ denotes the circuit depth.

Our numerical results indicate that the gradient variances  stabilize at depth $p \approx 8$ for $n = 6$, $p \approx 20$ for $n \in \{8, 10, 12\}$, and $p \approx 30$ for $n = 15$. All gradients were computed for a range of circuit depths using PennyLane with the \texttt{autograd} interface and a \texttt{backprop} differentiation method, which resulted in the fastest computation of the gradient \cite{bergholm2018pennylane}. The results are presented in Figure~\ref{fig:variance}.

\begin{figure}[htbp]
    \centering
    \includegraphics[width=0.55\linewidth]{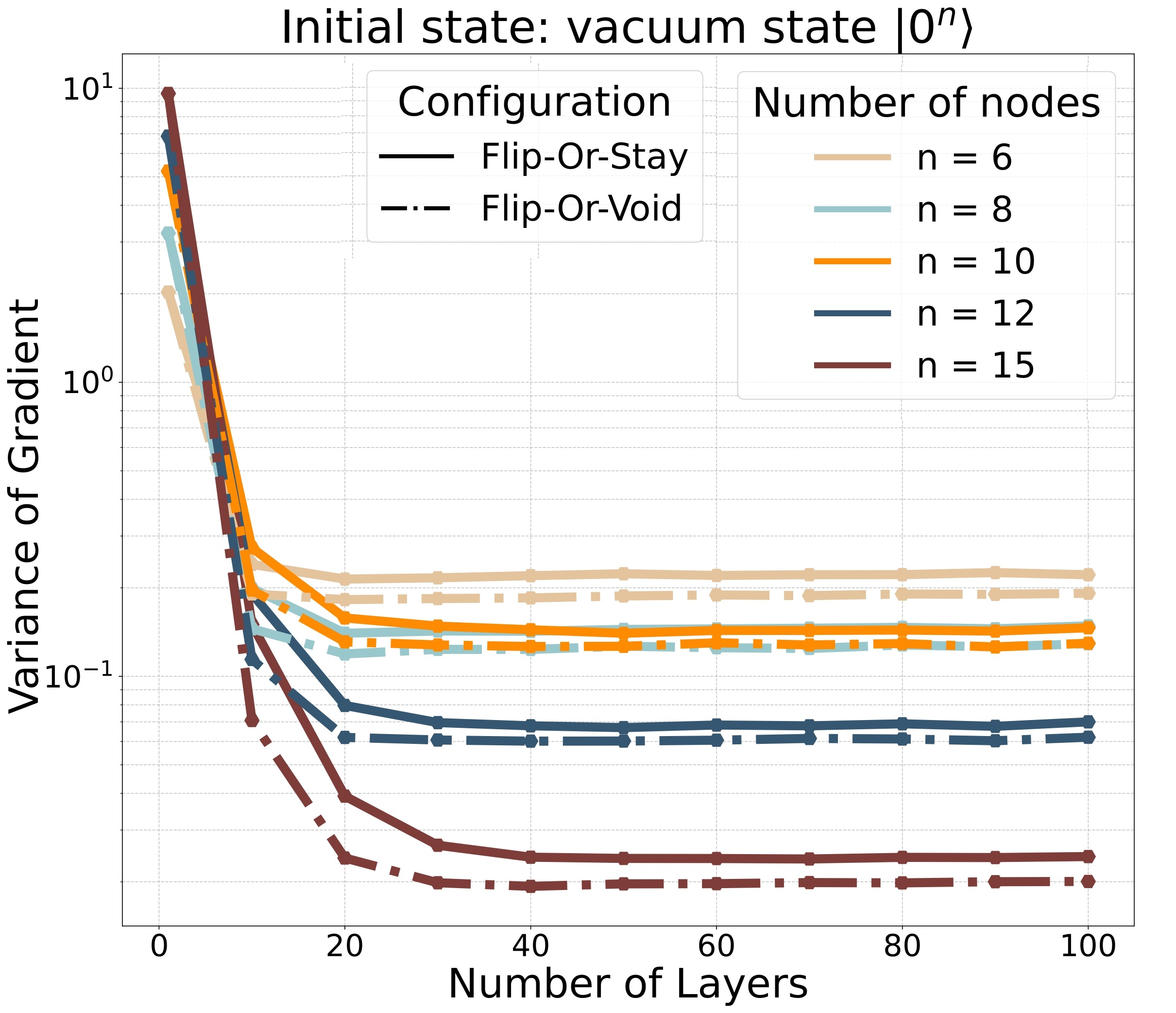}
    \caption{Variance of the loss function gradient for the Maximum Independent Set problem across graph sizes $n \in \{6, 8,10,12, 15\}$ and varying circuit depths $1\le p\le 100$. Panels compare the standard Flip-or-Void mixer $H_{CX}$ QAOA against the proposed Flip-or-Stay mixer $\widehat{H}_{CX}$ variant. Data points represent the mean gradient variance evaluated over 20 non-isomorphic, random Erd\"os-R\"enyi graphs with edge probability of $1.5\log(n)/n$ per system size $n$, with 100 random parameter initializations $\boldsymbol{\theta} \sim U([0, 2\pi]^{2p})$ per instance. 
    }
    \label{fig:variance}
\end{figure}

\section{Conclusion}
\label{sec:conclusion}

We have characterized the dynamical Lie algebras and reachable state spaces of constrained MIS-QAOA with Flip-or-Void and Flip-or-Stay mixers. On the feasible subspace, these mixers act as the adjacency
operator and a shifted negative Laplacian of the independent-set reconfiguration graph, respectively. For every connected graph with at least two vertices, both free dynamical Lie algebras equal
$\mathfrak u(W_{\mathcal F})$. The same full controllability holds for standard QAOA with positive, pairwise distinct vertex weights, despite retaining only one mixer parameter and one cost parameter per layer.

For the unweighted problem, the two standard architectures
can have substantially different reachable state spaces.
Our explicit family $\Gamma_r$, with $n=r(r-1)/2$ vertices,
exhibits a Flip-or-Void vacuum orbit of real dimension
at most $7$, whereas the Flip-or-Stay vacuum orbit has
dimension $3r-2=\Theta(\sqrt n)$. The latter fills the unit sphere of its cyclic subspace, whose dimension nevertheless remains small compared with the $\Theta(n^{3/2})$-dimensional feasible space. Thus, the diagonal correction distinguishing the mixers can produce an unbounded increase in reachable orbit dimension without changing the feasible transitions, problem Hamiltonian, or initial state.

We also established finite-depth exact preparation of
an optimal MIS state from the computational vacuum
under either standard mixer. Under the unitary $2$-design assumption, the free architectures admit an exact loss-variance formula determined by the number of feasible independent sets and the variance of their cardinalities. For Flip-or-Stay ma-QAOA, this further implies a barren plateau whenever the feasible-set cardinality grows exponentially, under the same unitary-design assumption and independent uniform parameter sampling. Together with the Grover-mixer comparison, this connects the algebraic structure of the accessible subspace to loss concentration. The exact-reachability guarantees are existential; constructive synthesis and quantitative circuit-depth bounds remain open.


\section*{AI Disclosure}

During the preparation of this manuscript, the authors used OpenAI ChatGPT (GPT-5.6 Sol) to assist in reviewing relevant literature, verifying the correctness of results, and refining the exposition. All AI-assisted material, including citations and bibliographic information, was reviewed and verified by the authors, who assume full responsibility for the manuscript.

\bibliographystyle{plain}
\bibliography{References}

\appendix

\section{Feasible Subspace Dimensions and Variance Bounds Across Graph Topologies}
\label{sec:graph_examples_and_variance}

To illustrate the dimensional reduction obtained by restricting the optimization to the feasible subspace \(W_{\mathcal F}\), and its direct effect on the deep-circuit loss-function variance, we apply the formula in \eqref{eq:HaarVarianceExactIdentity} to some graph families for which the number and cardinality distribution of independent sets can be determined explicitly. Throughout this section, we consider the free constrained QAOA architectures under the exact unitary $2$-design assumption, with a fixed pure initial state supported on $W_{\mathcal F}$ and the unnormalized problem Hamiltonian $H_P=\sum_{v\in V}Z_v$.

Computing $\operatorname{Var}(|S|)$ directly for uniform distribution $S \sim \operatorname{Unif}(\mathcal{F})$ yields sharper estimates than substituting $|\mathcal{F}|$ into the general bounds of~\eqref{eq:HaarVarianceAsymptoticScaling}.

\subsubsection{Path and Cycle Graphs}
The numbers of independent sets of the path graph $\mathcal{P}_n$ ($n \ge 1$) and cycle graph $\mathcal{C}_n$ ($n \ge 3$) (Figure~\ref{fig:graph_families_combined}) are given by
\begin{equation}
\label{eq:cardinality_path_cycle}
\begin{aligned}
|\mathcal{F}_{\mathcal{P}_n}| &= F_{n+2}, \\
|\mathcal{F}_{\mathcal{C}_n}| &= F_{n-1} + F_{n+1} = L_n,
\end{aligned}
\end{equation}
respectively~\cite{Arocha1984}.
Here, $F_n$ and $L_n$ denote the Fibonacci and Lucas numbers, with $F_0 = 0$, $F_1 = 1, F_n = F_{n-1} + F_{n-2}$ and $L_0=2, L_n = F_{n-1} + F_{n+1}$.
Consequently, both feasible subspace dimensions grow as $\Theta(\varphi^n)$, where $\varphi = \frac{1 + \sqrt{5}}{2}$.

To evaluate the cardinality variance, let
\[
I_\Gamma(z) \coloneqq \sum_{S\in\mathcal F_\Gamma} z^{|S|},
\qquad
\mathcal D \coloneqq z \frac{\mathrm d}{\mathrm dz}.
\]
For $S \sim \operatorname{Unif}(\mathcal F_\Gamma)$, differentiating the generating function gives
\[
\mathbb E[|S|] = \frac{(\mathcal D I_\Gamma)(1)}{I_\Gamma(1)},
\qquad
\mathbb E[|S|^2] = \frac{(\mathcal D^2 I_\Gamma)(1)}{I_\Gamma(1)}.
\]
Consequently,
\begin{equation}
\label{eq:independence_polynomial_variance}
\operatorname{Var}(|S|) = \mathcal D^2 \log (I_\Gamma(z)) \bigg|_{z=1}.
\end{equation}

For paths, conditioning on the inclusion of an endpoint yields
\[
I_{\mathcal P_n}(z) = I_{\mathcal P_{n-1}}(z) + z I_{\mathcal P_{n-2}}(z),
\qquad
I_{\mathcal P_0}(z) = 1, \quad I_{\mathcal P_1}(z) = 1 + z.
\]
The characteristic roots of this recurrence are
\[
\lambda_\pm(z) \coloneqq \frac{1 \pm \sqrt{1 + 4z}}{2},
\]
and the initial conditions give
\[
I_{\mathcal P_n}(z) = \frac{\lambda_+(z)^{n+2} - \lambda_-(z)^{n+2}}{\sqrt{1 + 4z}}.
\]
For cycles, excluding a fixed vertex leaves a path on $n-1$ vertices, whereas including it excludes its two neighbors and leaves a path on $n-3$ vertices. Thus,
\[
I_{\mathcal C_n}(z) = I_{\mathcal P_{n-1}}(z) + z I_{\mathcal P_{n-3}}(z) = \lambda_+(z)^n + \lambda_-(z)^n, \qquad n \ge 3,
\]

where the last equality follows by substituting the path
formula and using $\lambda_+(z)\lambda_-(z)=-z$ and $\lambda_+(z)+\lambda_-(z)=1$.

Set $q(z) \coloneqq \lambda_-(z) / \lambda_+(z)$. The preceding expressions imply
\[
\begin{aligned}
\log (I_{\mathcal P_n}(z)) &= (n+2)\log (\lambda_+(z)) - \frac{1}{2}\log(1 + 4z) + \log\bigl(1 - q(z)^{n+2}\bigr), \\
\log (I_{\mathcal C_n}(z)) &= n \log (\lambda_+(z)) + \log\bigl(1 + q(z)^n\bigr).
\end{aligned}
\]
Since $|q(1)| = \varphi^{-2} < 1$, the contributions of the last terms to $\mathcal D^2 \log (I_\Gamma(z))$ at $z=1$ are $\mathcal O(n^2 \varphi^{-2n})$. Moreover,
\[
\mathcal D \log (\lambda_+(z)) = \frac{1}{2} \left( 1 - \frac{1}{\sqrt{1 + 4z}} \right),
\qquad
\mathcal D^2 \log (\lambda_+(z)) = \frac{z}{(1 + 4z)^{3/2}}.
\]
Applying \eqref{eq:independence_polynomial_variance} therefore gives
\[
\operatorname{Var}(|S|) = 
\begin{cases}
\displaystyle \frac{n+2}{5\sqrt{5}} - \frac{2}{25} + \mathcal O(n^2 \varphi^{-2n}), & \Gamma = \mathcal P_n, \\[1.2ex]
\displaystyle \frac{n}{5\sqrt{5}} + \mathcal O(n^2 \varphi^{-2n}), & \Gamma = \mathcal C_n.
\end{cases}
\]
In particular, $\operatorname{Var}(|S|) = \Theta(n)$ for both families.

 Combining this with
$|\mathcal F_\Gamma|=\Theta(\varphi^n)$ and
\eqref{eq:HaarVarianceExactIdentity} yields
\begin{equation}
\label{eq:variance_path_cycle}
\operatorname{Var}_{\boldsymbol\theta}
(\ell_{\boldsymbol\theta})
=
\Theta\!\left(\frac{n}{\varphi^n}\right).
\end{equation}

\subsubsection{Star Graphs}
The star graph $K_{1,n-1}$ consists of a single central hub vertex connected to $n-1$ peripheral leaf vertices ($n1$ total vertices; see Figure~\ref{fig:graph_families_combined}). Its independent sets consist of all subsets of the leaves and the singleton containing the hub. Hence,
\begin{equation}
\label{eq:cardinality_star}
|\mathcal{F}_{K_{1,n-1}}| = 2^{n-1} + 1.
\end{equation}
The feasible subspace therefore occupies asymptotically one half of the ambient $2^{n}$-dimensional Hilbert space.

Set
\[
q_n \coloneqq \frac{2^{n-1}}{2^{n-1} + 1}.
\]
With probability $q_n$, a uniformly chosen independent set is a uniformly chosen subset of the leaves, whose cardinality has distribution $\operatorname{Bin}(n-1, 1/2)$. With probability $1 - q_n$, it is the hub singleton. The law of total variance gives
\[
\operatorname{Var}(|S|) = q_n \frac{n-1}{4} + q_n (1 - q_n) \left( \frac{n-1}{2} - 1 \right)^2.
\]
Consequently, the formula in \eqref{eq:HaarVarianceExactIdentity} gives
\begin{equation}
\label{eq:variance_star}
\operatorname{Var}_{\boldsymbol\theta}
(\ell_{\boldsymbol\theta})
=
\frac{
q_n (n-1)+q_n(1-q_n)(n-3)^2
}{2^{n-1}+2}
=
\Theta\!\left(\frac{n-1}{2^{n-1}}\right).
\end{equation}

\subsubsection{Complete Bipartite Graphs}
Consider the complete bipartite graph $K_{s,\ell}$ with partition sizes $s$ and $\ell$, where $n = s + \ell$ (Figure~\ref{fig:graph_families_combined}). Any valid independent set must reside entirely within either the first or the second partition. Accounting for the empty set (counted in both partitions), the cardinality is:
\begin{equation}
\label{eq:cardinality_bipartite}
|\mathcal{F}_{K_{s,\ell}}| = 2^{s} + 2^{\ell} - 1.
\end{equation}

For balanced partitions $s = \ell = n/2$, we have $|\mathcal{F}| = 2^{n/2+1} - 1 = \Theta(2^{n/2})$, which reduces the effective dimension from $2^n$ to $2^{n/2}$. Equation~\eqref{eq:HaarVarianceAsymptoticScaling} yields:
\begin{equation}
\label{eq:variance_bipartite}
\operatorname{Var}_{\boldsymbol{\theta}}\!\left(\ell_{\boldsymbol{\theta}}\right) 
\in \Omega\left(\frac{1}{2^n}\right) \cap \mathcal{O}\left(\frac{n^2}{\sqrt{2^n}}\right).
\end{equation}

\subsubsection{Complete Graphs}
For the complete graph $K_n$
(Figure~\ref{fig:graph_families_combined}), the independent sets
are the empty set and the $n$ singletons. Therefore,
\begin{equation}
\label{eq:cardinality_complete}
|\mathcal F_{K_n}|=n+1.
\end{equation}
The cardinality of a uniformly chosen independent set
is a Bernoulli random variable with
\[
\Pr(|S|=1)=\frac{n}{n+1},
\qquad
\operatorname{Var}(|S|)
=
\frac{n}{(n+1)^2}.
\]
Hence the loss variance is exactly
\begin{equation}
\label{eq:variance_complete}
\operatorname{Var}_{\boldsymbol\theta}
(\ell_{\boldsymbol\theta})
=
\frac{4n}{(n+1)^2(n+2)}
=
\Theta(n^{-2}).
\end{equation}
Unlike the preceding families, complete graphs therefore
exhibit only inverse-polynomial loss concentration under
the stated ensemble assumption.

\begin{figure}[h!]

\centering

\begin{tikzpicture}[scale=0.9, every node/.style={circle, draw, fill=black, inner sep=1.8pt}]

  \def\npath{5}

  \foreach \i in {1,...,\npath} {

    \node (v\i) at ({1.2*(\i-1)},0) {};

  }

  \foreach \i in {1,...,4} {

    \pgfmathtruncatemacro\ip{\i+1}

    \draw[thick] (v\i) -- (v\ip);

  }

  \foreach \i in {1,...,\npath} {

    \node[draw=none, fill=none, below=4pt] at (v\i) {$v_\i$};

  }
\end{tikzpicture}
\begin{tikzpicture}[scale=0.9, every node/.style={circle, draw, fill=black, inner sep=1.8pt}]

  \def\ncycle{5}

  \foreach \i in {1,...,\ncycle} {

    \node (c\i) at (90-\i*360/\ncycle:1.4) {};

    \node[draw=none, fill=none] at (90-\i*360/\ncycle:1.75) {$v_\i$};

  }

  \foreach \i in {1,...,\ncycle} {

    \pgfmathtruncatemacro\ip{mod(\i,\ncycle) + 1}

    \draw[thick] (c\i) -- (c\ip);

  }
\end{tikzpicture}
\begin{tikzpicture}[scale=1.1, every node/.style={circle, draw, fill=black, inner sep=2pt}]

  \def\nleaves{5}

  \node (hub) at (0,0) {};

  \node[draw=none, fill=none, right=3pt] at (hub) {$v_{\text{hub}}$};

  \foreach \i in {1,...,\nleaves} {

    \node (leaf\i) at (90-\i*360/\nleaves:1.5) {};

    \node[draw=none, fill=none] at (90-\i*360/\nleaves:1.85) {$\ell_\i$};

    \draw[thick, gray!70] (hub) -- (leaf\i);

  }

\end{tikzpicture}

\begin{tikzpicture}[scale=1.1, every node/.style={circle, draw, fill=black, inner sep=2pt}]

  \def\s{3}

  \def\l{4}

  \foreach \i in {1,...,\s} {

    \node (L\i) at (0, {-\i*1.1 + (\s+1)*0.55}) {};

    \node[draw=none, fill=none, left=3pt] at (L\i) {$L_\i$};

  }

  \foreach \j in {1,...,\l} {

    \node (R\j) at (2.8, {-\j*1.1 + (\l+1)*0.55}) {};

    \node[draw=none, fill=none, right=3pt] at (R\j) {$R_\j$};

  }

  \foreach \i in {1,...,\s} {

    \foreach \j in {1,...,\l} {

      \draw[thick, gray!70] (L\i) -- (R\j);

    }

  }

\end{tikzpicture}
\begin{tikzpicture}[scale=1.5, every node/.style={circle, draw, fill=black, inner sep=2pt}]

  \def\n{5} 


  \foreach \i in {1,...,\n} {

    \node (v\i) at (90-\i*360/\n:1.5) {};

    \node[draw=none, fill=none] at (90-\i*360/\n:1.8) {$v_\i$};

  }


  \foreach \i in {1,...,\n} {

    \foreach \j in {1,...,\n} {

      \ifnum\i<\j

        \draw[thick, gray!70] (v\i) -- (v\j);

      \fi

    }

  }

\end{tikzpicture}

\caption{Benchmark graph structures analyzed in this work: path graph $\mathcal{P}_5$ (top left), cycle graph $\mathcal{C}_5$ (top middle), star graph $K_{1,5}$ with central hub vertex $v_{\text{hub}}$ (top right), complete bipartite graph $K_{3,4}$ with partition sizes $s=3$ and $\ell=4$ (bottom left), and complete graph $K_5$ (bottom right).}
\label{fig:graph_families_combined}
\end{figure}
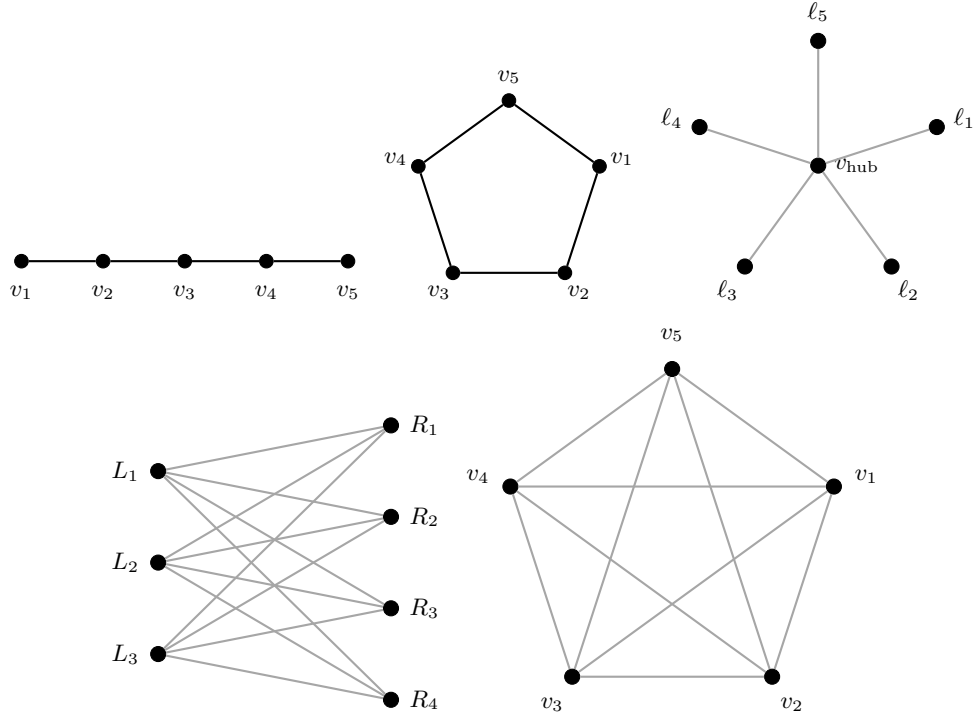

\section{Proofs and Technical Details}
\label{sec:proofs}
This appendix provides the technical details and complete proofs underlying the theoretical results presented in the main text. The exposition is organized to be largely self-contained and includes the intermediate arguments and derivations needed to verify the stated claims.

We begin by showing that the restrictions of the local and global Flip-or-Void mixers, $H_{CX,v}$ and $H_{CX}$, to the feasible subspace $W_{\mathcal{F}}$ can be expressed as the projected actions of the local Pauli-$X$ operator $X_v$ and the global transverse-field mixer $B = \sum\limits_{v \in V} X_v$ via the subspace projector $\mathcal{P}_{\mathcal{F}}$, as stated in \eqref{eq:constrained_hamiltonians}. Furthermore, we establish explicit matrix representations for these operators in terms of the adjacency matrix $A_{\mathcal{R}(\Gamma)}$ and graph Laplacian $L_{\mathcal{R}(\Gamma)}$ of the independent set reconfiguration graph $\mathcal{R}(\Gamma)$, as presented in \eqref{eq:FlipOrStayAsLaplacian}.

\begin{proof}[Proof of Proposition \ref{prop:mixer_restriction_properties}]
Let $\ket{x}, \ket{y} \in W_\mathcal{F}$ be any two feasible basis states. Using the definition of the projection operator $ \mathcal{P}_{\mathcal{F}}$ (see \eqref{eq:FeasibleProjector}), we evaluate the matrix elements of $\mathcal{P}_{\mathcal{F}} X_v \mathcal{P}_{\mathcal{F}}$ as:
\begin{equation}
    \bra{y}\mathcal{P}_{\mathcal{F}} X_v \mathcal{P}_{\mathcal{F}}\ket{x} = \bra{y}X_v\ket{x}.
\end{equation}
The expression $\bra{y}X_v\ket{x}$ is non-zero (equal to $1$) if and only if $y = x^{(v)}$. Because $x, y \in \mathcal{F}$, the state $x$ cannot have any active neighbors in $\mathcal{N}_v$ if $y = x^{(v)}$ is also a valid independent set. Thus, the neighborhood projector satisfies $\mathcal{P}_v\ket{x} = \ket{x}$. Comparing this to the action of $H_{CX,v}$:
\begin{equation}
    \bra{y}H_{CX,v}\ket{x} = \bra{y}X_v\mathcal{P}_v\ket{x} = \bra{y}X_v\ket{x}.
    \end{equation}
If $y \neq x^{(v)}$ or if $\mathcal{P}_v$ vanishes, both expressions equal $0$. Since their matrix elements coincide identically on all computational basis vectors spanning $W_{\mathcal{F}}$, it follows that $H_{CX,v}\big|_{\mathcal{W}_\mathcal{F}} = \mathcal{P}_{\mathcal{F}} X_v \mathcal{P}_{\mathcal{F}}\big|_{\mathcal{W}_\mathcal{F}}$. Summing over all $v \in V(\Gamma)$ completes the proof for the total operator.
\end{proof}

\begin{proof}[Proof of Proposition~\ref{prop:MixersAsReconfigurationOperators}]
By the definition of the Flip-or-Void mixer in~\eqref{eq:HCXv_action}, the matrix element $\bra{y} H_{CX}\big|_{W_{\mathcal{F}}} \ket{x}$ equals $1$ if $x$ and $y$ are adjacent in $\mathcal{R}(\Gamma)$ (i.e., differ by a valid single-vertex bit flip) and $0$ otherwise. Thus, $H_{CX}\big|_{W_{\mathcal{F}}} = A_{\mathcal{R}}$, proving the first identity.

To establish the second identity, recall that for any $x \in \mathcal{F}$, the constraint projector $\mathcal{P}_v$ acts as
\[
\mathcal{P}_v \ket{x} =
\begin{cases}
    \ket{x}, & x^{(v)} \in \mathcal{F}, \\[1mm]
    0,       & x^{(v)} \notin \mathcal{F}.
\end{cases}
\]
Summing the complement projectors $I - \mathcal{P}_v$ over all vertices $v \in V(\Gamma)$ yields
\[
\sum_{v \in V(\Gamma)} (I - \mathcal{P}_v) \ket{x} = \bigl(n - \deg_{\mathcal{R}}(x)\bigr) \ket{x},
\]
where $\deg_{\mathcal{R}}(x)$ denotes the degree of vertex $x$ in $\mathcal{R}(\Gamma)$. Restricting to $W_{\mathcal{F}}$, this operator identity reads
\begin{equation}
\label{eq:BlockedProjectorAsDegreeComplement}
    \sum_{v \in V(\Gamma)} (I - \mathcal{P}_v)\Big|_{W_{\mathcal{F}}} = nI - D_{\mathcal{R}}.
\end{equation}
Finally, using the mixer decomposition $\widehat{H}_{CX} = H_{CX} + \sum\limits_{v \in V(\Gamma)} (I - \mathcal{P}_v)$, we substitute the restricted terms on $W_{\mathcal{F}}$:
\[
    \widehat{H}_{CX}\big|_{W_{\mathcal{F}}} = A_{\mathcal{R}} + (nI - D_{\mathcal{R}}) = nI - (D_{\mathcal{R}} - A_{\mathcal{R}}) = nI - L_{\mathcal{R}},
\]
which completes the proof.
\end{proof}

Next, we verify the structural properties of the free dynamical Lie algebras stated in the main text.

\begin{proof}[Proof of Proposition \ref{prop:free_DLA_block_embedding}]

   To establish the block-diagonal structure, it suffices to demonstrate that each orbit subspace $W_{\mathcal{O}_j}$ (including the feasible subspace $W_\mathcal{F}$) is preserved under the action of the generators of $\widehat{\mathfrak{g}}_{\Gamma,\mathrm{free}}$. Recall that $\widehat{\mathfrak{g}}_{\Gamma,\mathrm{free}}$ is defined as the Lie algebra generated under the Lie bracket by the set of local mixers $\{i\widehat{H}_{CX,v}\}_{v \in V}$ and the local phase operators $\{iZ_v\}_{v \in V}$.
    
    First, we consider the local mixers. By definition of the orbit decomposition, the action of any local mixer $\widehat{H}_{CX,v}$ on a basis state $\ket{x} \in W_{\mathcal{O}_j}$ permutes it to another basis state $\ket{x'}$ within the same orbit $\mathcal{O}_j$. Therefore, $\widehat{H}_{CX,v} (W_{\mathcal{O}_j}) \subseteq W_{\mathcal{O}_j}$.
    
    Second, consider the local phase operators. Each $Z_v$ acts diagonally on the computational basis via $Z_v \ket{x} = (-1)^{x_v} \ket{x}$. Because $Z_v$ scales every basis state individually, it trivially preserves any subspace spanned by a subset of these basis states, implying $Z_v (W_{\mathcal{O}_j}) \subseteq W_{\mathcal{O}_j}$.
    
    Since all generating elements map $W_{\mathcal{O}_j}$ into itself, any linear combination or nested commutator of these generators must also preserve $W_{\mathcal{O}_j}$. It follows that every element in $\widehat{\mathfrak{g}}_{\Gamma,\mathrm{free}}$ is block-diagonal with respect to the orthogonal decomposition of the total Hilbert space 
    $W =W_{\mathcal{F}}\oplus\bigoplus\limits_{j=1}^k W_{\mathcal{O}_j}$, completing the proof.
\end{proof}

Next, we establish the precise structure of the restricted free dynamical Lie algebras on the feasible subspace $W_{\mathcal{F}}$. Specifically, we prove that assigning independent variational parameters to each generator expands the dynamical Lie algebra to the full unitary Lie algebra $\mathfrak{u}(W_{\mathcal{F}})$, thereby establishing Theorem~\ref{thm:ExtendedFreeLieAlgFullUnitary}.

\begin{proof}[Proof of Theorem~\ref{thm:ExtendedFreeLieAlgFullUnitary}]

The proof proceeds in four main steps: we localize a single vacuum-to-singleton transition, propagate this operator across all singletons, extend the construction inductively by Hamming weight, and finally generate the central scalar direction.

For brevity, write
\[
\mathfrak g_{\mathcal F}
:=
\mathfrak g_{\Gamma,\mathrm{free},\mathcal F},
\qquad
\widehat{\mathfrak g}_{\mathcal F}
:=
\widehat{\mathfrak g}_{\Gamma,\mathrm{free},\mathcal F}.
\]
All operators below are restricted to \(W_{\mathcal F}\) whenever
matrix units are used.

First observe that the second identity in \eqref{eq:double_commutators} implies that
every generator of \(\mathfrak g_{\mathcal F}\) belongs to
\(\widehat{\mathfrak g}_{\mathcal F}\), and therefore
\begin{equation}
\label{eq:algebra_chain_inclusion}
\mathfrak g_{\mathcal F}
\subseteq
\widehat{\mathfrak g}_{\mathcal F}
\subseteq
\mathfrak u(W_{\mathcal F}).
\end{equation}
It consequently suffices to prove that
\(\mathfrak g_{\mathcal F}=\mathfrak u(W_{\mathcal F})\).

For a subset \(S\subseteq V\), define the operator
\begin{equation}
\label{eq:zero_projector_subset}
Q_S
:=
\prod_{u\in S}\frac{I+Z_u}{2},
\qquad
Q_{\varnothing}:=I.
\end{equation}
Thus \(\mathcal P_v=Q_{\mathcal N_v}\) and
\[
H_{CX,v}=X_vQ_{\mathcal N_v}.
\]
Since \(Q_{\mathcal N_v}\) acts trivially on the \(v\)-th qubit, it
commutes with \(X_v,Y_v\), and \(Z_v\). Using
\([Z_v,X_v]=2iY_v\) and \([X_v,Y_v]=2iZ_v\), we obtain
\begin{equation}
\label{eq:local_reduction}
\begin{aligned}
\frac12[iZ_v,iH_{CX,v}]
&=
-iY_vQ_{\mathcal N_v},\\
\frac12
\left[
iH_{CX,v},
-iY_vQ_{\mathcal N_v}
\right]
&=
iZ_vQ_{\mathcal N_v}.
\end{aligned}
\end{equation}
In particular,
\begin{equation}
\label{eq:local_cz}
D(v,\mathcal N_v)
:=
iZ_vQ_{\mathcal N_v}
\in
\mathfrak g_{\Gamma,\mathrm{free}}.
\end{equation}

We next establish the propagation identity that will be used to
localize this diagonal operator. For \(r\in V\) and
\(C\subseteq V\setminus\{r\}\), write
\[
D(r,C):=iZ_rQ_C.
\]
Suppose that \(\{r,w\}\in E\) and \(w\in C\). Since
\(Q_C=Q_{\{w\}}Q_{C\setminus\{w\}}\), while
\(r\in\mathcal N_w\), we have
\[
[Q_{\{w\}},X_w]=iY_w,
\qquad
Z_rQ_{\mathcal N_w}=Q_{\mathcal N_w}.
\]
A direct calculation therefore gives
\begin{equation}
\label{eq:diagonal_propagation_first}
\left[
D(r,C),iH_{CX,w}
\right]
=
-iY_w
Q_{(C\setminus\{w\})\cup\mathcal N_w}.
\end{equation}
Commuting once more with \(iH_{CX,w}\) yields
\begin{equation}
\label{eq:diagonal_propagation}
\frac12
\left[
iH_{CX,w},
\left[
D(r,C),iH_{CX,w}
\right]
\right]
=
D\left(
w,
(C\setminus\{w\})\cup\mathcal N_w
\right).
\end{equation}

Because \(\Gamma\) is connected, there exists a finite walk
\[
r_0,r_1,\ldots,r_m
\]
that visits every vertex and satisfies
\((r_k,r_{k+1})\in E\) for every \(k\). Repetitions of vertices are
allowed. Define
\begin{equation}
\label{eq:walk_projector_sets}
C_k
:=
\left(
\bigcup_{j=0}^{k}\mathcal N_{r_j}
\right)
\setminus\{r_k\}.
\end{equation}
The initial operator
\[
D(r_0,C_0)
=
iZ_{r_0}Q_{\mathcal N_{r_0}}
\]
belongs to \(\mathfrak g_{\Gamma,\mathrm{free}}\) by
\eqref{eq:local_cz}. Moreover,
\(r_{k+1}\in C_k\), and
\[
C_{k+1}
=
(C_k\setminus\{r_{k+1}\})
\cup\mathcal N_{r_{k+1}}.
\]
Thus, repeated application of
\eqref{eq:diagonal_propagation} proves inductively that
\begin{equation}
\label{eq:walk_diagonal_induction}
D(r_k,C_k)
\in
\mathfrak g_{\Gamma,\mathrm{free}}
\qquad
\text{for every }k.
\end{equation}

Let $v_{\mathrm{end}} := r_m$ be the final vertex of the walk. Since the walk visits every vertex of the connected graph $\Gamma$, we have
\[
    \bigcup_{j=0}^{m} \mathcal{N}_{r_j} = V.
\]
Consequently, $C_m = V \setminus \{v_{\mathrm{end}}\}$, which implies
\begin{equation}
\label{eq:global_cz}
    D_{v_{\mathrm{end}}} := i Z_{v_{\mathrm{end}}} Q_{V \setminus \{v_{\mathrm{end}}\}} \in \mathfrak{g}_{\Gamma,\mathrm{free}}.
\end{equation}
The projector
\(Q_{V\setminus\{v_{\mathrm{end}}\}}\) vanishes on every computational
basis state except \(0^n\) and \(e_{v_{\mathrm{end}}}\). Since
\(Z_{v_{\mathrm{end}}}\) has eigenvalues \(+1\) and \(-1\) on these
two states, respectively, restriction to \(W_{\mathcal F}\) gives
\begin{equation}
\label{eq:rank2_diagonal}
D_{v_{\mathrm{end}}}
=
iE_{0^n,0^n}
-
iE_{e_{v_{\mathrm{end}}},
     e_{v_{\mathrm{end}}}}
\in
\mathfrak g_{\mathcal F}.
\end{equation}

For distinct \(x,y\in\mathcal F\), introduce the standard
skew-Hermitian matrix units
\begin{equation}
\label{eq:skew_matrix_units}
A_{x,y}:=E_{x,y}-E_{y,x},
\qquad
B_{x,y}:=i(E_{x,y}+E_{y,x}).
\end{equation}
The local mixer has the matrix-unit expansion
\begin{equation}
\label{eq:local_mixer_matrix_units}
H_{CX,v}\big|_{W_{\mathcal F}}
=
\sum_{\substack{x\in\mathcal F\\
                 x_v=0,\;x^{(v)}\in\mathcal F}}
\left(
E_{x,x^{(v)}}+E_{x^{(v)},x}
\right).
\end{equation}
In particular,
\begin{equation}
\label{eq:root_mixer_decomposition}
H_{CX,v_{\mathrm{end}}}\big|_{W_{\mathcal F}}
=
E_{0^n,e_{v_{\mathrm{end}}}}
+
E_{e_{v_{\mathrm{end}}},0^n}
+
R_{v_{\mathrm{end}}},
\end{equation}
where none of the matrix units in \(R_{v_{\mathrm{end}}}\) has a row
or column indexed by \(0^n\) or \(e_{v_{\mathrm{end}}}\). It follows
that
\[
[D_{v_{\mathrm{end}}},R_{v_{\mathrm{end}}}]=0.
\]
We obtain
\begin{equation}
\label{eq:root_antisymmetric_matrix_unit}
-\frac12
[D_{v_{\mathrm{end}}},
 iH_{CX,v_{\mathrm{end}}}]
=
A_{0^n,e_{v_{\mathrm{end}}}}
\in
\mathfrak g_{\mathcal F},
\end{equation}
and
\begin{equation}
\label{eq:root_symmetric_matrix_unit}
-\frac14
[D_{v_{\mathrm{end}}},
 [D_{v_{\mathrm{end}}},
  iH_{CX,v_{\mathrm{end}}}]]
=
B_{0^n,e_{v_{\mathrm{end}}}}
\in
\mathfrak g_{\mathcal F}.
\end{equation}

We now propagate these vacuum--singleton matrix units through a
breadth-first-search tree of \(\Gamma\) rooted at
\(v_{\mathrm{end}}\). Suppose that
\[
A_{0^n,e_p},\ B_{0^n,e_p}
\in\mathfrak g_{\mathcal F},
\]
and let \(w\) be a child of \(p\) in this tree. Since
\((p,w)\in E\), the configuration \(e_p\) blocks the local mixer at
\(w\), so
\[
H_{CX,w}\ket{e_p}=0.
\]
On the other hand,
\[
H_{CX,w}\ket{0}^n=\ket{e_w},
\qquad
H_{CX,w}\ket{e_w}=\ket{0}^n.
\]
Therefore, the only terms of \(H_{CX,w}\) contributing to the
following commutators are the matrix units on positions
\((0^n,e_w)\) and \((e_w,0^n)\). Direct multiplication gives
\begin{equation}
\label{eq:singleton_exchange_units}
\begin{aligned}
[B_{0^n,e_p},iH_{CX,w}]
&=
-A_{e_p,e_w},\\
[A_{0^n,e_p},iH_{CX,w}]
&=
-B_{e_p,e_w}.
\end{aligned}
\end{equation}
Commuting these exchange units with the already known transition at
\(p\) gives
\begin{equation}
\label{eq:singleton_bfs_step}
\begin{aligned}
[B_{e_p,e_w},B_{0^n,e_p}]
&=
A_{0^n,e_w},\\
[A_{e_p,e_w},B_{0^n,e_p}]
&=
-B_{0^n,e_w}.
\end{aligned}
\end{equation}
Equations~\eqref{eq:root_antisymmetric_matrix_unit},
\eqref{eq:root_symmetric_matrix_unit}, and
\eqref{eq:singleton_bfs_step} prove by induction on the BFS distance
that
\begin{equation}
\label{eq:all_vacuum_singletons}
A_{0^n,e_v},\ B_{0^n,e_v}
\in
\mathfrak g_{\mathcal F}
\qquad
\text{for every }v\in V.
\end{equation}

We next prove by induction on Hamming weight that
\begin{equation}
\label{eq:all_vacuum_transitions}
A_{0^n,x},\ B_{0^n,x}
\in
\mathfrak g_{\mathcal F}
\qquad
\text{for every }
x\in\mathcal F\setminus\{0^n\}.
\end{equation}
The weight-one case is exactly
\eqref{eq:all_vacuum_singletons}. Furthermore, whenever the
vacuum transitions corresponding to two distinct nonzero feasible
strings \(x\) and \(z\) have been generated, we also have
\begin{equation}
\label{eq:lower_weight_pair_units}
\begin{aligned}
A_{x,z}
&=
-[A_{0^n,x},A_{0^n,z}],\\
B_{x,z}
&=
-[A_{0^n,x},B_{0^n,z}].
\end{aligned}
\end{equation}

Now let \(x\in\mathcal F\) have Hamming weight \(k\geq2\), choose
\(v\in\operatorname{supp}(x)\), and set
\[
y:=x^{(v)}=x\setminus\{v\}.
\]
Then \(y\in\mathcal F\) has weight \(k-1\), and the local mixer at
\(v\) contains both transitions
\[
0^n\longleftrightarrow e_v,
\qquad
y\longleftrightarrow x.
\]
Because \(H_{CX,v}\big|_{W_{\mathcal F}}\) is a matching operator,
all its remaining matrix units have rows and columns disjoint from
\(0^n,y,e_v,x\). Direct calculation therefore yields
\begin{equation}
\label{eq:hamming_weight_induction}
\begin{aligned}
[B_{0^n,y},iH_{CX,v}]
&=
-A_{0^n,x}-A_{y,e_v},\\
[A_{0^n,y},iH_{CX,v}]
&=
B_{0^n,x}-B_{y,e_v}.
\end{aligned}
\end{equation}
Both \(y\) and \(e_v\) have weight at most \(k-1\). Hence the
induction hypothesis and \eqref{eq:lower_weight_pair_units} imply
that \(A_{y,e_v}\) and \(B_{y,e_v}\) already belong to
\(\mathfrak g_{\mathcal F}\). We can therefore isolate
\begin{equation}
\label{eq:isolate_next_weight}
\begin{aligned}
A_{0^n,x}
&=
-[B_{0^n,y},iH_{CX,v}]
-A_{y,e_v},\\
B_{0^n,x}
&=
[A_{0^n,y},iH_{CX,v}]
+B_{y,e_v}.
\end{aligned}
\end{equation}
This completes the induction and proves
\eqref{eq:all_vacuum_transitions}.

Applying \eqref{eq:lower_weight_pair_units} to all distinct nonzero
\(x,z\in\mathcal F\), we conclude that
\begin{equation}
\label{eq:all_off_diagonal_units}
A_{x,z},\ B_{x,z}
\in
\mathfrak g_{\mathcal F}
\qquad
\text{for every distinct }x,z\in\mathcal F.
\end{equation}
Their commutators give all traceless diagonal directions:
\begin{equation}
\label{eq:all_diagonal_differences}
\frac12[A_{x,z},B_{x,z}]
=
i(E_{x,x}-E_{z,z})
\in
\mathfrak g_{\mathcal F}.
\end{equation}
The elements in
\eqref{eq:all_off_diagonal_units} and
\eqref{eq:all_diagonal_differences} form the standard real basis of
\(\mathfrak{su}(W_{\mathcal F})\). Thus,
\begin{equation}
\label{eq:su_containment}
\mathfrak{su}(W_{\mathcal F})
\subseteq
\mathfrak g_{\mathcal F}.
\end{equation}

It remains to generate the central direction. Since
\(|V|\geq2\) and \(\Gamma\) is connected, choose an edge
\((v,w)\in E\). Let
\[
K:=i(Z_v+Z_w)\big|_{W_{\mathcal F}}
\in\mathfrak g_{\mathcal F}.
\]
No feasible configuration contains both \(v\) and \(w\). Hence
\(Z_v+Z_w\) has eigenvalue \(2\) on feasible configurations
containing neither vertex and eigenvalue \(0\) on those containing
exactly one of them. Therefore,
\begin{equation}
\label{eq:nonzero_trace_generator}
\operatorname{Tr}(K)
=
2i\,
\left|
\left\{
x\in\mathcal F:
x_v=x_w=0
\right\}
\right|
=:i\tau,
\qquad
\tau>0.
\end{equation}
Writing \(d:=\dim W_{\mathcal F}=|\mathcal F|\), the traceless
skew-Hermitian operator
\[
K_0:=K-\frac{i\tau}{d}I
\]
belongs to \(\mathfrak{su}(W_{\mathcal F})\), and hence to
\(\mathfrak g_{\mathcal F}\) by \eqref{eq:su_containment}. It follows
that
\[
iI=\frac{d}{\tau}(K-K_0)
\in\mathfrak g_{\mathcal F}.
\]
Together with \eqref{eq:su_containment}, this proves
\[
\mathfrak g_{\mathcal F}
=
\mathfrak u(W_{\mathcal F}).
\]
Finally, the inclusions in \eqref{eq:algebra_chain_inclusion} force
\[
\mathfrak g_{\mathcal F}
=
\widehat{\mathfrak g}_{\mathcal F}
=
\mathfrak u(W_{\mathcal F}),
\]
as claimed.

\end{proof}

\begin{proof}[Proof of Proposition~\ref{prop:WeightedMISStandardDLA}]
Let $\mathfrak{h} := \left\langle i H_{\mathrm{CX}}\big|_{W_{\mathcal{F}}}, i H_{P,w}\big|_{W_{\mathcal{F}}} \right\rangle_{\mathrm{Lie}}$ denote the standard dynamical Lie algebra for the weighted MIS QAOA employing $H_{\mathrm{CX}}$ mixer. Consider the double-adjoint superoperator 
\[
\mathcal{T} := -\operatorname{ad}_{iH_{P,w}\big|_{W_{\mathcal{F}}}}^2.
\]
The action of $\mathcal{T}$ on each local mixer generator yields
\begin{equation}
\label{eq:superoperator_action_local_mixers}
    \mathcal{T}(i H_{\mathrm{CX},v}\big|_{W_{\mathcal{F}}}) = 4 w_v^2 \cdot i H_{\mathrm{CX},v}\big|_{W_{\mathcal{F}}},
\end{equation}
Thus the local mixers are eigenvectors of \(\mathcal T\) with nonzero, pairwise distinct eigenvalues  $\{4 w_v^2 : v \in V\}$. For each vertex $v \in V$, let
\begin{equation}
\label{eq:lagrange_interpolating_poly}
p_v(t) \coloneqq \prod_{u \in V \setminus \{v\}} \frac{t - 4w_u^2}{4w_v^2 - 4w_u^2} \in \mathbb{R}[t]
\end{equation}
be the Lagrange interpolating polynomial satisfying $p_v(4w_u^2) = \delta_{vu}$. Evaluating the superoperator $p_v(\mathcal T)$ on the global mixer generator isolates the corresponding local mixer, confirming its presence in the dynamical Lie algebra:
\begin{equation}
\label{eq:local_mixer_isolation}
p_v(\mathcal T)\!\left(iH_{\mathrm{CX}}\big|_{W_{\mathcal{F}}}\right) = iH_{\mathrm{CX},v}\big|_{W_{\mathcal{F}}} \in \mathfrak{h} \qquad \text{for all } v \in V.
\end{equation}

For every $v\in V$,
\[
[iH_{P,w},iH_{\mathrm{CX},v}]
=-2iw_vY_vP_v,
\qquad
\frac{1}{4w_v}
[iH_{\mathrm{CX},v},[iH_{P,w},iH_{\mathrm{CX},v}]]
=iZ_vP_v\in\mathfrak h.
\]
The localization and matrix-unit construction in the proof of
Theorem~\ref{thm:ExtendedFreeLieAlgFullUnitary} therefore applies,
yielding
\[
\mathfrak{su}(W_{\mathcal F})\subseteq\mathfrak h.
\]

Pairing each independent set containing $v$ with the set
obtained by deleting $v$ cancels their trace contributions;
every unpaired set contributes $+1$, and the singleton of
any neighbor of $v$ is unpaired. Thus
$\operatorname{Tr}_{W_{\mathcal F}}(Z_v)>0$. Hence, positivity of the weights implies
$\tau:=\operatorname{Tr}_{W_{\mathcal F}}(H_{P,w})>0$.
Since
\[
iH_{P,w}-\frac{i\tau}{|\mathcal F|}I
\in\mathfrak{su}(W_{\mathcal F})\subseteq\mathfrak h,
\]
subtraction gives $iI\in\mathfrak h$, and hence
$\mathfrak h=\mathfrak u(W_{\mathcal F})$.

For the Flip-or-Stay mixer, $\mathcal T$ annihilates its
diagonal correction. Interpolation at
$\{0\}\cup\{4w_v^2:v\in V\}$ again isolates every
$iH_{\mathrm{CX},v}$, so the same argument proves the result.
 
\end{proof}

\begin{proof}[Proof of Corollary \ref{cor:controllability}]
This assertion follows directly from Theorem~\ref{thm:ExtendedFreeLieAlgFullUnitary}, which implies that the dynamical Lie groups associated with the restricted DLAs $\mathfrak{g}_{\Gamma,\mathrm{free}}\big|_{W_\mathcal{F}}\cong\widehat{\mathfrak{g}}_{\Gamma,\mathrm{free}, \mathcal{F}} \cong \mathfrak{u}(W_\mathcal{F})$ are both the full unitary group acting on the feasible subspace, i.e., 
\begin{equation}
    G := \exp(\widehat{\mathfrak{g}}_{\Gamma,\mathrm{free}, \mathcal{F}}) \cong U(W_{\mathcal{F}}).
\end{equation}
Because the unitary group $U(W_{\mathcal{F}})$ acts transitively on the unit sphere of $W_{\mathcal{F}}$, there must exist a group element $g \in G$ that maps the initial state $\ket{h}$ exactly to the target state $\ket{t}$. 

Furthermore, since $G$ is a compact and connected Lie group, standard results in quantum control theory~\cite[Corollary~3.2.6]{dalessandro2021quantum} guarantee that any such element $g \in G$ can be decomposed and realized as a finite-depth alternating product of the exponentials of the DLA's generators. This yields the desired parameterized quantum circuit $U(\boldsymbol{\theta})$ satisfying \eqref{eq:exact_state_prep}.
\end{proof}

Before proceeding to the proof of Theorem~\ref{thm:explicit_family_orbit_separation}, we introduce the degree-resolved and aggregate basis vectors in the feasible subspace $W_{\mathcal{F}_r}$. We then establish a structural proposition identifying the exact forms of the cyclic subrepresentations generated from the computational vacuum state $\ket{0}^n$ under the Flip-or-Void and Flip-or-Stay DLAs, providing a foundational building block for establishing the results stated in the theorem.

For each $k = 1, \ldots, r-1$, define the degree-resolved superposition of singleton set states
\begin{equation}
\label{eq:sk_definition}
    s_k := \sum_{\substack{v \in V(H_r) \\ d_{H_r}(v) = k}} \ket{\{v\}},
\end{equation}
and, for $a = 1, \ldots, m$, define the component-resolved superposition of two-element set states
\begin{equation}
\label{eq:wa_definition}
    w_a := \sum_{\{v,w\} \in E(K_{a,r-a})} \ket{\{v,w\}}.
\end{equation}
We also introduce the aggregate vectors
\begin{equation}
\label{eq:u_d_w_definition}
    u := \sum_{k=1}^{r-1} s_k, \qquad d := \sum_{k=1}^{r-1} k s_k, \qquad w := \sum_{a=1}^{m} w_a,
\end{equation}
along with the corresponding candidate invariant subspaces: the four-dimensional aggregate subspace
\begin{equation}
\label{eq:aggregate_subspace_def}
    \mathcal{W}_r^A := \operatorname{span}_{\mathbb{C}} \left\{ \ket{0}^n, u, d, w \right\}
\end{equation}
and the degree-resolved subspace
\begin{equation}
\label{eq:aggregate_subspace_hat_def}
    \mathcal{W}_r^L := \operatorname{span}_{\mathbb{C}} \left\{ \ket{0}^n, s_1, \ldots, s_{r-1}, w_1, \ldots, w_m \right\}.
\end{equation}

\begin{defn}
    For a Lie algebra $\mathfrak{g} \subseteq \mathfrak{u}(W_{\mathcal{F}_r})$, we denote by
\begin{equation}
\label{eq:cyclic_subspace_def}
    W_{\mathfrak{g},\ket{0}^n} \subseteq W_{\mathcal{F}_r}
\end{equation}
the \emph{cyclic subrepresentation} generated by $\ket{0}^n$, i.e., the smallest complex $\mathfrak{g}$-invariant subspace of $W_{\mathcal{F}_r}$ containing $\ket{0}^n$.

\end{defn}

\begin{prop}
\label{prop:vacuum_cyclic_subspaces}
Let $r \ge 5$ be an odd integer, set $n := \frac{r(r-1)}{2}$, and let $W_{\mathfrak{g}_{\Gamma_r,\mathrm{std}},\ket{0}^n}, W_{\widehat{\mathfrak{g}}_{\Gamma_r,\mathrm{std}},\ket{0}^n} \subseteq W_{\mathcal{F}_r}$ be the cyclic subrepresentations generated by the computational vacuum state $\ket{0}^n$ under the action of the standard Flip-or-Void DLA $\mathfrak{g}_{\Gamma_r,\mathrm{std}}$ and Flip-or-Stay DLA $\widehat{\mathfrak{g}}_{\Gamma_r,\mathrm{std}}$, respectively. Then,
\begin{equation}
\label{eq:vacuum_cyclic_subspace_identifications}
\begin{aligned}
    W_{\mathfrak{g}_{\Gamma_r,\mathrm{std}},\ket{0}^n} &= \mathcal{W}_r^A, &\quad \dim_{\mathbb{C}}(W_{\mathfrak{g}_{\Gamma_r,\mathrm{std}},\ket{0}^n}) &= 4, \\
    W_{\widehat{\mathfrak{g}}_{\Gamma_r,\mathrm{std}},\ket{0}^n} &= \mathcal{W}_r^L, &\quad \dim_{\mathbb{C}}(W_{\widehat{\mathfrak{g}}_{\Gamma_r,\mathrm{std}},\ket{0}^n}) &= \frac{3r-1}{2}.
\end{aligned}
\end{equation}
\end{prop}

\begin{proof}
Recall that the matrix representing the restriction of the Flip-or-Void mixer $H_{\mathrm{CX}}$ to the feasible subspace $W_{\mathcal{F}_{\Gamma_r}}$ in the computational basis coincides with the adjacency matrix $A_{\mathcal{R}(\Gamma_r)}$ of the reconfiguration graph (see \eqref{eq:FlipOrStayAsLaplacian}).

Hence, the corresponding restricted standard dynamical Lie algebra is given by
\begin{equation}
\label{eq:family_std_dla}
    \mathfrak{g}_{\Gamma_r,\mathrm{std}} = \left\langle i A_{\mathcal{R}(\Gamma_r)}, \, i H_P \big|_{W_{\mathcal{F}_{\Gamma_r}}} \right\rangle_{\operatorname{Lie}}.
\end{equation}

The action of $A_{\mathcal{R}(\Gamma_r)}$ on the computational vacuum $\ket{0}^n$ and the aggregate vectors $\{u, d, w\}$ defined in \eqref{eq:u_d_w_definition} is given explicitly by
\begin{equation}
\label{eq:family_adjacency_actions}
\begin{aligned}
    A_{\mathcal{R}(\Gamma_r)} \ket{0}^n &= u, \\
    A_{\mathcal{R}(\Gamma_r)} u &= n \ket{0}^n + 2w, \\
    A_{\mathcal{R}(\Gamma_r)} w &= d, \\
    A_{\mathcal{R}(\Gamma_r)} d &= 2|E(H_r)| \ket{0}^n + rw.
\end{aligned}
\end{equation}
Consequently, $\mathcal W_r^A$ is stable under the action of
$A_{\mathcal R(\Gamma_r)}$. Since $H_P\ket{x}=(n-2|x|)\ket{x}$, the spanning vectors $\ket{0}^n$, $u$, $d$, and $w$ are eigenvectors of $H_P$ with eigenvalues $n$, $n-2$, $n-2$, and $n-4$, respectively.
Thus, $\mathcal W_r^A$ is invariant under both generators of $\mathfrak g_{\Gamma_r,\mathrm{std}}$ and contains $\ket{0}^n$. By minimality of the vacuum-generated cyclic subspace,
\begin{equation}
\label{eq:family_adjacency_cyclic_first_containment}
    W_{\mathfrak g_{\Gamma_r,\mathrm{std}},\ket{0}^n}
    \subseteq \mathcal W_r^A.
\end{equation}

Conversely, expressing $u$, $w$, and $d$ via successive applications of $A_{\mathcal{R}(\Gamma_r)}$ yields
\begin{equation}
\label{eq:family_reverse_generation}
    u = A_{\mathcal{R}(\Gamma_r)} (\ket{0}^n), \qquad 
    w = \frac{1}{2} \left( A_{\mathcal{R}(\Gamma_r)} (u) - n \ket{0}^n \right), \qquad 
    d = A_{\mathcal{R}(\Gamma_r)} (w).
\end{equation}
As $\ket{0}^n \in W_{\mathfrak{g}_{\Gamma_r,\mathrm{std}},\ket{0}^n}$ and $W_{\mathfrak{g}_{\Gamma_r,\mathrm{std}},\ket{0}^n}$ is closed under $A_{\mathcal{R}(\Gamma_r)}$, all three vectors $u, w, d$ necessarily belong to the subspace $W_{\mathfrak{g}_{\Gamma_r,\mathrm{std}},\ket{0}^n}$. Consequently,
\begin{equation}
\label{eq:family_adjacency_cyclic_space}
    W_{\mathfrak{g}_{\Gamma_r,\mathrm{std}},\ket{0}^n} = \mathcal{W}_r^A.
\end{equation}

Finally, $\ket{0}^n$ (weight $0$) and $w$ (weight $2$) belong to orthogonal weight sectors and are mutually orthogonal to the weight-$1$ sector containing $u$ and $d$. Within the weight-$1$ sector, $u$ and $d$ expand as
\begin{equation}
\label{eq:weight1_vectors_expansion}
    u = \sum_{k=1}^{r-1} s_k, \qquad d = \sum_{k=1}^{r-1} k s_k.
\end{equation}
The vectors $s_1, \ldots, s_{r-1}$ are non-zero and have mutually disjoint supports. As the coefficient vectors $(1, \ldots, 1)^T$ and $(1, 2, \ldots, r-1)^T$ are linearly independent, $u$ and $d$ are linearly independent as well. Therefore, $\{\ket{0}^n, u, d, w\}$ forms a basis for $\mathcal{W}_r^A$, establishing
\begin{equation}
\label{eq:family_adjacency_cyclic_dimension}
    \dim_{\mathbb{C}} \left( W_{\mathfrak{g}_{\Gamma_r,\mathrm{std}},\ket{0}^n} \right) = \dim_{\mathbb{C}} \left( \mathcal{W}_r^A \right) = 4.
\end{equation}

Next, we determine the cyclic subspace $W_{\widehat{\mathfrak{g}}_{\Gamma,\mathrm{std}},\ket{0}^n}$ generated by the computational vacuum state $\ket{0}^n$ under the standard dynamical Lie algebra employing the Flip-or-Stay mixer.
We establish the element
\begin{equation}
    \label{eq:Std_DLA_containment_Gamma_r}
    iA_{\mathcal R(\Gamma_r)}
    =
    -\frac14
    \left[
        iH_P,
        \left[
            iH_P,
            i\widehat H_{CX}
        \right]
    \right]
    \in\widehat{\mathfrak{g}}_{\Gamma,\mathrm{std}}.
\end{equation}
inside the DLA.

Consequently, subtracting $i A_{\mathcal{R}(\Gamma_r)}$ from $i \widehat{H}_{\mathrm{CX}}$ isolates the diagonal degree generator:
\begin{equation}
\label{eq:diagonal_degree_generator_def}
    \mathcal{D}_r : = i \widehat{H}_{\mathrm{CX}} - i A_{\mathcal{R}(\Gamma_r)} \in \widehat{\mathfrak{g}}_{\Gamma_r,\mathrm{std}}.
\end{equation}

The operator $\mathcal{D}_r$ acts diagonally on each  singleton state $s_k$ via
\begin{equation}
\label{eq:Dr_action_sk}
    \mathcal{D}_r (s_k) = i(n - 1 - k) s_k, \qquad k = 1, \ldots, r-1,
\end{equation}
where the eigenvalues $\lambda_k := i(n - 1 - k)$ are pairwise distinct. Applying $A_{\mathcal{R}(\Gamma_r)}$ to the computational vacuum state yields 
\[
u=A_{\mathcal{R}(\Gamma_r)}(\ket{0}^n) = \sum_{k=1}^{r-1} s_k \in W_{\widehat{\mathfrak{g}}_{\Gamma_r,\mathrm{std}},\ket{0}^n}.
\]
Repeated action of $\mathcal{D}_r$ on $u$ generates the sequence of states
\begin{equation}
\label{eq:vandermonde_system}
    \mathcal{D}_r^t (u) = \sum_{k=1}^{r-1} \lambda_k^t s_k \in W_{\widehat{\mathfrak{g}}_{\Gamma_r,\mathrm{std}},\ket{0}^n}, \qquad t = 0, 1, \ldots, r-2.
\end{equation}
Since the associated $(r-1) \times (r-1)$ coefficient matrix $V_{tk} = (\lambda_k^t)$ is a non-singular Vandermonde matrix, the system can be inverted to isolate each singleton vector individually:
\begin{equation}
\label{eq:sk_in_cyclic_subspace}
    s_k \in W_{\widehat{\mathfrak{g}}_{\Gamma_r,\mathrm{std}},\ket{0}^n}, \qquad k = 1, \ldots, r-1.
\end{equation}

Similarly, the action of $A_{\mathcal{R}(\Gamma_r)}$ on each individual $s_k$ is given by
\begin{equation}
\label{eq:adjacency_action_sk}
    A_{\mathcal{R}(\Gamma_r)} (s_k) = (r-k)\ket{0}^n + w_{a(k)}, \qquad \text{where } a(k) := \min\{k, r-k\}.
\end{equation}

As $\ket{0}^n, s_k \in \mathcal{W}_{\widehat{\mathfrak{g}}_{\Gamma_r,\mathrm{std}},\ket{0}^n}$ and the cyclic subspace is stable under $A_{\mathcal{R}(\Gamma_r)}$, it follows directly that
\begin{equation}
\label{eq:wa_in_cyclic_subspace}
    w_a \in \mathcal{W}_{\widehat{\mathfrak{g}}_{\Gamma_r,\mathrm{std}},\ket{0}^n} \qquad \text{for all } a = 1, \dots, m.
\end{equation}
Since $\mathcal{W}_r^L = \operatorname{span}\{\ket{0}^n, s_1, \dots, s_{r-1}, w_1, \dots, w_m\}$, this establishes the forward inclusion
\begin{equation}
\label{eq:forward_subspace_containment}
    \mathcal{W}_r^L \subseteq \mathcal{W}_{\widehat{\mathfrak{g}}_{\Gamma_r,\mathrm{std}},\ket{0}^n}.
\end{equation}

Conversely, the vacuum state satisfies $\ket{0}^n \in \mathcal{W}_r^L$, and $\mathcal{W}_r^L$ is $\widehat{\mathfrak{g}}_{\Gamma_r,\mathrm{std}}$-stable (as verified by checking the explicit action of the operators $i A_{\mathcal{R}(\Gamma_r)}$, $i H_P$ and $\mathcal{D}_r$ on the orthonormal basis \eqref{eq:orthonormal_basis_WrL}). Since $\mathcal{W}_{\widehat{\mathfrak{g}}_{\Gamma_r,\mathrm{std}},\ket{0}^n}$ is defined as the minimal invariant subspace containing $\ket{0}^n$, we obtain the reverse inclusion $\mathcal{W}_{\widehat{\mathfrak{g}}_{\Gamma_r,\mathrm{std}},\ket{0}^n} \subseteq \mathcal{W}_r^L$, establishing the subspace equality
\begin{equation}
\label{eq:cyclic_subspace_equality}
    \mathcal{W}_{\widehat{\mathfrak{g}}_{\Gamma_r,\mathrm{std}},\ket{0}^n} = \mathcal{W}_r^L.
\end{equation}

Finally, it follows from \eqref{eq:sk_definition} and \eqref{eq:wa_definition} that the vectors in \eqref{eq:aggregate_subspace_hat_def} have mutually disjoint, non-empty supports and are therefore linearly independent. Therefore,
\begin{equation}
\label{eq:dimension_complex_subspace}
\dim_{\mathbb{C}} W_r^L = 1 + (r-1) + \frac{r-1}{2} = \frac{3r-1}{2}.
\end{equation}
\end{proof}

Having characterized the vacuum-generated cyclic subrepresentations in Proposition~\ref{prop:vacuum_cyclic_subspaces}, we now proceed to the proof of Theorem~\ref{thm:explicit_family_orbit_separation}.

\begin{proof}[Proof of Theorem \ref{thm:explicit_family_orbit_separation}]
We start by analyzing the structure of the restriction of the DLA $\widehat{\mathfrak{g}}_{\Gamma_r,\mathrm{std}}$ to the cyclic subspace $\mathcal{W}_{\widehat{\mathfrak{g}}_{\Gamma_r,\mathrm{std}},\ket{0}^n}=\mathcal{W}_r^L$. With respect to the orthonormal basis of $\mathcal{W}_r^L$ given by (cf.~\eqref{eq:aggregate_subspace_hat_def})
\begin{equation}
\label{eq:orthonormal_basis_WrL}
    \ket{0} := \ket{0}^n,
    \qquad
    \ket{k} := \frac{s_k}{\sqrt{r-k}},
    \qquad
    \ket{\overline{a}} := \frac{w_a}{\sqrt{a(r-a)}},
\end{equation}
where $k \in \{1, \dots, r-1\}$ and $a \in \left\{1, \dots, \frac{r-1}{2}\right\}$,
the only nonzero off-diagonal matrix elements of the reconfiguration adjacency operator $A_{\mathcal{R}(\Gamma_r)}\big|_{\mathcal{W}_r^L}$ (where $iA_{\mathcal{R}(\Gamma_r)}\big|_{\mathcal{W}_r^L} \in \widehat{\mathfrak{g}}_{\Gamma_r,\mathrm{std}}\big|_{\mathcal{W}_r^L}$ as established in \eqref{eq:Std_DLA_containment_Gamma_r}) are
\begin{equation}
\label{eq:adjacency_matrix_elements_WrL}
    \bra{0} A_{\mathcal{R}(\Gamma_r)} \ket{k}
    =
    \sqrt{r-k},
    \qquad
    \left\langle\overline{a(k)}\left| A_{\mathcal{R}(\Gamma_r)}\right| k\right\rangle
    =
    \sqrt{k},
    \qquad
    a(k) := \min\{k, r-k\}.
\end{equation}
Thus, the nonzero off-diagonal entries of this operator correspond precisely to coupling pairs of the form $(0,k)$ and $(k,\overline{a(k)})$.

As shown in \eqref{eq:Dr_action_sk}, the restriction of the degree operator $\mathcal{D}_r$ to $\mathcal{W}_r^L$ is diagonal in the basis \eqref{eq:orthonormal_basis_WrL}, with eigenvalues
\begin{equation}
\label{eq:Dr_eigenvalues_WrL}
    z_0 = 0,
    \qquad
    z_k = i(n - 1 - k),
    \qquad
    z_{\overline{a}} = i(n - 2).
\end{equation}

Consequently, the differences of the $\mathcal{D}_r$-eigenvalues across the two types of coupling pairs are

\begin{equation}
\label{eq:eigenvalue_differences_WrL}
    z_k - z_0 = i(n - 1 - k)
    \quad\text{for } \{0,k\},
    \qquad
    z_k - z_{\overline{a(k)}} = i(1 - k)
    \quad\text{for } \{k, \overline{a(k)}\}.
\end{equation}

As $k$ ranges from $1$ to $r-1$, these eigenvalue differences are strictly distinct within each family. Furthermore, the two sets of differences are disjoint because the smallest absolute value in the first set strictly exceeds the largest in the second set. Indeed, using $n = \frac{r(r-1)}{2}$, we obtain

\begin{equation}
\label{eq:frequency_disjointness_bound}
    (n - r) - (r - 2)
    =
    \frac{r(r-3)}{2} - (r - 2)
    =
    \frac{(r-1)(r-4)}{2} > 0
    \qquad \text{for } r \ge 5.
\end{equation}

Thus, all absolute eigenvalue differences, and therefore all squared transition coefficients, are pairwise distinct. 

Consequently, polynomial interpolation in the superoperator $-\operatorname{ad}_{\mathcal{D}_r}^2$, applied to the restricted adjacency operator $iA_{\mathcal{R}(\Gamma_r)}\big|_{\mathcal{W}_r^L}$, isolates the purely imaginary symmetric, skew-Hermitian generators
\begin{equation}
\label{eq:X_pq_definition}
    X_{pq} := i(E_{pq} + E_{qp})
\end{equation}
for every coupling edge $\{p,q\}$ in the basis connectivity graph of $\mathcal{W}_r^L$. Whenever the corresponding transition frequency $\omega_{pq} := z_p - z_q$ is non-zero, the complementary skew-symmetric generator is obtained via the commutator
\begin{equation}
\label{eq:Y_pq_definition}
    Y_{pq} := E_{pq} - E_{qp} = \frac{1}{i(z_p-z_q)} [\mathcal{D}_r, X_{pq}].
\end{equation}

The unique zero-frequency coupling pair is $(1, \overline{1})$, where $\omega_{1,\overline{1}} = z_1 - z_{\overline{1}} = 0$. Its missing direction $Y_{1,\overline{1}}$ is synthesized by first constructing the diagonal generator
\begin{equation}
\label{eq:diagonal_generator_isolation}
    i(E_{11} - E_{00}) = \frac{1}{2} [X_{01}, Y_{01}],
\end{equation}
and then commuting it with $X_{1,\overline{1}}$, which yields 
\[
Y_{1,\overline{1}} = -[i(E_{11} - E_{00}), X_{1,\overline{1}}]
\]
Since the coupling graph is connected, commutators along paths generate \(X_{pq}\) and \(Y_{pq}\) for every distinct pair of basis indices. Their commutators generate all traceless diagonal directions. Hence the restricted algebra contains \(\mathfrak{su}(W_r^L)\).

Furthermore, the restriction of the cost Hamiltonian $H_P$ to $\mathcal{W}_r^L$ has a non-zero trace. Indeed, since the basis of $\mathcal{W}_r^L$ consists of $1$ vacuum vector, $r-1$ singleton vectors, and $m = \frac{r-1}{2}$ pair vectors, direct evaluation yields
\begin{equation}
\label{eq:HP_trace_WrL}
    \operatorname{Tr}\!\left( H_P\big|_{\mathcal{W}_r^L} \right)
    =
    \frac{(r-1)(3r^2-r-16)}{4}.
\end{equation}

As $3r^2 - r - 16 > 0$ for all $r \ge 5$, the trace in \eqref{eq:HP_trace_WrL} does not vanish. Consequently, the imaginary scalar identity operator $i I_{\mathcal{W}_r^L}$ lies in the algebra,  establishing that the restriction of the dynamical Lie algebra $\widehat{\mathfrak{g}}_{\Gamma_r,\mathrm{std}}$ to the cyclic subspace $\mathcal{W}_r^L$ is the full unitary Lie algebra:
\begin{equation}
\label{eq:full_unitary_dla_equality}
     \widehat{\mathfrak{g}}_{\Gamma_r,\mathrm{std}} \big|_{\mathcal{W}_r^L} = \mathfrak{u}(\mathcal{W}_r^L).
\end{equation}

Exponentiating the identity \eqref{eq:full_unitary_dla_equality} shows that the restricted dynamical Lie group is 
\[
\widehat{G}_{\Gamma_r}\big|_{\mathcal{W}_r^L} \cong \mathrm{U}(\mathcal{W}_r^L)
\]
and, therefore, acts transitively on the complex unit sphere $\mathbb{S}(\mathcal{W}_r^L)$. In particular, the reachable state orbit from the vacuum state $\ket{0}^n \in \mathcal{W}_r^L$ fills the entire unit sphere:
\begin{equation}
\label{eq:laplacian_orbit_sphere}
    \widehat{G}_{\Gamma_r} \cdot \ket{0}^n = \mathbb{S}(\mathcal{W}_r^L).
\end{equation}
By contrast, the invariance of the adjacency cyclic subspace $\mathcal{W}_r^A$ with respect to the DLA $\mathfrak{g}_{\Gamma_r,\mathrm{std}}$ restricts the corresponding orbit to
\begin{equation}
\label{eq:adjacency_orbit_sphere}
    G_{\Gamma_r} \cdot \ket{0}^n \subseteq \mathbb{S}(\mathcal{W}_r^A).
\end{equation}

Recall that the real dimension of the unit sphere $\mathbb{S}(\mathcal{V})$ in a complex Hilbert space $\mathcal{V}$ is $2\dim_{\mathbb{C}}(\mathcal{V}) - 1$. Given $\dim_{\mathbb{C}}(\mathcal{W}_r^L) = 1 + (r-1) + \frac{r-1}{2} = \frac{3r-1}{2}$ and $\dim_{\mathbb{C}}(\mathcal{W}_r^A) \le 4$, we obtain the following orbit dimension scaling and bounds:
\begin{equation}
\label{eq:orbit_dimension_bounds}
    \dim_{\mathbb{R}}\!\left(G_{\Gamma_r} \cdot \ket{0}^n\right) \le 7,
    \qquad\text{whereas}\qquad
    \dim_{\mathbb{R}}\!\left(\widehat{G}_{\Gamma_r} \cdot \ket{0}^n\right) = 2\left(\frac{3r-1}{2}\right) - 1 = 3r - 2.
\end{equation}

\end{proof}

\begin{rmk}
\label{rmk:strict_orbit_containment}
The Lie algebra containment $\mathfrak{g}_{\Gamma_r,\mathrm{std}} \subseteq \widehat{\mathfrak{g}}_{\Gamma_r,\mathrm{std}}$ (established in \eqref{eq:Std_DLA_containment_Gamma_r}) directly induces the Lie group inclusion $G_{\Gamma_r} \subseteq \widehat{G}_{\Gamma_r}$, which in turn yields the reachable state orbit inclusion

\begin{equation}
\label{eq:orbit_inclusion_chain}
    G_{\Gamma_r} \cdot \ket{0}^n \subseteq \widehat{G}_{\Gamma_r} \cdot \ket{0}^n.
\end{equation}

As $7 < 3r - 2$ for all $r \ge 5$, this orbit inclusion is strictly proper, confirming that the Laplacian-based QAOA possesses strictly greater state-reachability expressivity on the computational vacuum than its adjacency-based counterpart.
\end{rmk}

Next, we rigorously verify the loss function variance formulas for deep ma-QAOA circuits proposed in Theorem~\ref{thm:BP_mitigation}.

\begin{proof}[Proof of Theorem~\ref{thm:BP_mitigation}]
Since the circuit ensemble forms an exact unitary \(2\)-design on
the associated dynamical group, its first two moments coincide with
the corresponding Haar moments. We may therefore compute the
parameter-averaged loss variance by integrating over
\(G\cong\mathrm{U}(W_{\mathcal F})\).

For convenience, set
\[
    H:=H_P\big|_{W_{\mathcal F}}.
\]
Since \(iH\in\mathfrak g\big|_{W_{\mathcal F}}\),
Theorem~1 of \cite{RBSKMLC} applies and gives
\begin{equation}
\label{eq:VarAndDim1}
    \operatorname{Var}_{G}
    \bigl[\ell_U(\rho,H)\bigr]
    =
    \frac{
        \mathcal P_{\mathfrak{su}(W_{\mathcal F})}(\rho)\,
        \mathcal P_{\mathfrak{su}(W_{\mathcal F})}(H)
    }{
        \dim(\mathfrak{su}(W_{\mathcal F}))
    }
    =
    \frac{
        \mathcal P_{\mathfrak{su}(W_{\mathcal F})}(\rho)\,
        \mathcal P_{\mathfrak{su}(W_{\mathcal F})}(H)
    }{
        d^2-1
    }.
\end{equation}

We identify the Lie algebra $\mathfrak{su}(W_{\mathcal{F}})$, up to a factor of $i$, with the real vector space of traceless Hermitian operators on $W_{\mathcal{F}}$. Accordingly, for any Hermitian operator $A \in \operatorname{End}(W_{\mathcal{F}})$, we define its traceless component $A_0$ and generalized purity $\mathcal{P}_{\mathfrak{su}(W_{\mathcal{F}})}(A)$ by
\[
    A_0 := A - \frac{\operatorname{Tr}(A)}{d} I_{W_{\mathcal{F}}}, \qquad \mathcal{P}_{\mathfrak{su}(W_{\mathcal{F}})}(A) := \operatorname{Tr}(A_0^2).
\]
For any pure state $\rho = \ket{\psi}\bra{\psi}$ corresponding to a unit vector $\ket{\psi} \in W_{\mathcal{F}}$, $\rho$ is a rank-one projection satisfying $\rho^2 = \rho$ and $\operatorname{Tr}(\rho) = 1$. It follows directly that
\begin{equation}
\label{eq:RhoProjection}
    \mathcal{P}_{\mathfrak{su}(W_{\mathcal{F}})}(\rho) = \operatorname{Tr}\left[ \left( \rho - \frac{1}{d} I_{W_{\mathcal{F}}} \right)^2 \right] = 1 - \frac{1}{d}.
\end{equation}
Substituting identity~\eqref{eq:RhoProjection} into~\eqref{eq:VarAndDim1} and using the simplification
\[
    \frac{1 - \frac{1}{d}}{d^2 - 1} = \frac{1}{d(d+1)},
\]
we obtain the simplified variance expression:
\begin{equation}
\label{eq:HaarVarianceSimplified}
    \operatorname{Var}_{G} \bigl[ \ell_U(\rho, H) \bigr] = \frac{\operatorname{Tr}(H_0^2)}{d(d+1)}.
\end{equation}

Recall that
\(
    H_P=\sum_{v\in V}Z_v.
\)
For every \(x\in\mathcal F\), we therefore have
\[
    H\ket{x}
    =
    \bigl(n-2F(x)\bigr)\ket{x}
    =
    \bigl(n-2|x|\bigr)\ket{x}.
\]
Let $S \sim \operatorname{Unif}(\mathcal{F})$ be a discrete random variable uniformly distributed over the feasible configuration set $\mathcal{F}$, so that $|S| = F(S)$ represents the cardinality of a uniformly chosen independent set.

The  projection of $H_P$ onto the traceless Lie algebra $\mathfrak{su}(W_{\mathcal{F}})$ can be evaluated via pairwise differences:
\begin{align}
\label{eq:ProblemHamiltonianProjectionVariance}
    \mathcal{P}_{\mathfrak{su}(W_{\mathcal{F}})}(H) 
    &= \frac{1}{2d} \sum_{x, y \in \mathcal{F}} (h_x - h_y)^2 \nonumber \\
    &= \frac{2}{d} \sum_{x, y \in \mathcal{F}} \bigl(F(x) - F(y)\bigr)^2 \nonumber \\
    &= \frac{2}{d} \left( 2d^2 \operatorname{Var}(|S|) \right) = 4d \operatorname{Var}(|S|).
\end{align}

Combining Equation~\eqref{eq:HaarVarianceSimplified} with Equation~\eqref{eq:ProblemHamiltonianProjectionVariance} yields the exact variance identity
\[
    \operatorname{Var}_{G}\bigl[ \ell_U(\rho, H_P) \bigr] = \frac{4\operatorname{Var}(|S|)}{d+1},
\]
completing the proof of \eqref{eq:HaarVarianceExactIdentity}.

It remains to bound the classical variance of $|S|$. The pairwise variance identity gives
\begin{equation}
\label{eq:IndependentSetPairwiseVariance}
    \operatorname{Var}(|S|) = \frac{1}{2d^2} \sum_{L, T \in \mathcal{F}} \bigl(|L| - |T|\bigr)^2.
\end{equation}

To establish a lower bound, let $m := \alpha(\Gamma)$ denote the independence number of $\Gamma$. Retaining the terms with \(T=\varnothing\) or \(L=\varnothing\) in \eqref{eq:IndependentSetPairwiseVariance} yields 
\[
\operatorname{Var}(|S|) \ge \frac{1}{d^2} \sum_{A \in \mathcal{F} \setminus \{\varnothing\}} |A|^2.
\]

Since every subset of an independent set is itself independent, we may select a chain of independent sets $A_1, \ldots, A_m \in \mathcal{F}$ such that $|A_j| = j$ for each $j \in \{1, \ldots, m\}$. The remaining $d - m - 1$ nonempty independent sets in $\mathcal{F}$ each have cardinality at least one. Consequently,
\begin{align}
    \operatorname{Var}(|S|) 
    &\ge \frac{1}{d^2} \left( \sum_{j=1}^{m} j^2 + d - m - 1 \right) \notag \\
    &= \frac{1}{d^2} \left( \frac{m(m+1)(2m+1)}{6} + d - m - 1 \right).
\label{eq:IndependentSetVarianceLowerBound}
\end{align}

For the upper bound, notice that $0 \le |S| \le m$. Setting $\mu := \mathbb{E}[|S|]$, the bound $|S|^2 \le m|S|$ implies
\begin{equation}
\label{eq:IndependentSetVarianceUpperBound}
    \operatorname{Var}(|S|) = \mathbb{E}(|S|^2) - \mu^2 \le m\mu - \mu^2 = \frac{m^2}{4} - \left(\mu - \frac{m}{2}\right)^2 \le \frac{m^2}{4}.
\end{equation}

Substituting the lower bound \eqref{eq:IndependentSetVarianceLowerBound} (using $m \ge 1 \implies \frac{m(m+1)(2m+1)}{6} - m \ge 0$) and upper bound \eqref{eq:IndependentSetVarianceUpperBound} into the exact variance identity \eqref{eq:HaarVarianceExactIdentity} yields
\[
    \frac{4}{d^2(d+1)} \left( \frac{m(m+1)(2m+1)}{6} + d - m - 1 \right) \le \operatorname{Var}_{G}\bigl[\ell_U(\rho, H)\bigr] \le \frac{m^2}{d+1}.
\]
For $d \ge 2$, these bounds establish the asymptotic scaling
\[
    \operatorname{Var}_{G}\bigl[\ell_U(\rho, H)\bigr] \in \Omega\left(\frac{1}{d^{2}}\right) \cap\mathcal{O}\left(\frac{m^2}{d}\right),
\]
which proves \eqref{eq:HaarVarianceAsymptoticScaling}.
\end{proof}

We conclude with the proof of Theorem~\ref{thm:equal_standard_orbits}, providing formal reachability and state-preparation guarantees for the standard Flip-or-Void and Flip-or-Stay ans\"atze.

\begin{proof}[Proof of Theorem \ref{thm:equal_standard_orbits}]
For every feasible string \(x\in\mathcal F\),
\begin{equation}
\label{eq:HP_weight_eigenvalue_orbit_proof}
    H_P\ket{x}
    =
    \bigl(n-2|x|\bigr)\ket{x}.
\end{equation}
Therefore, \(\ket{0}^n\) is the unique highest-energy eigenstate of
\(H_P\big|_{W_{\mathcal F}}\). Equivalently,
\(\ket{\zeta_\Gamma}\), \(\ket{\xi_{\mathcal F}}\), and
\(\ket{0}^n\) are the unique ground states of
\[
    -H_{CX}\big|_{W_{\mathcal F}},
    \qquad
    -\widehat H_{CX}\big|_{W_{\mathcal F}},
    \qquad
    -H_P\big|_{W_{\mathcal F}},
\]
respectively.

We first consider the Flip-or-Void mixer. Applying the QAOA
convergence theorem \cite[Theorem~12]{BKZS} with mixer \(H_{CX}\),
phase separator \(H_P\), and initial state
\(\ket{\zeta_\Gamma}\), we obtain a sequence of finite-depth QAOA
unitaries \(U_p\in G_{\Gamma,\mathrm{std}}\) such that
\begin{equation}
\label{eq:flip_void_energy_convergence_to_vacuum}
    \lim_{p\to\infty}
    \bra{\zeta_\Gamma}
        U_p^\dagger H_PU_p
    \ket{\zeta_\Gamma}
    =
    n.
\end{equation}

Expressing $U_p \ket{\zeta_\Gamma}$ in the feasible computational basis as
\begin{equation}
    U_p \ket{\zeta_\Gamma} = \sum_{x \in \mathcal{F}} a_x^{(p)} \ket{x},
\end{equation}
equation \eqref{eq:HP_weight_eigenvalue_orbit_proof} yields
\begin{equation}
\label{eq:energy_controls_vacuum_overlap}
    n - \bra{\zeta_\Gamma} U_p^\dagger H_P U_p \ket{\zeta_\Gamma}
    = 2 \sum_{x \in \mathcal{F}} |x| \left| a_x^{(p)} \right|^2
    \ge 2 \left( 1 - \left| a_{0^n}^{(p)} \right|^2 \right).
\end{equation}

Consequently, \eqref{eq:flip_void_energy_convergence_to_vacuum} implies that along this sequence of circuits,

\begin{equation}
    \lim_{p \to \infty} \left| \bra{0^n} U_{p} \ket{\zeta_\Gamma} \right|^2 = 1,
\end{equation}

or equivalently,

\begin{equation}
\label{eq:vacuum_in_flip_void_orbit_closure}
    [\ket{0}^n] \in \overline{G_{\Gamma,\mathrm{std}} \cdot [\ket{\zeta_\Gamma}]},
\end{equation}

where $[v] \in \mathbb{P}(W)$ denotes the ray associated with a non-zero vector $v \in W$ in the projective Hilbert space $\mathbb{P}(W) = (W \setminus \{0\}) / \mathbb{C}^*$.

We now show that the projective orbit appearing in
\eqref{eq:vacuum_in_flip_void_orbit_closure} is closed. Let
\[
    \mathcal K_\Gamma
    :=
    W_{\mathfrak g_{\Gamma,\mathrm{std}},
      \ket{\zeta_\Gamma}}
\]
be the cyclic subspace generated by \(\ket{\zeta_\Gamma}\). The
representation of \(\mathfrak g_{\Gamma,\mathrm{std}}\) on
\(\mathcal K_\Gamma\) is irreducible. Indeed, if
\(\mathcal K_\Gamma\) admitted a nontrivial orthogonal decomposition into invariant subspaces, then the corresponding orthogonal projections would commute with \(H_{CX}\). Applying these projections to \(\ket{\zeta_\Gamma}\) would produce mutually orthogonal eigenvectors belonging to the largest eigenvalue of \(H_{CX}\). Since that eigenvalue is simple, \(\ket{\zeta_\Gamma}\) can have a nonzero component in only one invariant summand. Its cyclicity then forces all other summands to vanish.

By Schur's lemma, the center of the restricted dynamical Lie algebra $\mathfrak{g}_{\Gamma,\mathrm{std},\mathcal{F}}$ acts on $\mathcal{K}_\Gamma$ via scalar multiples of the identity and thus acts trivially upon projectivization. As  $\mathfrak{g}_{\Gamma,\mathrm{std},\mathcal{F}}$ is reductive, the action on $\mathbb{P}(\mathcal{K}_\Gamma)$ factors through a semisimple Lie algebra whose associated connected Lie group is compact~\cite[Proposition~7.9]{Knapp}. As a result, the corresponding group orbit in projective space is compact, and hence closed.

Therefore, Eq.~\eqref{eq:vacuum_in_flip_void_orbit_closure} yields
\begin{equation}
\label{eq:vacuum_in_flip_void_orbit}
    [\ket{0}^{\otimes n}] \in G_{\Gamma,\mathrm{std}} \cdot [\ket{\zeta_\Gamma}].
\end{equation}

By orbit symmetry under group inversion ($g \mapsto g^{-1}$), \eqref{eq:vacuum_in_flip_void_orbit} is equivalent to
\begin{equation}
\label{eq:projective_orbit_containment}
    [\ket{\zeta_\Gamma}] \in G_{\Gamma,\mathrm{std}} \cdot [\ket{0}^n],
\end{equation}
which implies the existence of an element $g \in G_{\Gamma,\mathrm{std}}$ and a global phase $\phi \in \mathbb{R}$ such that
\begin{equation}
\label{eq:orbit_action_with_phase}
    g \ket{0}^n = e^{i\phi} \ket{\zeta_\Gamma}.
\end{equation}

Since $i H_P \in \mathfrak{g}_{\Gamma,\mathrm{std}}$ and $H_P \ket{0}^n = n \ket{0}^n$, the global phase can be absorbed directly into the group element. Defining
\begin{equation}
\label{eq:phase_adjusted_group_element}
    g' := g e^{-it H_P} \in G_{\Gamma,\mathrm{std}} \text{ with } t = \frac{\phi}{n},
\end{equation}
we obtain exact state preparation:
\begin{equation}
\label{eq:exact_state_preparation}
    g' \ket{0}^n = g \left( e^{-i t n} \ket{0}^n \right) = e^{-i\phi} g \ket{0}^n = \ket{\zeta_\Gamma}.
\end{equation}

The argument for the Flip-or-Stay mixer proceeds identically upon replacing $H_{\mathrm{CX}}$ and $\ket{\zeta_\Gamma}$ with $\widehat{H}_{\mathrm{CX}}$ and $\ket{\xi_{\mathcal{F}}}$, respectively.

\end{proof}

The finite-depth exact-preparation result in
Theorem~\ref{thm:qaoa_perron_frobenius} is a direct consequence of
the fact that each of the respective states
\(\ket{\zeta_\Gamma}\) and \(\ket{\xi_{\mathcal F}}\) can, up to a
global phase, be obtained from \(\ket{0}^n\) through the action of an
element of the corresponding connected analytic dynamical Lie group.
Since these states are the ground states of the respective mixer
Hamiltonians, the main convergence result of \cite{BKZS} implies that
the closure of each corresponding projective orbit intersects the
projectivization of the problem ground-state subspace:
\[
    \overline{
        G_{\Gamma,\mathrm{std}}
        \cdot[\ket{\zeta_\Gamma}]
    }
    \cap
    \mathbb P(W_{P,\min})
    \neq\varnothing,
    \qquad
    \overline{
        \widehat G_{\Gamma,\mathrm{std}}
        \cdot[\ket{\xi_{\mathcal F}}]
    }
    \cap
    \mathbb P(W_{P,\min})
    \neq\varnothing.
\]
The closedness of these projective orbits, established in the proof
of Theorem~\ref{thm:equal_standard_orbits}, promotes these orbit-closure
statements to actual orbit intersections. Every element of either
dynamical group is a finite product of evolutions generated by the
corresponding mixer and problem Hamiltonians and can therefore be
written, after inserting zero parameters if necessary, as a
finite-depth QAOA circuit. Finally,
Theorem~\ref{thm:equal_standard_orbits} shows that the orbit of the
appropriate mixer ground state coincides with the orbit of
\(\ket{0}^n\), allowing the same finite-depth exact-preparation
conclusion for both initial states. This establishes
Theorem~\ref{thm:qaoa_perron_frobenius}.

\end{document}